\documentclass[11pt,letterpaper]{article}

\usepackage[margin=1in]{geometry}

\usepackage[utf8]{inputenc}
\usepackage[T1]{fontenc}
\usepackage{lmodern}

\usepackage{amsmath,amssymb, amsfonts, amsthm, mathtools, bm} 
\usepackage{empheq}
\usepackage{thm-restate}
\newtheorem{theorem}{Theorem}
\newtheorem{definition}[theorem]{Definition}
\newtheorem{corollary}[theorem]{Corollary}

\newtheorem{lemma}[theorem]{Lemma}

\newtheoremstyle{note}{0.5cm}{0.5cm}{\upshape}{}
   {\bfseries}{}{\newline}{#1\, #2\, #3}
\theoremstyle{note}
\newtheorem{remark}[theorem]{Remark}

\usepackage{tabularx}
\usepackage{enumerate}
\usepackage{url}
\usepackage{graphicx}
\graphicspath{{../figures/}}

\usepackage{calc}
\usepackage{tikz}
\usetikzlibrary{calc}
\usetikzlibrary{decorations.markings}
\usetikzlibrary{arrows,calc,decorations.pathmorphing}
\usetikzlibrary{decorations.pathreplacing,shapes.misc}
\usetikzlibrary{3d}
\usetikzlibrary{positioning}

\usepackage{xcolor}
\definecolor{orange}{rgb}{1.0,0.5,0}
\definecolor{violet}{rgb}{0.6,0,0.8}
\definecolor{darkgreen}{rgb}{0,0.5,0}
\definecolor{verydarkgreen}{rgb}{0,0.3,0}
\definecolor{darkblue}{rgb}{0,0,0.6}
\definecolor{darkred}{rgb}{0.75,0,0}
\definecolor{grey}{rgb}{0.35,0.35,0.35}

\usepackage{comment}
\usepackage[textsize=footnotesize, color=blue!30!white]{todonotes}

\newcommand{\opt}{\mbox{\scriptsize\rm OPT}}
\newcommand{\lp}{\mbox{\scriptsize\rm LP}}
\newcommand{\val}{\textnormal{value}}

\def\cupp{\stackrel{.}{\cup}}
\renewcommand{\epsilon}{\varepsilon}
\renewcommand{\phi}{\varphi}

\newcommand{\sfrac}[2]{{\textstyle\frac{#1}{#2}}}

\def\Ascr{\mathcal{A}}

\def\Cscr{\mathcal{C}}

\def\Iscr{\mathcal{I}}
\def\Lscr{\mathcal{L}}

\def\Pscr{\mathcal{P}}

\definecolor{orange}{rgb}{1.0,0.5,0}
\definecolor{violet}{rgb}{0.6,0,0.8}
\definecolor{darkgreen}{rgb}{0,0.5,0}
\definecolor{darkblue}{rgb}{0,0,0.5}

\usepackage[%
 pdfpagelabels,
 hyperindex,
 bookmarksopen=true,
 bookmarksopenlevel=1,
]{hyperref}

\hypersetup{
    colorlinks=true,
    linkcolor=darkblue,
    citecolor = darkgreen,
    urlcolor=darkblue
}

\usepackage[vlined,ruled,algo2e]{algorithm2e}
\SetAlCapHSkip{0em}
\SetCustomAlgoRuledWidth{\textwidth}

\usepackage{titling}
\thanksmarkseries{arabic}

\title{An improved approximation algorithm for ATSP\thanks{An extended abstract of this paper appeared in the proceedings of STOC 2020.}}
\date{}

\author{
Vera Traub\thanks{
Department of Mathematics, ETH Zurich, Zurich, Switzerland.
Email: \href{mailto:vera.traub@ifor.math.ethz.ch}%
{vera.traub@ifor.math.ethz.ch}.
Supported by Swiss National Science Foundation grant 200021\_184622.
}
\and
Jens Vygen\thanks{
Research Institute for Discrete Mathematics and Hausdorff Center for Mathematics, University of Bonn, Bonn, Germany.
Email: \href{mailto:vygen@or.uni-bonn.de}%
{vygen@or.uni-bonn.de}.
}
}

\begin{document}
\maketitle

\begin{abstract}
 We revisit the constant-factor approximation algorithm for the asymmetric traveling salesman problem
 by Svensson, Tarnawski, and V\'egh~\cite{SveTV18}.
 We improve on each part of this algorithm.
 We avoid the reduction to irreducible instances and thus obtain a simpler and much better reduction to vertebrate pairs.
 We also show that a slight variant of their algorithm for vertebrate pairs has a much smaller approximation ratio.
 Overall we improve the approximation ratio from $506$ to $22+\epsilon$ for any $\epsilon > 0$.
 This also improves the upper bound on the integrality ratio from $319$ to $22$.
\end{abstract}

\section{Introduction}
The asymmetric traveling salesman problem (ATSP) is one of the most fundamental and challenging combinatorial optimization problems.
Given a finite set of cities with pairwise non-negative distances, we ask for a shortest tour that visits all cities and returns to the starting point.

The first non-trivial approximation algorithm was due to Frieze, Galbiati, and Maffioli \cite{FriGM82}.
Their $\log_2(n)$-approximation ratio, where $n$ is the number of cities, was improved to $0.99\log_2(n)$ by Bl\"aser \cite{Bla03}, 
to $0.842 \log_2(n)$ by Kaplan, Lewenstein, Shafrir, and Sviridenko \cite{KapXX05}, and to $\sfrac{2}{3} \log_2(n)$ by Feige and Singh
\cite{FeiS07}.
Then a $O(\log(n)/\log(\log(n)))$-approximation algorithm was discovered by Asadpour, Goemans, M\k{a}dry, Oveis Gharan, and Saberi \cite{AsaXX17}, and this inspired further work on the traveling salesman problem.
Major progress towards a constant-factor approximation algorithm was made by Svensson~\cite{Sve15}:
he devised such an algorithm for the special case in which the distances are given by an unweighted digraph.
This was extended to two different edge weights by Svensson, Tarnawski, and V\'egh~\cite{SveTV16}.

In a recent breakthrough, Svensson, Tarnawski, and V\'egh~\cite{SveTV18} devised the first constant-factor approximation algorithm 
for the general ATSP.
In their STOC~2018 paper, they showed an approximation ratio of $5500$.
Later they optimized their analysis and obtained an approximation ratio of $506$.

Since this algorithm is analyzed with respect to the natural linear programming relaxation, it also yields a constant upper bound on the integrality ratio.
In fact, Svensson, Tarnawski, and V\'egh showed an upper bound of $319$ on the integrality ratio, but their algorithm that computes such a solution does not have polynomial running time.
Before~\cite{SveTV18}, the best known upper bound was $(\log(\log(n)))^{O(1)}$~\cite{AnaO15}.
The strongest known lower bound on the integrality ratio is $2$~\cite{ChaGK06}.

We describe a polynomial-time algorithm that computes a tour of length at most $22+\epsilon$ times the LP value for any given ATSP instance.
Hence, the integrality ratio is at most $22$.
Via the reductions of~\cite{FeiS07} and \cite{KohTV19}, our result also implies stronger upper bounds for the path version, where the start and end of the tour are given and distinct.

\section{Outline}

An instance of ATSP can be described as a strongly connected digraph $G=(V,E)$ and a cost (or length) function $c : E \to \mathbb{R}_{\ge 0}$.
We look for a minimum-cost closed walk in $G$ that visits every vertex at least once. Such a closed walk may use (and then has to pay) edges several times.
A \emph{tour} is a multi-set $F$ of edges such that $(V,F)$ is connected and Eulerian, i.e.\ every vertex has the same number of entering and leaving edges.
Since such a graph admits an Eulerian walk (a closed walk that visits every vertex at least once and every vertex exactly once), an equivalent formulation of ATSP asks for a tour $F$ with $c(F)$ minimum.

The algorithm by Svensson, Tarnawski, and V\'egh~\cite{SveTV18} proceeds through a sequence of reductions,
which we follow with some modifications (see Figure~\ref{fig:reduction_steps}).
First they show that it suffices to consider so-called laminarly-weighted instances.
We strengthen this reduction to what we call \emph{strongly laminar} instances (Section~\ref{sect:strongly_laminar}).
In contrast to the following reductions this causes no loss in the approximation ratio.
In a strongly laminar instance the cost of an edge $e$ is given by the cost of entering or leaving sets in a laminar family $\Lscr$
each of whose elements induces a strongly connected subgraph.
More precisely,
\[ c(e) = \sum_{L \in\Lscr : e \in \delta(L)} y_L \]
for some positive weights $y_L$ ($L\in \Lscr$) and all $e\in E$.

\begin{figure}
\begin{center}
 \begin{tikzpicture}[xscale=0.95,yscale=0.9]
 \begin{scope}[align=left, anchor=north west,text height=4ex, inner sep=8pt, outer sep =3pt]
 \node[draw] (a) at (-2,2) {ATSP\\ };
 \node[draw] (b) at (0.6,2) {laminarly- \\ weighted ATSP};
 \node[draw] (c) at (4.7,2) {irreducible \\ instances};
 \node[draw] (d) at (7.9,2) {vertebrate \\ pairs};
 \node[draw] (e) at (11.2,2) {Subtour Partition \\ Cover};
 \end{scope}
 \draw[->, thick, >=latex] (a) --(b);
 \draw[->, thick, >=latex] (b) -- (c);
 \draw[->, thick, >=latex] (c) -- (d);
 \draw[->, thick, >=latex] (d) -- (e);
 
 \begin{scope}[align=left, anchor=north west,text height=4ex, inner sep=8pt, outer sep =3pt]
 \node[draw] (a') at (-2,0) {ATSP\\ };
 \node[draw] (b') at (0.6,0) { strongly \\ laminar ATSP \hspace*{0.5mm}};
 \node[draw] (d') at (7.9,0) {vertebrate \\ pairs};
 \node[draw] (e') at (11.2,0) {Subtour \phantom{Partition} \\ Cover \qquad\quad {\tiny Sec.\,5}};
 \end{scope}
 \draw[->, thick, >=latex] (a') -- node[below] {\tiny Sec.\,3} (b');
 \draw[->, thick, >=latex] (b') --  node[below] {\tiny Sec.\,4} (d');
 \draw[->, thick, >=latex] (d') -- node[below] {\tiny Sec.\,6} (e');
 \end{tikzpicture}
\end{center}
\caption{The sequence of reductions of Svensson, Tarnawski, and V\'egh~\cite{SveTV18} (top)
and our algorithm (bottom).
\label{fig:reduction_steps}
}
\end{figure}
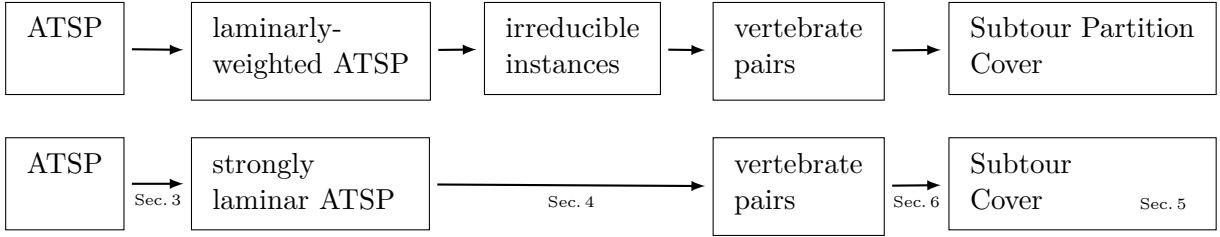

In Section~\ref{sect:reduction_veterbrate_pair} we reduce strongly laminar instances to even more structured instances 
called \emph{vertebrate pairs}.
In a vertebrate pair we already have a given subtour, called \emph{backbone}, that visits not necessarily all vertices but all non-singleton
elements of the laminar family $\Lscr$. 
In contrast to the reduction to strongly laminar instances, the reduction to vertebrate pairs causes some loss in the approximation ratio.
While Svensson, Tarnawski, and V\'egh also reduce to vertebrate pairs, 
they first reduce to what they call an \emph{irreducible instance} as an intermediate step before reducing to vertebrate pairs.
We show that this intermediate step is not necessary. This leads to a simpler algorithm.
Moreover, the loss in the approximation ratio in this step is much smaller.
In fact, a significant part of the improvement of the overall approximation ratio is due to our new reduction to vertebrate pairs.

Finally, in Section~\ref{sect:subtour_cover} and Section~\ref{sect:solve_vertebrate_pairs} 
we explain how to compute good solutions for vertebrate pairs.
The main algorithmic framework, essentially due to Svensson~\cite{Sve15}, follows
on a very high level the cycle cover approach by Frieze, Galbiati and Maffioli~\cite{FriGM82}. 
It maintains an Eulerian subgraph $H$ which intially consists of the backbone only.
In each iteration it computes an Eulerian set $F$ of edges that connects every connected component of $H$, except possibly the backbone, to 
another connected component.
However, in order to achieve a constant-factor approximation for ATSP we need additional properties and will not always add all edges of $F$ to $H$.

In Section~\ref{sect:subtour_cover} we explain a sub-routine that computes the edge set $F$ in every iteration of Svensson's algorithm.
The problem solved by the sub-routine, which we call \emph{Subtour Cover},
can be viewed as the analogue of the cycle cover problem that is solved in every iteration of the $\log_2(n)$-approximation 
algorithm by Frieze, Galbiati and Maffioli~\cite{FriGM82}.
It is very similar to what Svensson, Tarnawski, and V\'egh call \emph{Subtour Partition Cover} and \emph{Eulerian Partition Cover} 
and Svensson~\cite{Sve15} calls \emph{Local Connectivity ATSP}.
Svensson, Tarnawski, and V\'egh compute a solution for Subtour Cover by rounding a circulation in a certain flow network, which is constructed
from the LP solution using a so-called \emph{witness flow}.
By using a special witness flow with certain minimality properties our Subtour Cover solution will obey stronger bounds.

In Section~\ref{sect:solve_vertebrate_pairs} we then explain how to compute solutions for vertebrate pairs using the algorithm for \emph{Subtour Cover} as a sub-routine.
The essential idea is due to Svensson~\cite{Sve15}, who considered node-weighted instances,
and was later adapted to vertebrate pairs in~\cite{SveTV18}.
In this part we make two improvements compared to the algorithm in~\cite{SveTV18}.

The more important change is the following. 
Svensson's algorithm uses a potential function to measure progress,
and in each of~\cite{Sve15} and~\cite{SveTV18} two different potential functions are considered.
One potential function is used to obtain an exponential time algorithm that yields an upper bound on the integrality ratio of the linear programming relaxation,
and the other potential function is used to obtain a polynomial-time algorithm.
This leads to different upper bounds on the integrality ratio of the LP and the approximation ratio of the algorithm.
We show in Section~\ref{sect:solve_vertebrate_pairs} that we can make this discrepancy arbitrarily small
by a slightly different choice of the potential function for the polynomial-time algorithm.
This leads to a better approximation ratio. 
Moreover, the analysis of the polynomial-time algorithm then immediately implies the best upper bound we know on the integrality ratio 
and there is no need anymore to consider two different potential functions.

The second change compared to the algorithm in~\cite{SveTV18} is that we include an idea that Svensson~\cite{Sve15} used for node-weighted instances.
This leads to another small improvement of the approximation guarantee.

Overall, we obtain for every $\epsilon >0$ a polynomial-time $(22+\epsilon)$-approximation algorithm for ATSP.
The algorithm computes a solution of cost at most $22+\epsilon$ times the cost of an optimum solution to the classic linear programming relaxation \eqref{eq:atsp_lp}, which we describe next.

\section{Reducing to strongly laminar instances}\label{sect:strongly_laminar}
As in the Svensson--Tarnawski--V\'egh algorithm, we begin by solving the classic linear programming relaxation:
  \begin{equation}\label{eq:atsp_lp}
  \begin{aligned}
   \min c(x)  \\
   s.t.& &  x(\delta^-(v)) -x(\delta^+(v)) =&\ 0 & & \text{ for } v\in V \\
   & &  x(\delta(U)) \ge&\ 2 & &\text{ for } \emptyset \ne U \subsetneq V \\
   & & x_e \ge & \ 0 & &\text{ for } e\in E,
  \end{aligned}
    \tag{ATSP LP}
 \end{equation}
 where $c(x):=\sum_{e\in E}c(e)x_e$, $x(F):=\sum_{e\in F}x_e$ for $F\subseteq E$, 
 $\delta^-(v)$ and $\delta^+(v)$ denote the sets of edges entering and leaving $v$, respectively, 
 and $\delta(U)$ denotes the set of edges with exactly one endpoint in $U$.
We also solve the dual LP:
\begin{equation}\label{eq:dual_atsp_lp}
   \begin{aligned}
    \max & &  \sum_{\emptyset \ne U \subsetneq V} \! 2 y_U \\
    s.t. & & a_w - a_v + \sum_{U: e\in\delta(U)} \! y_U \le&\ c(e) & &\text{ for } e=(v,w)\in E \\
   & & y_U \ge&\ 0 & &\text{ for } \emptyset \ne U \subsetneq V,
  \end{aligned}
  \tag{ATSP DUAL}
\end{equation}
where the variables $a_v$ ($v\in V$) are unbounded.
A family $\Lscr$ of subsets of $V$ is called \emph{laminar} if for any $A,B\in \Lscr$ we have $A\subseteq B$, $B\subseteq A$, or $A\cap B = \emptyset$.
Such a family has at most $2|V|$ elements.
The following is well-known (see e.g.~\cite{SveTV18}).
\begin{lemma}\label{lem:solve_lps} 
  Let $(G,c)$ be an instance of ATSP.
  Then we can compute in polynomial time an optimum solution $x$ to \eqref{eq:atsp_lp} and an optimum solution $(a,y)$ to \eqref{eq:dual_atsp_lp},
  such that $y$ has laminar support, i.e.\ $\Lscr := \{U: y_U>0\}$ is a laminar family.
\end{lemma}
We will now obtain LP solutions with more structure.
By $G[U]=(U, E[U])$ we denote the subgraph of $G=(V,E)$ induced by the vertex set $U$.
\begin{definition}
 Let $(G,c)$ be an instance of ATSP.
 Moreover, let $(a,y)$ be a dual LP solution, i.e.\ a solution to \eqref{eq:dual_atsp_lp}.
 We say that $y$ or $(a,y)$ has \emph{strongly laminar support} if
  \begin{itemize}\itemsep0pt
  \item $\Lscr := \{U: y_U>0\}$ is a laminar family, and 
  \item for every set $U\in \Lscr$, the graph $G[U]$ is strongly connected.
 \end{itemize}
\end{definition}

The following lemma allows us to assume that our optimum dual solution has strongly laminar support.

\begin{lemma}\label{lemma:strongly_laminar}
 Let $(G,c)$ be an instance of ATSP.
 Moreover, let $x$ be an optimum solution to  \eqref{eq:atsp_lp} and $(a,y)$ an optimum solution to \eqref{eq:dual_atsp_lp}
 with laminar support.
 Then we can compute in polynomial time $(a',y')$ such that
 \begin{itemize}\itemsep0pt
  \item $(a',y')$ is an optimum solution of \eqref{eq:dual_atsp_lp}, and
  \item $(a',y')$ has strongly laminar support.
 \end{itemize}
\end{lemma}
\begin{proof}
As long as there is a set $U$ with $y_U > 0$, but $G[U]$ is not strongly connected, we do the following.
Let $U$ be a minimal set with $y_U >0$ and such that $G[U]$ is not strongly connected.
Moreover, let $S$ be the vertex set of the first strongly connected component of $G[U]$ in a topological order.
Then we have $\delta^-(S) \subseteq \delta^-(U)$.

Define a dual solution $(a', y')$ as follows. 
We set $y'_U := 0$, $y'_S := y_S + y_U$, and $y'_W := y_W$ for other sets $W$.
Moreover, $a'_v:= a_v - y_U$ for $v\in U\setminus S$ and $a'_v := a_v$ for all other vertices $v$.
The only edges $e=(v,w)$ for which $ a'_w - a'_v + \sum_{U: e\in\delta(U)} y'_U > a_w - a_v + \sum_{U: e\in\delta(U)} y_U $,
are edges from $U\setminus S$ to $S$.
However, such edges do not exist by choice of $S$.
Hence, $(a',y')$ is a feasible dual solution.
Since $\sum_{\emptyset \ne U \subsetneq V} 2y'_U =\sum_{\emptyset \ne U \subsetneq V} 2y_U$, it is also optimal.

We now show that the support of $y'$ is laminar.
Suppose there is a set $W$ in the support of $y'$ that crosses $S$.
Then $W$ must be in the support of $y$ and hence a subset of $U$ because the support of $y$ is laminar.
By the minimal choice of $U$, $G[W]$ is strongly connected.
But this implies that $G$ contains an edge from $W\setminus S$ to $W\cap S$, contradicting $\delta^-(S) \subseteq \delta^-(U)$.

We now decreased the number of sets $U$ in the support for which $G[U]$ is not strongly connected.
After iterating this at most $2|V|$ times the dual solution has the desired properties.
\end{proof}

As Svensson, Tarnawski, and V\'egh~\cite{SveTV18}, we next show that we may assume the dual variables $a_v$ to be $0$ for all $v\in V$.
This leads to the following definition.

\begin{definition}\label{def:strongly_laminar}
A \emph{strongly laminar ATSP instance} is a quadruple $(G,\Lscr,x,y)$, where 
\begin{enumerate}[(i)]\itemsep0pt
 \item $G=(V,E)$ is a strongly connected digraph;
 \item $\Lscr$ is a laminar family of subsets of $V$ such that $G[U]$ is strongly connected for all $U\in \Lscr$;
 \item $x$ is a feasible solution to \eqref{eq:atsp_lp} such that $x(\delta(U))=2$ for all $U\in\Lscr$ and $x_e > 0$ for all $e\in E$;
 \label{item:x_in_strongly_laminar}
 \item $y:\Lscr\to\mathbb{R}_{> 0}$.
\end{enumerate}
This induces the ATSP instance $(G,c)$, where $c$ is the induced weight function defined by
$c(e):=\sum_{U\in\Lscr: e\in\delta(U)} y_U$ for all $e\in E$.
\end{definition}

By complementary slackness, $x$ and $(0,y)$ are optimum solutions of  \eqref{eq:atsp_lp} and \eqref{eq:dual_atsp_lp} for $(G,c)$.
For a strongly laminar instance $\Iscr$ we denote by $\lp(\Iscr)=c(x)$ the value of these LPs. 
We now prove that for ATSP it is sufficient to consider strongly laminar instances.

\begin{theorem}
\label{thm:stronglylaminar}
Let $\alpha\ge 1$. 
If there is a polynomial-time algorithm  that computes for every strongly laminar ATSP instance $(G,\Lscr,x,y)$
a solution of cost at most $\alpha \cdot c(x)$,
then there is a polynomial-time algorithm that computes for every instance of ATSP a solution of cost at most 
$\alpha$ times the cost of an optimum solution to \eqref{eq:atsp_lp}.
\end{theorem}

\begin{proof}
Let $(G,c)$ be an arbitrary instance.
We apply Lemma~\ref{lem:solve_lps} to compute an optimum solution $x$ of \eqref{eq:atsp_lp} and an optimum
solution $(a,y)$ of \eqref{eq:dual_atsp_lp} such that the support of $y$ is a laminar familiy $\Lscr$.

Now let $E'$ be the support of $x$ and define $G':=(V,E')$. 
Let $x'$ be the vector $x$ restricted to its support $E'$.
Then apply Lemma~\ref{lemma:strongly_laminar} to $(G',x',y)$.
We obtain an optimum dual solution $(a',y')$ to \eqref{eq:dual_atsp_lp} with strongly laminar support $\Lscr$.
By complementary slackness we have $x'(\delta(U))=2$ for all $U\in\Lscr$ with $y'_{U} > 0$.

Then the induced weight function of the strongly laminar ATSP instance $(G',\Lscr,x',y')$ 
is given by $c'(e)=\sum_{S\in\Lscr: e\in\delta(S)} y'_S = c(e)+a_v-a_w$ for all $e=(v,w)\in E'$ (by complementary slackness).
Because every tour in $G'$ is Eulerian, it has the same cost with respect to $c$ and with respect to $c'$. 
Moreover, $c(x)=c'(x')$ and $(0,y')$ is an optimum dual solution for $(G',c')$.
Hence also the LP values are the same and thus the theorem follows.
\end{proof}

One advantage of this structure is that we always have \emph{nice paths}, which are defined as follows.

\begin{definition}
\label{def:nice_path}
 Let $G=(V,E)$ be a directed graph and let $\Lscr$ be a laminar family.
 Let $v,w\in V$ and let $\tilde U$ be the minimal set in $\Lscr\cup\{V\}$ with $v,w\in \tilde U$. 
 A $v$-$w$-path is \emph{nice} if it is in $G[\tilde U]$ and it enters and leaves every set $U\in \Lscr$ at most once.
\end{definition}

\begin{lemma}
\label{lemma:enterandleaveonlyonce}
 Let $G=(V,E)$ be a strongly connected directed graph and let $\Lscr$ be a laminar family such that $G[U]$ is strongly connected for every $U\in \Lscr$.
Then for any $v,w\in V$ we can find a nice $v$-$w$-path in polynomial time.
\end{lemma}

\begin{proof}
Let $P$ be a path from $v$ to $w$ in $G[\tilde U]$. 
Now repeat the following until $P$ enters and leaves every set in $\Lscr$ at most once. 
Let $U$ be a maximal set with $U\in \Lscr$ that $P$ enters or leaves more than once.
Let $v'$ be the first vertex that $P$ visits in $U$ and let $w'$ be the last vertex that $P$ visits in $U$.
Since $G[U]$ is strongly connected, we can replace the $v'$-$w'$-subpath of $P$ by a path in $G[U]$. 
After at most $|\Lscr| < 2 |V|$ iterations, $P$ is a nice $v$-$w$-path.
\end{proof}

\section{Reducing to vertebrate pairs}\label{sect:reduction_veterbrate_pair}

Let $\Lscr_{\ge 2} := \{ L\in \Lscr : |L|\ge 2 \}$ be the family of all non-singleton elements of  $\Lscr$.
In this section we show how to reduce ATSP to the case where we have already a given subtour $B$, called \emph{backbone},
that visits all elements of $\Lscr_{\ge 2}$; see Figure~\ref{fig:vertebrate_pair}.
By a \emph{subtour} we mean a connected Eulerian multi-subgraph of $G$.
We call a strongly laminar ATSP instance together with a given backbone a \emph{vertebrate pair}.

\begin{definition}
A \emph{vertebrate pair} consists of 
\begin{itemize}\itemsep0pt
 \item a strongly laminar ATSP instance $\Iscr=(G,\Lscr,x,y)$ and 
 \item a connected Eulerian multi-subgraph $B$ of $G$ (the \emph{backbone}) such that \\
       $V(B)\cap L\not=\emptyset$ for all $L\in\Lscr_{\ge 2}$.
\end{itemize}
Let $\kappa,\eta \ge 0$.
A \emph{$(\kappa,\eta)$-algorithm for vertebrate pairs} is an algorithm that computes, for any 
given vertebrate pair $(\Iscr,B)$, a multi-set $F$ of edges such that $E(B)\cupp F$ is a tour 
and
\begin{equation}\label{eq:definition_kappa_eta_algorithm}
 c(F)\le \kappa\cdot\lp(\Iscr) + \eta \cdot \sum_{v\in V\setminus V(B): \{v\} \in \Lscr} 2y_{\{v\}}. \\
\end{equation}
\end{definition}
Note that this definition is slightly different to the one in~\cite{SveTV18} 
(where $G[L]$ was not required to be strongly connected for $L\in \Lscr$), but this will not be relevant.

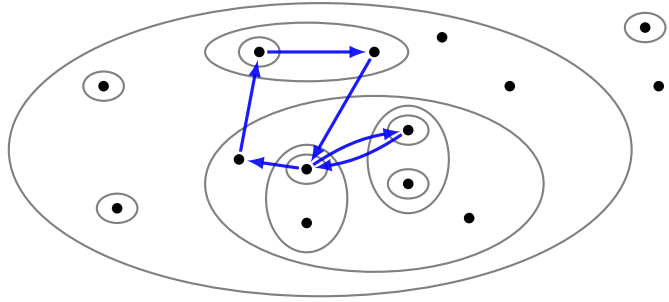
\begin{figure}
\begin{center}
 \begin{tikzpicture}[yscale=0.65, xscale=0.9]
  \tikzset{vertex/.style={fill=black, circle,inner sep=0em,minimum size=4pt,outer sep =1pt }}
  \tikzset{Backbone/.style={->, >=latex, line width=1.2pt,blue, opacity=0.9}}
  \tikzset{Laminar/.style={draw=gray, fill=none, thick}}
     \begin{scope}[every node/.style={vertex}]
    \node (v1) at (5,4) {};

    \node (v12) at (8,2.3) {};
    \node (v13) at (8,1.2) {};
    \node (v14) at (9.5,3.1) {};
    \node (v15) at (9.5,2) {};
    \node (v16) at (10.4,1.3) {};
    
    \node (v17) at (9,4.7) {};
    \node (v18) at (10,5) {};
    \node (v19) at (11,4) {};
    \node (v20) at (7,2.5) {};
    \node (v21) at (13.2,4) {};
    
    \node (a) at (7.3,4.7){};
    \node (b) at (5.2,1.5) {};
   \end{scope}
    \node[vertex] (outside) at (13,5.2) {};

   \begin{scope}[Laminar]
   \draw (8,4.7) ellipse (1.5 and 0.6);
   \draw (v1) ellipse (0.3 and 0.3);
   \draw (a) ellipse (0.3 and 0.3);
   \draw (b) ellipse (0.3 and 0.3); 
   \draw (v12) ellipse (0.3 and 0.3);
   \draw (v14) ellipse (0.3 and 0.3);
   \draw (v15) ellipse (0.3 and 0.3);
   \draw (outside) ellipse (0.3 and 0.3);
   \draw (8,1.7) ellipse (0.6 and 1.1);
   \draw (9.5,2.5) ellipse (0.6 and 1.1);
   \draw (9, 2.0) ellipse (2.5 and 1.8); 
   \draw (8.2,2.7) ellipse (4.6 and 3);
   \end{scope}
   
   \begin{scope}[Backbone]
    \draw[bend left=15] (v12) to (v14);
    \draw[bend left=15] (v14) to (v12);
    \draw (v12) to (v20);
    \draw (v20) to (a);
    \draw (a) to (v17);
    \draw (v17) to (v12);
   \end{scope}
 \end{tikzpicture}
\end{center}
\caption{Example of a vertebrate pair.
The laminar family $\Lscr$ is shown in gray and a backbone $B$ is shown in blue.
\label{fig:vertebrate_pair}
}
\end{figure}

In this section we will show that a $(\kappa,\eta)$-algorithm for vertebrate pairs (for any constants
$\kappa$ and $\eta$) implies a $(3\kappa + \eta + 2)$-approximation algorithm for ATSP.
This is the reason for using the bound in \eqref{eq:definition_kappa_eta_algorithm}; we would get a worse overall approximation guarantee if we just worked with the weaker inequality $c(F) \le (\kappa + \eta) \cdot \lp(\Iscr)$.

Let $(G,\Lscr,x,y)$ be a strongly laminar ATSP instance and $c$ the induced cost function.
In the following we fix for every $u,v\in V$ a nice $u$-$v$-path $P_{u,v}$.
Such paths can be computed in polynomial time by Lemma~\ref{lemma:enterandleaveonlyonce}.
\begin{lemma}\label{lemma:cost_of_nice_path}
Let $W\in \Lscr\cup \{V\}$ and let $u,v\in W$. Then 
\[
 c\left(E\left(P_{u,v}\right)\right) \ =\ \sum_{L\in \Lscr:\, L\subsetneq W, L\cap V(P_{u,v}) \ne \emptyset} 2y_L 
 \ - \sum_{L\in\Lscr:\, u\in L\subsetneq W} y_L\ - \sum_{L\in\Lscr:\, v\in L\subsetneq W} y_L.
\]
\end{lemma}
\begin{proof}
Since the path $P_{u,v}$ is nice, it is contained in $G[W]$.
Moreover, it leaves every set $L\in \Lscr$ at most once and enters every set $L\in \Lscr$ at most once.
A set $L\in \Lscr$ with $u\in L$ is never entered by $P_{u,v}$ and a set $L\in \Lscr$  with $w\in L$ is never left by $P_{u,v}$.
\end{proof}
We define 
\[ \val(W):=\sum_{L\in\Lscr: L\subsetneq W}2y_L. \]
and 
$$ D_W(u,v) \ := \ \sum_{L\in\Lscr:\, u\in L\subsetneq W} y_L + \sum_{L\in\Lscr:\, v\in L\subsetneq W} y_L + c(E(P_{u,v}))$$
for $u,v\in W$.
Note that $D_W(u,v)\le \val(W)$ by Lemma~\ref{lemma:cost_of_nice_path}.
We write 
\[
 D_W:=\max\{D_W(u,v) : u,v\in W\}.
\]
The intuitive meaning of $D_W$ in the analysis of our reduction to vertebrate pairs is the following.
On the one hand, it can be useful if $D_W$ is small:
if we enter the set $W$ at some vertex $s\in W$ and leave it at some other vertex $t\in W$,
we can always find a cheap $s$-$t$-walk inside $G[W]$.
On the other hand, if $D_W$ is large, we can find a nice path inside $W$ that visits many sets $L\in\Lscr$
(or more precisely, sets of high weight in the dual solution $y$).

The reduction to vertebrate pairs is via a recursive algorithm.
For a given set $W\in\Lscr \cup \{V\}$ it constructs a tour in $G[W]$. 
See Figure~\ref{fig:reduction_veterbrate_pair} for an illustration.

\medskip
\begin{algorithm2e}[H]
\vspace*{1mm}
\KwIn{a strongly laminar ATSP instance $\Iscr=(G,\Lscr,x,y)$ with $G=(V,E)$,\newline
      a set $W\in\Lscr \cup \{V\}$, and \newline
      a $(\kappa, \eta)$-algorithm $\Ascr$ for vertebrate pairs (for some constants $\kappa,\eta\ge 0$)}
\KwOut{a tour $F$ in $G[W]$}
\vspace*{2mm}

\begin{enumerate}\itemsep1pt
\item 
      If $W \ne V$, contract $V\setminus W$ into a single vertex $v_{\bar W}$ and redefine $y_W := \sfrac{D_W}{2}$.      
\item \textbf{Construct a vertebrate pair:}
      Let $u^*,v^* \in W$ such that $D_W(u^*,v^*) = D_W$. Let $B$ be the multi-graph corresponding to the closed walk
      that results from appending $P_{u^*\!,v^*}$ and $P_{v^*\!,u^*}$.
      \\[1mm]
      Let $\Lscr_{\bar B}$ be the set of all maximal sets $L\in \Lscr$ with $L\subsetneq W$ and $V(B)\cap L =\emptyset$.
      Contract every set $L\in \Lscr_{\bar B}$ to a single vertex $v_L$ and set $y_{\{v_L\}} := y_L + \sfrac{D_L}{2}$.
      Let $G'$ be the resulting graph.
      \\[1mm]
      Let $\Lscr'$ be the laminar family of subsets of $V(G')$ that contains singletons $\{v_L\}$ for $L\in \Lscr_{\bar B}$
      and all the sets arising from $L \in \Lscr$ with $L \subseteq W$ and $L\cap V(B) \ne \emptyset$.
      \\[1mm]
      Let $\Iscr'=(G',\Lscr',x,y)$ be the resulting strongly laminar instance.      
\item \textbf{Compute a solution for the vertebrate pair:} 
       Apply the given algorithm $\Ascr$ to the vertebrate pair $(\Iscr', B)$.
       Let $F'$ be the resulting Eulerian edge set.
\item \textbf{Lift the solution to a subtour:}
       Fix an Eulerian walk in every connected component of $F'$.
       Now uncontract every $L\in\Lscr_{\bar B}$. 
       Whenever an Eulerian walk passes through $v_L$, we get two edges $(u',u)\in \delta^-(L)$ and $(v,v') \in\delta^+(L)$.
       To connect $u$ and $v$ within $L$, add the path $P_{u,v}$.
       \\[1mm]
       Moreover, if $W\ne V$ do the following.
       Whenever an Eulerian walk passes through $v_{\bar W}$ using the edges $(u, v_{\bar W})$ and $(v_{\bar W}, v)$,
       replace them by the path $P_{u,v}$. 
\item \textbf{Recurse to complete to a tour of the original instance:}
      For every set $L\in\Lscr_{\bar B}$, apply Algorithm~\ref{algo:reduction_veterbrate_pair} recursively to obtain a tour $F_L$
      in $G[L]$.
      Let $F''$ be the union of $F'$ and all these tours $F_L$ for $L\in\Lscr_{\bar B}$.
\item Return $F := F'' \cupp E(B)$. 
\end{enumerate}
\caption{Recursive algorithm to reduce to vertebrate pairs. \label{algo:reduction_veterbrate_pair}
}
\end{algorithm2e}
\medskip

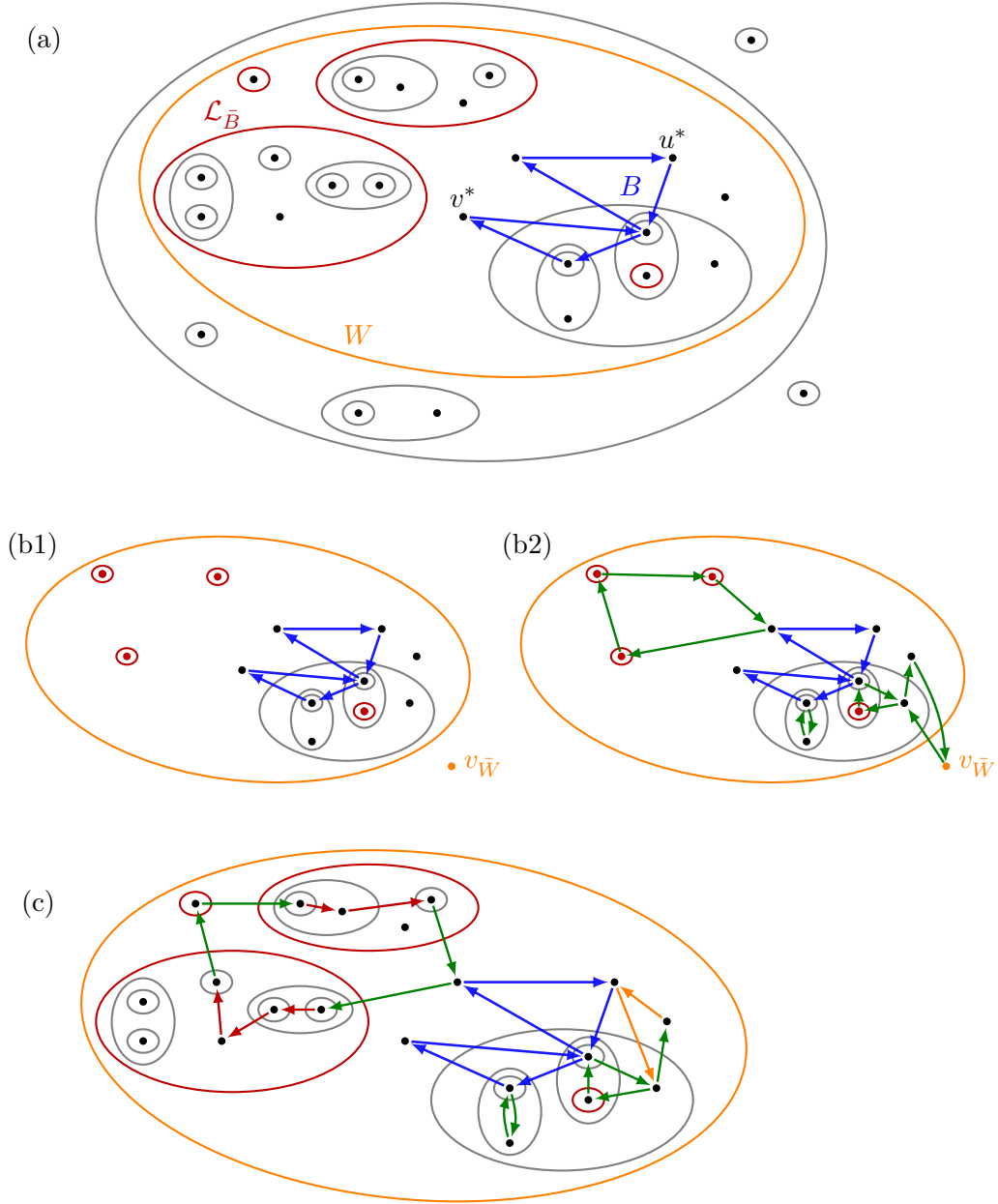
\begin{figure}
\begin{center}
 \begin{tikzpicture}[xscale=0.8,yscale=0.54]
 
  \tikzset{vertex/.style={fill=black, circle,inner sep=0em,minimum size=3pt,outer sep =1pt }}
  \tikzset{contracted/.style={fill=darkred, circle, inner sep=0em,minimum size=3pt }}
  \tikzset{PathCompletion/.style={->, >=latex, line width=0.9pt,darkred}}
  \tikzset{setNotB/.style={draw=darkred, fill=none, thick}}
  \tikzset{OuterContracted/.style={fill=orange, circle, inner sep=0em,minimum size=3pt }}
  \tikzset{setW/.style={draw=orange, fill=none, thick}}
  \tikzset{OuterPathCompletion/.style={->, >=latex, line width=0.9pt,orange}}
  \tikzset{Backbone/.style={->, >=latex, line width=1pt,blue, opacity=0.9}}
  \tikzset{VertebratePairSolution/.style={->, >=latex, line width=0.9pt, darkgreen}}
  \tikzset{Laminar/.style={draw=gray, fill=none, thick}}

  \begin{scope}[shift={(2,0)},xscale=0.9]
   \node () at (-2,8) {(a)};
   
   \begin{scope}[every node/.style={vertex}]
    \node (v1) at (2,7) {};
    \node (v2) at (4,7) {};
    \node (v3) at (4.8,6.8) {};
    \node (v4) at (6.5,7.1) {};
    \node (v5) at (6,6.4) {};
    
    \node (v6) at (2.4,5) {};
    \node (v7) at (1,4.5) {};
    \node (v8) at (1,3.5) {};
    \node (v9) at (2.5,3.5) {};
    \node (v10) at (3.5,4.3) {};
    \node (v11) at (4.4,4.3) {};
    
    \node (v12) at (8,2.3) {};
    \node (v13) at (8,0.9) {};
    \node (v14) at (9.5,3.1) {};
    \node (v15) at (9.5,2) {};
    \node (v16) at (10.8,2.3) {};
    
    \node (v17) at (7,5) {};
    \node (v18) at (10,5) {};
    \node (v19) at (11,4) {};
    \node (v20) at (6,3.5) {};
    
    \node (v21) at (4,-1.5) {};
    \node (v22) at (5.5, -1.5) {};
    \node (v23) at (1,0.5) {};
    \node (v24) at (12.5,-1) {};
    \node (v25) at (11.5,8) {};
   \end{scope}

   \begin{scope}[Laminar]
   \draw[setNotB] (v1) ellipse (0.3 and 0.3);
   \draw (v2) ellipse (0.3 and 0.3);
   \draw (4.5, 6.9) ellipse (1 and 0.7);
   \draw (v4) ellipse (0.3 and 0.3);
   \draw[setNotB] (5.3,6.9) ellipse (2.1 and 1.1);
   
   \draw (v6) ellipse (0.3 and 0.3);
   \draw (v7) ellipse (0.3 and 0.3);
   \draw (v8) ellipse (0.3 and 0.3);
   \draw (1,4) ellipse (0.6 and 1.1);
   \draw (v10) ellipse (0.3 and 0.3);
   \draw (v11) ellipse (0.3 and 0.3);
   \draw (4,4.3) ellipse (1.0 and 0.6);
   \draw[setNotB] (2.7,4) ellipse (2.6 and 1.8);
   
   \draw (v12) ellipse (0.3 and 0.3);
   \draw (v14) ellipse (0.3 and 0.3);
   \draw[setNotB] (v15) ellipse (0.3 and 0.3);
   \draw (8,1.7) ellipse (0.6 and 1.1);
   \draw (9.5,2.5) ellipse (0.6 and 1.1);
   \draw (9, 2.0) ellipse (2.5 and 1.8);
   
   \draw (v21) ellipse (0.3 and 0.3);
   \draw (v23) ellipse (0.3 and 0.3);
   \draw (4.8, -1.5) ellipse (1.5 and 0.7);
   \draw (v24) ellipse (0.3 and 0.3);
   \draw (v25) ellipse (0.3 and 0.3);
   
   \draw[setW, rotate=-10] (5.4,4.9) ellipse (6.4 and 4.4);
   \draw[rotate=-9] (5.4,4) ellipse (7 and 5.8);
   \end{scope}

   \begin{scope}[Backbone]
    \draw (v14) to (v12);
    \draw (v12) to (v20);
    \draw (v20) to (v14);
    \draw (v14) to (v17);
    \draw (v17) to (v18);
    \draw (v18) to (v14);
   \end{scope}
   
   \node[above] (u) at (v18) {$u^*$};
   \node[above] (v) at (v20) {$v^*$};
   \node[blue] (B) at (9.2,4.3) {$B$};
   \node[orange] (W) at (4,0.5) {$W$};
   \node[darkred] () at (1.4,6.1) {$\Lscr_{\bar B}$};
\end{scope}

\begin{scope}[shift={(0,-10.5)},xscale=0.6, yscale=0.7]
\node () at (0,8) {(b1)};

   \begin{scope}[every node/.style={vertex}]
    \node (v1) at (2,7) {};

    \node (v12) at (8,2.3) {};
    \node (v13) at (8,0.9) {};
    \node (v14) at (9.5,3.1) {};
    \node (v15) at (9.5,2) {};
    \node (v16) at (10.8,2.3) {};
    
    \node (v17) at (7,5) {};
    \node (v18) at (10,5) {};
    \node (v19) at (11,4) {};
    \node (v20) at (6,3.5) {};

   \end{scope}

   \begin{scope}[Laminar]
   \draw[setNotB] (v1) ellipse (0.3 and 0.3);
   \node[contracted] (L1) at (v1) {};

   \draw[setNotB] (5.3,6.9) ellipse (0.3 and 0.3);
   \node[contracted] (L2) at (5.3,6.9) {};
   
   \draw[setNotB] (2.7,4) ellipse (0.3 and 0.3);
   \node[contracted] (L3) at (2.7,4) {};
   
   \draw (v12) ellipse (0.3 and 0.3);
   \draw (v14) ellipse (0.3 and 0.3);
   \draw[setNotB] (v15) ellipse (0.3 and 0.3);
   \node[contracted] (L4) at (v15) {};
   \draw (8,1.7) ellipse (0.6 and 1.1);
   \draw (9.5,2.5) ellipse (0.6 and 1.1);
   \draw (9, 2.0) ellipse (2.5 and 1.8);
   
   \draw[setW, rotate=-10] (5.4,4.9) ellipse (6.4 and 4.4);
   \end{scope}
   
   \node[OuterContracted] (outside) at (12,0) {};
   \node[right=1pt, orange] () at (outside) {$v_{\bar W}$};

   \begin{scope}[Backbone]
    \draw (v14) to (v12);
    \draw (v12) to (v20);
    \draw (v20) to (v14);
    \draw (v14) to (v17);
    \draw (v17) to (v18);
    \draw (v18) to (v14);
   \end{scope}
 
\end{scope}

\begin{scope}[shift={(8.5,-10.5)},xscale=0.6, yscale=0.7]
\node () at (0,8) {(b2)};

   \begin{scope}[every node/.style={vertex}]
    \node (v1) at (2,7) {};

    \node (v12) at (8,2.3) {};
    \node (v13) at (8,0.9) {};
    \node (v14) at (9.5,3.1) {};
    \node (v15) at (9.5,2) {};
    \node (v16) at (10.8,2.3) {};
    
    \node (v17) at (7,5) {};
    \node (v18) at (10,5) {};
    \node (v19) at (11,4) {};
    \node (v20) at (6,3.5) {};

   \end{scope}

   \begin{scope}[Laminar]
   \draw[setNotB] (v1) ellipse (0.3 and 0.3);
   \node[contracted] (L1) at (v1) {};

   \draw[setNotB] (5.3,6.9) ellipse (0.3 and 0.3);
   \node[contracted] (L2) at (5.3,6.9) {};
   
   \draw[setNotB] (2.7,4) ellipse (0.3 and 0.3);
   \node[contracted] (L3) at (2.7,4) {};
   
   \draw (v12) ellipse (0.3 and 0.3);
   \draw (v14) ellipse (0.3 and 0.3);
   \draw[setNotB] (v15) ellipse (0.3 and 0.3);
   \node[contracted] (L4) at (v15) {};
   \draw (8,1.7) ellipse (0.6 and 1.1);
   \draw (9.5,2.5) ellipse (0.6 and 1.1);
   \draw (9, 2.0) ellipse (2.5 and 1.8);
   
   \draw[setW, rotate=-10] (5.4,4.9) ellipse (6.4 and 4.4);
   \end{scope}
   
   \node[OuterContracted] (outside) at (12,0) {};
   \node[right=1pt, orange] () at (outside) {$v_{\bar W}$};

   \begin{scope}[Backbone]
    \draw (v14) to (v12);
    \draw (v12) to (v20);
    \draw (v20) to (v14);
    \draw (v14) to (v17);
    \draw (v17) to (v18);
    \draw (v18) to (v14);
   \end{scope}
   
   \begin{scope}[VertebratePairSolution]
    \draw (L1) to (L2);
    \draw (L3) to (L1);
    \draw (L2) to (v17);
    \draw (v17) to (L3);
    
    \draw[bend left=15pt] (v12) to (v13);
    \draw[bend left=15pt] (v13) to (v12);
    
    \draw (L4) to (v14);
    \draw (v14) to (v16);
    \draw(v16) to (v19);
    \draw[bend left=10pt] (v19) to (outside);
    \draw (outside) to (v16);
    \draw (v16) to (L4);
   \end{scope}   
 
\end{scope}

 \begin{scope}[shift={(1,-21)},xscale=0.9]
 \node () at (-1,7) {(c)};

    \begin{scope}[every node/.style={vertex}]
    \node (v1) at (2,7) {};
    \node (v2) at (4,7) {};
    \node (v3) at (4.8,6.8) {};
    \node (v4) at (6.5,7.1) {};
    \node (v5) at (6,6.4) {};
    
    \node (v6) at (2.4,5) {};
    \node (v7) at (1,4.5) {};
    \node (v8) at (1,3.5) {};
    \node (v9) at (2.5,3.5) {};
    \node (v10) at (3.5,4.3) {};
    \node (v11) at (4.4,4.3) {};
    
    \node (v12) at (8,2.3) {};
    \node (v13) at (8,0.9) {};
    \node (v14) at (9.5,3.1) {};
    \node (v15) at (9.5,2) {};
    \node (v16) at (10.8,2.3) {};
    
    \node (v17) at (7,5) {};
    \node (v18) at (10,5) {};
    \node (v19) at (11,4) {};
    \node (v20) at (6,3.5) {};
   \end{scope}

   \begin{scope}[Laminar]
   \draw[setNotB] (v1) ellipse (0.3 and 0.3);
   \draw (v2) ellipse (0.3 and 0.3);
   \draw (4.5, 6.9) ellipse (1 and 0.7);
   \draw (v4) ellipse (0.3 and 0.3);
   \draw[setNotB] (5.3,6.9) ellipse (2.1 and 1.1);
   
   \draw (v6) ellipse (0.3 and 0.3);
   \draw (v7) ellipse (0.3 and 0.3);
   \draw (v8) ellipse (0.3 and 0.3);
   \draw (1,4) ellipse (0.6 and 1.1);
   \draw (v10) ellipse (0.3 and 0.3);
   \draw (v11) ellipse (0.3 and 0.3);
   \draw (4,4.3) ellipse (1.0 and 0.6);
   \draw[setNotB] (2.7,4) ellipse (2.6 and 1.8);
   
   \draw (v12) ellipse (0.3 and 0.3);
   \draw (v14) ellipse (0.3 and 0.3);
   \draw[setNotB] (v15) ellipse (0.3 and 0.3);
   \draw (8,1.7) ellipse (0.6 and 1.1);
   \draw (9.5,2.5) ellipse (0.6 and 1.1);
   \draw (9, 2.0) ellipse (2.5 and 1.8);
   
   \draw[setW, rotate=-10] (5.4,4.9) ellipse (6.4 and 4.4);
   \end{scope}

   \begin{scope}[Backbone]
    \draw (v14) to (v12);
    \draw (v12) to (v20);
    \draw (v20) to (v14);
    \draw (v14) to (v17);
    \draw (v17) to (v18);
    \draw (v18) to (v14);
   \end{scope}
   
    \begin{scope}[VertebratePairSolution]
     \draw (v1) to (v2);
     \draw (v6) to (v1);
     \draw (v4) to (v17);
    \draw (v17) to (v11);
    
    \draw[bend left=12pt] (v12) to (v13);
    \draw[bend left=12pt] (v13) to (v12);
    
    \draw (v15) to (v14);
    \draw (v14) to (v16);
    \draw(v16) to (v19);
    \draw (v16) to (v15);
   \end{scope} 
   
   \begin{scope}[PathCompletion]
    \draw (v2) to (v3);
    \draw (v3) to (v4);
    \draw (v11) to (v10);
    \draw (v10) to (v9);
    \draw (v9) to (v6);
   \end{scope}
   
    \draw[OuterPathCompletion] (v19) to (v18);
    \draw[OuterPathCompletion] (v18) to (v16);
 \end{scope}

\end{tikzpicture}
\end{center}
\caption{Illustration of Algorithm~\ref{algo:reduction_veterbrate_pair}.
The ellipses show the laminar family $\Lscr$. 
Picture~(a) shows the set $W$ (orange), the subtour $B$ (blue), and the elements of $\Lscr_{\bar B}$ (red). 
The subtour $B$ is the union of the paths $P_{u^*,v^*}$ and $P_{v^*,u^*}$.
Picture~(b1) shows the resulting vertebrate pair instance as constructed in step~2 of Algorithm~\ref{algo:reduction_veterbrate_pair}.
The vertices resulting from the contraction of elements of $\Lscr_{\bar B}$ are shown in red and the vertex $v_{\bar W}$ that 
results from the contraction of $V \setminus W$ is shown in orange.
Picture~(b2) shows in green a possible solution to this vertebrate pair.
Picture~(c) illustrates step~4 of Algorithm~\ref{algo:reduction_veterbrate_pair}:
the green edges are those that arise from the vertebrate pair solution from Picture~(b2) by undoing the contraction of the 
sets in $\Lscr_{\bar B}$. The red edges are the paths that we add to connect within $L\in \Lscr_{\bar B}$ when
uncontracting $L$. The orange edges show the $u$-$v$-path in $G[W]$ that we add to replace the edges
$(u, v_{\bar W})$ and $(v_{\bar W}, v)$ in the vertebrate pair solution from Picture~(b2).
\label{fig:reduction_veterbrate_pair}
}
\end{figure}

First, we observe that Algorithm~\ref{algo:reduction_veterbrate_pair} indeed returns a tour in $G[W]$.

\begin{lemma}\label{lemma:algo_feasible_reduction_vertebrate_pairs}
Let $\kappa,\eta\ge 1$. 
Suppose we have a polynomial-time $(\kappa,\eta)$-algorithm $\Ascr$ for vertebrate pairs.
Then  Algorithm~\ref{algo:reduction_veterbrate_pair} has polynomial runtime and returns a tour in $G[W]$
for every strongly laminar ATSP instance $\Iscr=(G,\Lscr,x,y)$ and every $W\in\Lscr\cup \{V\}$.
\end{lemma}
\begin{proof}
 We apply induction on $|W|$. For $|W|=1$, the algorithm returns $F=\emptyset$.
 Now let $|W| > 1$. 
 At the end of step~3, we have that $F'$ is Eulerian and $F'\cupp E(B)$ is a tour in the instance $\Iscr'$.
 In step~4, the set $F'$ remains Eulerian and $F'\cupp E(B)$ remains connected.
 Moreover, the subtour $F'\cupp E(B)$ visits all sets in $\Lscr_{\bar B}$, i.e.\ we have 
 $F' \cap \delta(L) \ne \emptyset$ for all $L\in\Lscr_{\bar B}$.
 The subtour $F'\cupp E(B)$ also visits all vertices in $W$ that are not contained in any set $L\in\Lscr_{\bar B}$,
 i.e.\ for these vertices $v$ we have  $\delta(v) \cap (F' \cup E(B)) \ne \emptyset$. 
 After step~4, we have $F'\subseteq E[W]$.
 We conclude that the graph $(W, F''\cupp E(B))$ is connected and Eulerian; 
 here we applied the induction hypothesis to the sets $L\in \Lscr_{\bar B}$.

 To see that the runtime of the algorithm is polynomially bounded we observe that there are in total at most 
 $|\Lscr| +1 \le 2 |V|$ recursive calls of the algorithm because $\Lscr\cup \{V\}$ is a laminar family.
\end{proof}

Next we observe that our backbone $B$ visits many sets $L\in\Lscr$ inside $W$ if $D_W$ is large.
This is because the path $P_{u^*,v^*}$ enters and leaves each set in $\Lscr$ at most once and thus $D_W$ can only be large if this path visits sets in $\{ L\in \Lscr : L \subsetneq W\}$ of large total weight.

\begin{lemma}\label{lemma:L_bar_W_small}
 Let $\Iscr=(G,\Lscr,x,y)$ be a strongly laminar ATSP instance, and let $W\in\Lscr\cup \{V\}$.
 Moreover, let $B$ be as in step~2 of Algorithm~\ref{algo:reduction_veterbrate_pair}.
 Then
 \begin{equation}\label{eq:L_bar_W_small}
 \sum_{L\in\Lscr_{\bar B}} (2y_L + \val(L)) \le \val(W) - D_W.
\end{equation}
\end{lemma}
\begin{proof}
 By Lemma~\ref{lemma:cost_of_nice_path} and the choice of $u^*$ and $v^*$ we get
\begin{align*}
  \val(W) -  \sum_{L\in\Lscr_{\bar B}} (2y_L + \val(L))\ =\ &
 \sum_{L\in \Lscr:\, L\subsetneq W,\, L\cap V(B) \ne \emptyset} 2y_L\\
 \ge\ & \sum_{L\in \Lscr:\, L\subsetneq W,\, L\cap V(P_{u^*\!, v^*}) \ne \emptyset} 2y_L \\
 =\ & c(E(P_{u^*\!,v^*})) + \sum_{L\in\Lscr:\, u^*\in L\subsetneq W} y_U + \sum_{L\in\Lscr:\, v^*\in L\subsetneq W} y_U \\
 =\ & D_W(u^*,v^*) \\
 =\ & D_W.     \qedhere 
\end{align*}
\end{proof}

Now we analyze the cost of the tour $F$ in $G[W]$ computed by Algorithm~\ref{algo:reduction_veterbrate_pair}.

\begin{lemma}
\label{lemma:atsp_sttour_insubsets}
Let $\kappa,\eta\ge 0$. 
Suppose we have a $(\kappa,\eta)$-algorithm $\Ascr$ for vertebrate pairs.
Let $\Iscr=(G,\Lscr,x,y)$ be a strongly laminar ATSP instance, $c$ the induced cost function, and $W\in\Lscr\cup \{V\}$.
Then the tour $F$ in $G[W]$ returned by Algorithm~\ref{algo:reduction_veterbrate_pair} has cost at most
\begin{align*}
c(F) \ \le \ &\left(2\kappa+2 \right)\cdot\val(W) + (\kappa+\eta)\cdot(\val(W)-D_W). \\
\end{align*}
\end{lemma}

\begin{proof}
By induction on $|W|$. The statement is trivial for $|W|=1$ since then $c(F)=0$ (because $F \subseteq E[W] = \emptyset$).
Let now $|W|\ge 2$.
By definition of $D_W$, we have 
\begin{equation}\label{eq:cost_bound_backbone}
\begin{aligned}
 c(E(B))\ =\ c(E(P_{u^*\!,v^*})) + c(E(P_{v^*\!,u^*})) \ \le\  \textstyle 2 D_W.
\end{aligned}
\end{equation}

We now analyze the cost of $F'$ in step~3 of Algorithm~\ref{algo:reduction_veterbrate_pair}.
Since $F'$ is the output of a $(\kappa, \eta)$-algorithm applied to the vertebrate pair $(\Iscr', B)$,
we have
$
 c(F')\ \le\  \kappa \cdot \lp(\Iscr') + \eta \cdot \sum_{L\in\Lscr_{\bar B}} 2y_{\{v_L\}}.
$
Using
 $
  \sum_{L\in\Lscr_{\bar B}} 2y_{\{v_L\}} = \sum_{L\in\Lscr_{\bar B}} (2y_L + D_L)
$
 and
\begin{equation*}
 \lp(\Iscr') \ \le \ D_W + \sum_{L\in\Lscr, L\subsetneq W, L\cap V(B)\neq \emptyset} 2y_L 
                    + \sum_{L \in \Lscr_{\bar B}} 2y_{\{v_L\}}
 \end{equation*}
(where we used that we set $y_{W} := \frac{D_W}{2}$ in step~1 of the algorithm if $W\ne V$),
this implies
\begin{equation}\label{eq:cost_F_step3}
\begin{aligned}
 c(F')\ \le\ & \kappa \cdot D_W + \sum_{L\in\Lscr, L\subsetneq W, L\cap V(B)\neq \emptyset} \kappa\cdot 2y_L 
               +  \sum_{L\in\Lscr_{\bar B}} (\kappa + \eta) \cdot (2y_L + D_L)
\end{aligned}
\end{equation}
at the end of step~3.
As in \cite{SveTV18}, the lifting and all the amendments of $F'$ in step~4 do not increase the cost of $F'$
by Lemma~\ref{lemma:cost_of_nice_path} and the choice of the values $y_{\{v_L\}}$ in step~2 and $y_{W}$ in step~1.
(Here we use that whenever a Eulerian walk passes through $v_{\bar W}$, we leave and enter $W$.)

To bound the cost increase in step~5 we apply the induction hypothesis.
Adding the edges resulting from a single recursive call of Algorithm~\ref{algo:reduction_veterbrate_pair} in step~5
for some $L\in\Lscr_{\bar B}$ increases the cost by at most
$c(F_L) \le (2\kappa+2)\cdot\val(L) + (\kappa+\eta) (\val(L) - D_L)$.
Using~\eqref{eq:cost_F_step3}, we  obtain the following bound:
\begin{align*}
 c(F'')\ \le\ & \kappa \cdot D_W + \sum_{L\in\Lscr, L\subsetneq W, L\cap V(B)\neq \emptyset} \kappa\cdot 2y_L \\
                 &   +  \sum_{L\in\Lscr_{\bar B}}
                 \left(\big.(2\kappa+2)\cdot\val(L) + (\kappa + \eta) \cdot (2y_L + \val(L))\right)\\
        \le\ & \kappa \cdot D_W + \kappa \cdot\val(W) \\
            & +  \sum_{L\in\Lscr_{\bar B}}
              \left(\big.(\kappa+2)\cdot\val(L) + (\kappa + \eta) \cdot (2y_L + \val(L)) \right)\\
        \le\ & \kappa \cdot D_W + \kappa \cdot \val(W) \\
           & + (\kappa+2)\cdot(\val(W) -D_W)  + (\kappa + \eta) \cdot (\val(W) - D_W) \\
        =\ & (2\kappa+ 2) \cdot \val(W) - 2 \cdot D_W + (\kappa + \eta) \cdot (\val(W) - D_W),
\end{align*}
where we used the definition of $\Lscr_{\bar B}$ for the second inequality
and Lemma~\ref{lemma:L_bar_W_small} for the third inequality;
note that the elements of $\Lscr_{\bar B}$ are pairwise disjoint.
Together with \eqref{eq:cost_bound_backbone} this implies the claimed bound on $c(F)$.
\end{proof}
Now we prove the main result of this section.
\begin{theorem}\label{thm:atspwithvertebratepairs}
Let $\kappa,\eta \ge 0$. 
Suppose we have a polynomial-time $(\kappa,\eta)$-algorithm for vertebrate pairs.
Then there is a polynomial-time algorithm that computes a solution of cost at most 
\begin{equation*}
 3\kappa+\eta+2
\end{equation*}
times the value of \eqref{eq:atsp_lp} for any given ATSP instance.
\end{theorem}

\begin{proof}
By Theorem~\ref{thm:stronglylaminar} it suffices to show that there is a polynomial-time algorithm that computes a solution of cost at most 
$ (3\kappa+\eta+2) \cdot\lp(\Iscr)$ for any given  strongly laminar ATSP instance $\Iscr$.
Given such an instance, we apply Algorithm~\ref{algo:reduction_veterbrate_pair} to $W=V$.
By Lemma~\ref{lemma:algo_feasible_reduction_vertebrate_pairs} and Lemma~\ref{lemma:atsp_sttour_insubsets}, this algorithm computes 
in polynomial time a tour of cost at most 
\begin{align*}
c(F) \ \le \ &\left(2\kappa+2 \right)\cdot\val(V) + (\kappa+\eta)\cdot(\val(V)-D_V) \\
      \ =\ & \left(2\kappa+2 \right)\cdot\lp(\Iscr) + (\kappa+\eta)\cdot(\lp(\Iscr)-D_V) \\
      \ \le\ & \left(3\kappa+\eta+2\right) \cdot \lp(\Iscr).
      \qedhere
\end{align*}
\end{proof}

In the following we will present a $(2,14+\epsilon)$-algorithm for vertebrate pairs, improving on the $(2,37+\epsilon)$-algorithm by 
Svensson, Tarnawski, and V\'egh~\cite{SveTV18}.
Using their vertebrate pair algorithm, Theorem~\ref{thm:atspwithvertebratepairs} immediately implies a $(45+\epsilon)$-approximation algorithm for ATSP.

\begin{remark}\label{remark:improve_approx_ratio}
One could achieve a slightly better overall approximation ratio for ATSP by the following modifications.
Change Algorithm~\ref{algo:reduction_veterbrate_pair} and generalize the notion of vertebrate pairs as follows.
First, in the definition of a vertebrate pair allow that the backbone is not necessarily Eulerian but 
could also be an $s$-$t$-path for some $s,t\in V$.
In this case the solution for the vertebrate pair would again be an Eulerian multi-set $F$ of edges 
such that $(V,E(B) \cupp F)$ is connected; then $E(B) \cupp F$ is an $s$-$t$-tour.
The algorithm for vertebrate pairs that we will describe in later sections extends to this more general version.

Then fix a constant $\delta \in [0,1]$ depending on $\kappa$ and $\eta$ and 
change step~5 of Algorithm~\ref{algo:reduction_veterbrate_pair} as follows.
If in step~4 we added a $u$-$v$-path $P_{u,v}$ in $G[L]$ with $D_L(u,v) \ge (1-\delta) \cdot D_L$ for a set $L\in \Lscr_{\bar B}$,
then we use this path as a backbone in the recursive call of Algorithm~\ref{algo:reduction_veterbrate_pair}
instead of constructing a new backbone.
This saves the cost $2D_L$ of the backbone in the recursive call, but we also pay some additional cost.
Because the total $y$-weight of the sets in $\Lscr$ visited by $P$ is not $D_L$ (as with the old choice of the backbone) but slightly less, 
we obtain a worse bound in Lemma~\ref{lemma:L_bar_W_small}.
If for a set $L\in \Lscr_{\bar B}$ we did not add a path in $G[L]$ of length at least $(1-\delta) \cdot D_L$ in step~4,
then we do not change the recursive call of Algorithm~\ref{algo:reduction_veterbrate_pair} in step~5.
In this case we gain because the bound on the cost of the edges that we added in step~4 is not tight.

Optimizing $\delta$ depending on $\kappa$ and $\eta$ leads to an improvement of the overall approximation ratio.
However, the improvement is small. 
We will later show that there is a polynomial-time $(2,14 +\epsilon)$-algorithm for vertebrate pairs for any fixed $\epsilon > 0$.
For $\kappa = 2$ and $\eta > 14$ the improvement is less than $0.2$, and it is less than $1$ for any $\kappa$ and~$\eta$.
\end{remark}

\section{Computing subtour covers}\label{sect:subtour_cover}

Very roughly, the algorithm that Svensson, Tarnawski, and V\'egh~\cite{SveTV18} use 
for vertebrate pairs follows the cycle cover approach by Frieze, Galbiati and Maffioli~\cite{FriGM82}. 
The algorithm by Frieze, Galbiati and Maffioli always maintains an Eulerian \mbox{(multi-)}set $H$ of edges and repeatedly computes 
another Eulerian \mbox{(multi-)}set $F$ of edges that enters and leaves every connected component of $(V,H)$
at least once. Then it adds the edges of $F$ to $H$ and iterates until $(V,H)$ is connected.

In order to achieve a constant approximation ratio, the algorithm for vertebrate pairs and its analysis are much more involved.
The main algorithm is essentially due to Svensson~\cite{Sve15}, and we describe an improved version of this algorithm 
in Section~\ref{sect:solve_vertebrate_pairs}.

In this section we discuss a sub-routine called by Svensson's  algorithm.
The sub-routine we present here is an improved version of an algorithm by Svensson, Tarnawski, and V\'egh~\cite{SveTV18}.
It computes solutions to the \emph{Subtour Cover} problem, which we define below.
One can view the Subtour Cover problem as the analogue of the cycle cover problem 
that is solved in every iteration of the algorithm by Frieze, Galbiati and Maffioli.
However, we do not only require that the multi-set $F$ of edges that we compute is Eulerian and enters and leaves every connected
component of $(V,H)$, but require in addition that every component of $(V,F)$ 
that crosses the boundary of a set $L\in\Lscr_{\ge 2}$ is connected to the backbone $B$.
Recall that $\Lscr_{\ge 2} = \{ L \in \Lscr : |L| \ge 2\}$.

\begin{definition}\label{def:subtour_cover}
An instance of \emph{Subtour Cover} consists of a vertebrate pair $(\Iscr,B)$ with $\Iscr=(G,\Lscr,x,y)$
and a multi-subset $H$ of $E[V\setminus V(B)]$ such that
\begin{itemize}
 \item $(V,H)$ is Eulerian, and
 \item $H\cap \delta(L) = \emptyset$ for all $L\in\Lscr_{\ge 2}$.
\end{itemize}

A solution to such an instance $(\Iscr,B, H)$ is a multi-set $F$ of edges such that the following three conditions 
are fulfilled:
\begin{itemize}
 \item[(i)] $(V,F)$ is Eulerian.
 \item[(ii)] $\delta(W)\cap F\ne\emptyset$ for all vertex sets $W$ of connected components of $(V\setminus V(B),H)$.
 \item[(iii)] If for a connected component $D$ of $(V, F)$ there is a set $L\in \Lscr_{\ge 2}$ with $E(D)\cap \delta(L) \ne \emptyset$,
       then $V(D) \cap V(B) \ne \emptyset$.
\end{itemize}
\end{definition}
Subtour Cover is very similar to the notions of Subtour Partition Cover from \cite{SveTV18} and Local Connectivity ATSP from \cite{Sve15}. 
The difference between instances of Subtour Cover and Subtour Partition Cover 
is that we require that $H\cap \delta(L) = \emptyset$ for all $L\in\Lscr_{\ge 2}$ in Definition~\ref{def:subtour_cover}.
Moreover, a solution for Subtour Partition Cover is not required to fulfill condition~$(iii)$.
However, the instances to which Svensson, Tarnawski, and V\'egh apply their algorithm for Subtour Partition Cover
also fulfill the definition of Subtour Cover and the solutions computed by this algorithm also fulfill condition~$(iii)$.
We include these properties explicitly in Definition~\ref{def:subtour_cover} because 
we will exploit them for some improvement in Svensson's algorithm (see Section~\ref{sect:solve_vertebrate_pairs}).

For the analysis of Svensson's algorithm for vertebrate pairs it is not sufficient to have only a bound on the total
cost of a solution to Subtour Cover.
In this section we explain an algorithm that computes solutions to Subtour Cover 
that fulfill certain ``local'' cost bounds.
More precisely, the goal of this section is to show the following theorem, 
where we write $y_v := y_{\{v\}}$ if $\{v\}\in \Lscr$ and $y_v := 0$ otherwise.

\begin{theorem}
\label{thm:subtourpartitioncover}
There is a polynomial-time algorithm for Subtour Cover that computes for every instance $(\Iscr,B,H)$ a solution $F$ such that
\begin{equation}\label{eq:backbone_component_light1}
 c(F) \ \le \ 2\cdot\lp(\Iscr)+ \sum_{v\in V\setminus V(B)} 2y_{v},
\end{equation}
and for every connected component $D$ of $(V,F)$ with $V(D)\cap V(B)=\emptyset$ we have
\begin{equation}\label{eq:def_light1}
   c(E(D)) \ \le \ 3 \cdot \sum_{v\in V(D)} 2 y_{v}.
\end{equation} 
\end{theorem}
Svensson, Tarnawski, and V\'egh~\cite{SveTV18} proved a similar statement, but instead of~\eqref{eq:def_light1}
they showed the weaker bound $ c(E(D)) \ \le \ 4 \cdot \sum_{v\in V(D)} 2 y_{v}$.

The reason why we need bounds on the cost of single connected components rather than the total Subtour Cover solution is the following.
When Svensson's algorithm computes a solution $F$ to Subtour Cover, it does not include all edges of $F$ in the tour
that it computes but only those edges that are part of some carefully selected connected components of $(V,F)$.

In the rest of this section we prove Theorem~\ref{thm:subtourpartitioncover}.
We first give a brief outline.

\subsection{Outline}

Let $W_1, \dots, W_k$ be the vertex sets of the connected components of $(V\setminus V(B), H)$.
To find a solution $F$ that fulfills the properties $(i)$ and $(ii)$ we would like to find an integral circulation $x^*$ in $G$
that satisfies $x^*(\delta(W_i)) \ge 2$ for $i=1,\dots,k$.
Note that $x$ is a fractional circulation with this property.
However, if we include the constraints $x^*(\delta(W_i)) \ge 2$ in the linear program describing a minimum-cost circulation problem,
we will in general not obtain an integral optimum solution.
Svensson~\cite{Sve15} suggested the following.
We can introduce new vertices $a_i$ for $i=1,\dots,k$ and reroute one unit of flow going through the set $W_i$ through the new vertex $a_i$.
Then we can add constraints $x^*(\delta^-(a_i)) = 1$ to our flow problem and maintain integrality.
After solving the minimum-cost circulation problem, we can map the one unit of flow through $a_i$ back to some flow entering and leaving $W_i$ 
(with some small additional cost).

The bound \eqref{eq:backbone_component_light1} is obtained by minimizing the total cost of the circulation.
The most difficult properties to achieve are $(iii)$ and \eqref{eq:def_light1}.
If we have $(iii)$, it is relatively easy to obtain a bound of a similar form as \eqref{eq:def_light1} (with some other constant):
we can add constraints of the form $x^*(\delta^-(v)) \le \lceil x(\delta^-(v)) \rceil$ to our minimum-cost circulation problem.
Because of $(iii)$ and the definition of the induced cost function $c$, this implies a bound similar to  \eqref{eq:def_light1}.

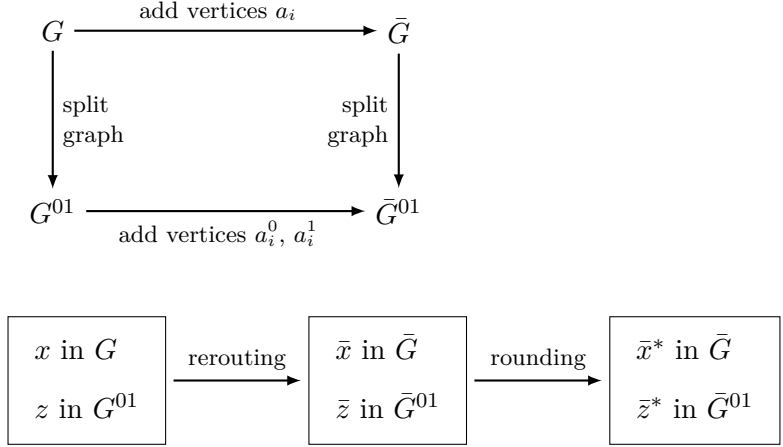
\begin{figure}
\begin{center}
 \begin{tikzpicture}
  \begin{scope}[shift={(1,1.3)}]
  \node (G) at (-0.3,2.4) {$G$};
  \node (bG) at (4.3,2.4) {$\bar G$};
  \node (Gs) at (-0.3,0) {$G^{01}$};
  \node (bGs) at (4.3,0) {$\bar G^{01}$};
  \begin{scope}[thick, >=latex]
     \draw[->] (G) -- (bG);
     \draw[->] (G) -- (Gs);
     \draw[->] (bG) -- (bGs);
     \draw[->] (Gs) -- (bGs);
  \end{scope}
  \node[above]  () at (1.9,2.4) {\footnotesize add vertices $a_i$};
  \node[below]  () at (1.9,0) {\footnotesize add vertices $a_i^0$, $a_i^1$};
  \node[right] () at (-0.3,1.4) {\footnotesize split};
  \node[right] () at (-0.3,1.0) {\footnotesize graph};
  \node[left] () at (4.3,1.4) {\footnotesize split};
  \node[left] () at (4.3,1.0) {\footnotesize graph};
\end{scope}
 \node[draw, align=left, anchor=north west,text height=6ex, inner sep=10pt, minimum width = 2, outer sep =3pt] (a) at (0,0) {$x$ in $G$\\[3mm] $z$ in $G^{01}$};
 \node[draw, align=left, anchor=north west,text height=6ex, inner sep=10pt, minimum width = 2, outer sep =3pt] (b) at (4,0) {$\bar x$ in $\bar G$\\[3mm] $\bar z$ in $\bar G^{01}$};
 \node[draw, align=left, anchor=north west,text height=6ex, inner sep=10pt, minimum width = 2, outer sep =3pt] (c) at (8,0) {$\bar x^*$ in $\bar G$\\[3mm] $\bar z^*$ in $\bar G^{01}$};
 \draw[->, thick, >=latex] (a) -- node[above] {\footnotesize rerouting} (b);
 \draw[->, thick, >=latex] (b) -- node[above] {\footnotesize rounding} (c);
 \end{tikzpicture}
\end{center}
\caption{\label{fig:overview_circulations}
Overview of the different graphs and circulations occurring in the proof of Theorem~\ref{thm:subtourpartitioncover}.
The integral circulation $\bar x^*$ corresponds to an Eulerian (multi) edge set $\bar F$ in~$\bar G$.
}
\end{figure}

To achieve property~$(iii)$, Svensson, Tarnawski, and V\'egh~\cite{SveTV18} introduced the concept of the \emph{split graph}.
This graph contains two copies of every vertex of the original graph $G$.
Every Eulerian edge set in the split graph can be projected to an Eulerian edge set in the original graph $G$.
The crucial property of the split graph is that every cycle that contains an edge corresponding to $e\in \delta_G(L)$ 
for some $L\in \Lscr_{\ge 2}$ also contains a copy of a backbone vertex $v\in V(B)$.
Therefore, if we round a circulation in the split graph (and then project the solution back to $G$),
we will automatically fulfill property~$(iii)$.

While every circulation in the split graph can be projected to a circulation in the original graph $G$,
we cannot lift any arbitrary circulation in $G$ to a circulation in the split graph.
However, Svensson, Tarnawski, and V\'egh~\cite{SveTV18} showed that this is possible for every solution $x$ to \eqref{eq:atsp_lp}.
For this, they use a so-called \emph{witness flow}.
We will choose the witness-flow with a certain minimality condition to achieve the bound \eqref{eq:def_light1}, improving on the 
Subtour Cover algorithm from~\cite{SveTV18}.
To obtain the improved bound we also choose the flow that is rerouted through the auxiliary vertices $a_i$ more carefully.

Because we cannot lift an arbitrary circulation in $G$ to a circulation in the split graph $G^{01}$ of $G$,
we proceed in the following order.
First, we lift the circulation $x$ to a circulation $z$ in the split graph $G^{01}$.
Then we add the auxiliary vertices $a_i$ to $G$ and add the two corresponding copies $a_i^0$ and $a_i^1$ to the split graph $G^{01}$.
In the resulting split graph $\bar G^{01}$ we reroute flow through the new auxiliary vertices $a_i^0, a_i^1$
and round our fractional circulation to an integral one.
See Figure~\ref{fig:overview_circulations}.

We now explain our algorithm in detail.

\subsection{The split graph}\label{sect:split_graph}

In this section we explain the concept of the \emph{split graph} due to Svensson, Tarnawski, and V\'egh (in earlier versions of \cite{SveTV18}).
This is an important tool for achieving property $(iii)$ of a solution to Subtour Cover.
This property will also be crucial in the proof of~\eqref{eq:def_light1}.
For defining the split graph, we number the non-singleton elements of our laminar family $\Lscr$ as follows.
Number $\Lscr_{\ge 2}\cup\{V\}=\{L_1,\ldots,L_{r_{\max}}\}$ such that $|V| =  |L_1|\ge\cdots\ge |L_{r_{\max}}|\ge 2$. 
Let $r(v):=\max\{i:v\in L_i\}$, and call an edge
$e=(v,w)\in E$ \emph{forward} if $r(v)<r(w)$, \emph{backward} if $r(v)>r(w)$, and 
\emph{neutral} if $r(v)=r(w)$. See Figure~\ref{fig:define_forward_backwards}.

\begin{figure}
\begin{center}
 \begin{tikzpicture}[scale=0.7]
  \tikzset{BackwardEdge/.style={->, >=latex, line width=1pt,darkgreen}}
  \tikzset{ForwardEdge/.style={->, >=latex, line width=1pt,darkred}}
  \tikzset{NeutralEdge/.style={->, >=latex, line width=1pt,gray}}
 
  \fill[blue, opacity=0.2] (3.6,4) ellipse (3.2 and 2.7) {};
  \fill[white, opacity=1]  (2.5,5) ellipse (1.2 and 0.7) {};
  \fill[white, opacity=1] (4,3.1) ellipse (1.7 and 1.2) {};
  
  \node[blue, opacity=0.6] () at (1.3,0.5) {$r(v) = 2$};

  \tikzset{L/.style={thick}}
  \draw[L] (7,4) ellipse (7 and 3.8) {};
  \node[below=1pt] () at (7,7.8) {$V=L_1$};
  
  \draw[L] (3.6,4) ellipse (3.2 and 2.7) {};
  \node[below=1pt] () at (3.6,6.7) {$L_2$};
  
  \draw[L] (10.5,4) ellipse (3 and 2.5) {};
  \node[below=1pt] () at (10.5,6.5) {$L_3$};
  
  \draw[L] (2.5,5) ellipse (1.2 and 0.7) {};
  \node[below=1pt] () at (2.5,5.7) {$L_6$};
  
  \draw[L] (4,3.1) ellipse (1.7 and 1.2) {};
  \node[below=1pt] () at (4,4.3) {$L_4$};
  
  \draw[L] (3.3,3) ellipse (0.5 and 0.5) {};
  \node[below=1pt] () at (3.3,3.5) {$L_9$};
  
  \draw[L] (10,4.5) ellipse (1.5 and 1.0) {};
  \node[below=1pt] () at (10,5.5) {$L_5$};
  
  \draw[L] (4.8,3) ellipse (0.5 and 0.5) {};
  \node[below=1pt] () at (4.8,3.5) {$L_{11}$};
  
  \draw[L] (12,3) ellipse (0.6 and 0.6) {};
  \node[below=1pt] () at (12,3.5) {$L_8$};
  
  \draw[L] (10.5, 4.2) ellipse (0.5 and 0.5) {};
  \node[below=1pt] () at (10.5,4.7) {$L_{10}$};
  
  \draw[L] (10.3,2.5) ellipse (0.8 and 0.6) {};
  \node[below=1pt] () at (10.3,3.1) {$L_{7}$};
  
 \end{tikzpicture}
\end{center}
\caption{The laminar family $\Lscr \cup \{V\} = \{L_1, \dots, L_{11}\}$.
In this example, the set $L_2 \setminus (L_6 \cup L_4)$ is the set of all vertices $v$ with $r(v)=2$;
it is shown in blue.}
\label{fig:define_forward_backwards}
\end{figure}
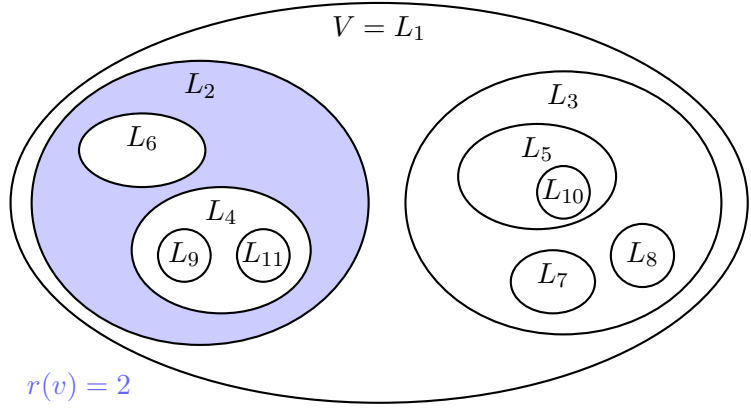
We will need the following simple observation about cycles in $G$.
A \emph{cycle} is a connected digraph in which every vertex has in-degree and out-degree exactly 1.
\begin{lemma}\label{lemma:simple_berservation_froward_edges}
 Let $C$ be the edge set of a cycle.
 If there exists a set $L\in \Lscr_{\ge 2}$ with $C \cap \delta(L) \ne \emptyset$, then $C$ contains a forward edge and a backward edge. 
\end{lemma}
\begin{proof}
 Because $C$ is Eulerian there exists an edge $e=(v,w) \in C \cap \delta^+(L)$.
 By the choice of the numbering $L_1,\ldots,L_{r_{\max}}$, we have $L_{r(v)} \subseteq L$ and hence $w\notin L_{r(v)}$.
 Therefore, the cycle with edge set $C$ contains vertices $v,w$ with $r(v) \ne r(w)$.
 Hence, $C$ contains both a forward and a backward edge. 
\end{proof}
Next we define the \emph{split graph} $G^{01}$ of $G$ and extend the cost function $c$ to it. 
\begin{itemize}
\item For every vertex $v\in V$ it contains two vertices $v^0$ and $v^1$ (on the lower and upper level).
\item For every $v\in V$ it contains an edge $e_v^{\downarrow}=(v^1,v^0)$ with $c(e_v^{\downarrow})=0$.
\item For every $v\in V(B)$ it also contains an edge $e_v^{\uparrow}=(v^0,v^1)$ with $c(e_v^{\uparrow})=0$.
\item For every forward edge $e=(v,w)\in E$, the split graph contains an edge $e^0=(v^0,w^0)$ with $c(e^0)=c(e)$.
\item For every backward edge $e=(v,w)\in E$, the split graph contains an edge $e^1=(v^1,w^1)$ with $c(e^1)=c(e)$.
\item For every neutral edge  $e=(v,w)\in E$, the split graph contains edges $e^0=(v^0,w^0)$ and $e^1=(v^1,w^1)$ with $c(e^0)=c(e^1)=c(e)$.
\end{itemize}
We write $V^0 := \{v^0 : v\in V\}$ and call $G^{01}[V^0]$ the \emph{lower level} of the split graph $G^{01}$. 
Similarly, we write $V^1 := \{v^1 : v\in V\}$ and call $G^{01}[V^1]$ the \emph{upper level}  of $G^{01}$.  
For a set $W \subseteq V$ let
$W^{01}:=\{v^j: v\in W,\, j\in\{0,1\}\}$ be the vertex set of $G^{01}$ that corresponds to $W$.

For any subgraph of $G^{01}$ we obtain a subgraph of $G$ (its \emph{image}) by replacing both $v^0$ and $v^1$ by $v$ and removing loops. 
Then, obviously, the image of a cycle is an Eulerian graph.
The next lemma shows how we can use the split graph to achieve property $(iii)$ of a solution to Subtour Cover. 

\begin{lemma}
\label{lemma:circuitsconnecttobackbone}
If the image of a cycle in $G^{01}$ contains an edge $e\in\delta(L)$ for some $L\in\Lscr_{\ge 2}$, it also contains a vertex of $B$.
\end{lemma}
\begin{proof}
Let $C^{01}$ be a cycle in $G^{01}$ such that its image $C$ (an Eulerian subgraph of $G$) contains an edge $e\in \delta(L)$ for some $L\in\Lscr_{\ge 2}$.
By Lemma~\ref{lemma:simple_berservation_froward_edges}, $C$ contains a forward edge and a backward edge.
Therefore $C^{01}$ visits both levels of $G^{01}$ and thus contains an edge $e_v^{\uparrow}$ for some $v\in V(B)$.
\end{proof}

\subsection{Witness flows}
We now want to map $x$ to a circulation $z$ in the split graph $G^{01}$.
To this end, we define a flow $f \le x$, which we will call a  \emph{witness flow}.
In the construction of $z$, we will map the witness flow $f$ to the lower level of $G^{01}$
and map the remaining flow $x-f$ to the upper level of $G^{01}$.
See Figure~\ref{fig:construction_of_z}.

\begin{definition}[witness flow]\label{def:witness_flow}
Let $x'$ be a circulation in $G$. Then we call a flow $f'$ in $G$ a \emph{witness flow (for $x'$)} if 
\begin{enumerate}[(a)]
\item $f'(e)=0$ for every backward edge $e$;
\item $f'(e)=x'(e)$ for every forward edge $e$;
\item $0\le f'(e)\le x'(e)$ for every neutral edge $e$; and
\item $f'(\delta^+(v))\ge f'( \delta^-(v))$ for all $v\in V\setminus V(B)$.
\end{enumerate}
\end{definition}
The concept of witness flow was introduced in~\cite{SveTV18}.
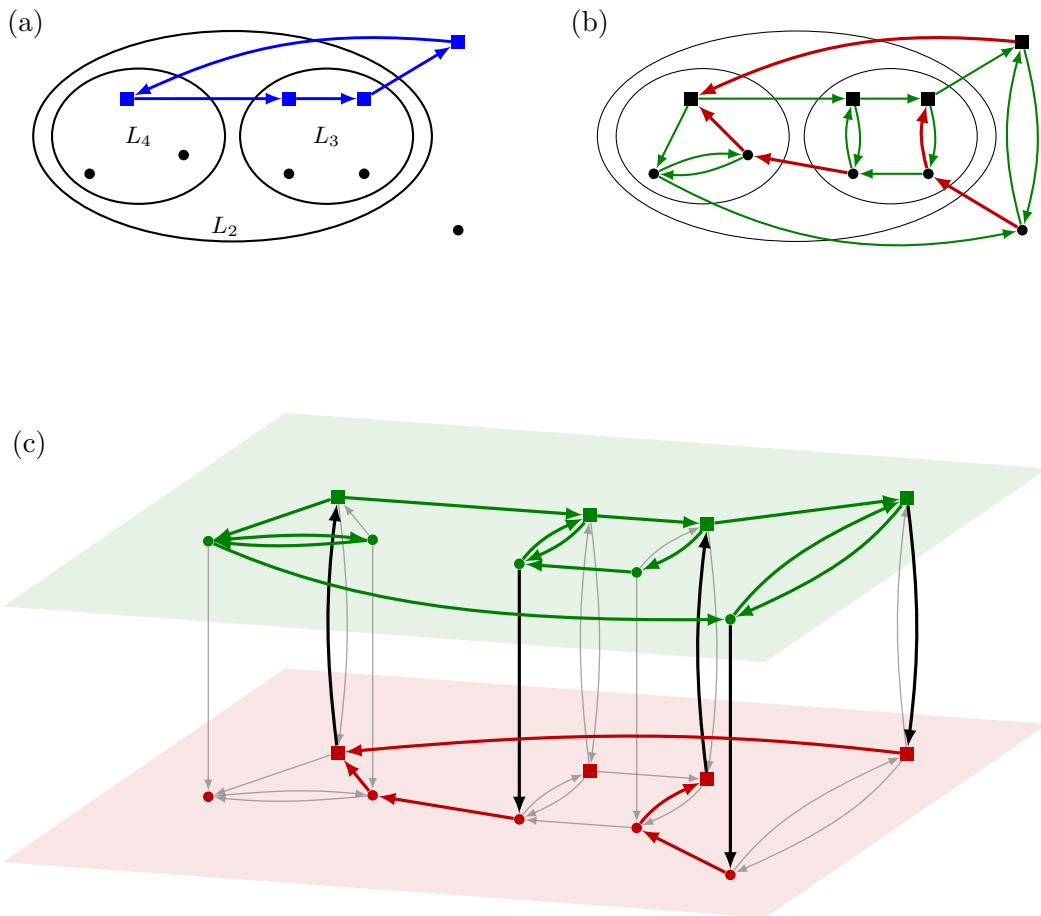
\begin{figure}
\begin{center}
 \begin{tikzpicture}[yscale=-0.5, xscale=0.5]
\tikzset{vertex/.style={
fill=black, circle,inner sep=0em,minimum size=4pt
}}
\tikzset{Bvertex/.style={fill=black, inner sep=0em,minimum size=5pt }}
\tikzset{Edge/.style={->, >=latex, line width=1.1pt,blue}}
\def \bs{15pt}

\node () at (0,1) {(a)};

\node () at (3,4) {\footnotesize $L_4$};
\node () at (8,4) {\footnotesize $L_3$};
\node () at (5.3,6.4) {\footnotesize $L_2$};
\begin{scope}[every node/.style={vertex}]
 \node (v1) at (1.7,5) {};
 \node (v2) at (4.2,4.5) {};
 \node (v4) at (7,5) {};
 \node (v5) at (9,5) {};
 \node (v8) at (11.5,6.5) {};
\end{scope}
 \node[Bvertex, blue] (v3) at (2.7,3) {};
 \node[Bvertex, blue] (v6) at (7,3) {};
 \node[Bvertex, blue] (v7) at (9,3) {};
 \node[Bvertex, blue] (v9) at (11.5,1.5) {};

\draw[thick, black] (3,4) ellipse (2.3 and 1.8);
\draw[thick, black] (8,4) ellipse (2.3 and 1.8);
\draw[thick, black] (5.5,4) ellipse (5.3 and 2.8);

\begin{scope}[Edge]
 \draw (v3) to (v6);
 \draw (v6) to (v7);
 \draw (v9) to[bend left=\bs] (v3);
 \draw (v7) to (v9);
\end{scope}


\begin{scope}[shift={(15,0)}]

\node () at (0,1) {(b)};
\tikzset{Edge/.style={->, >=latex, line width=0.8pt,darkgreen}}
\tikzset{FlowEdge/.style={->, >=latex, line width=1.2pt,darkred}}
\def \bs{15pt}

\begin{scope}[every node/.style={vertex}]
 \node (v1) at (1.7,5) {};
 \node (v2) at (4.2,4.5) {}; 
 \node (v4) at (7,5) {};
 \node (v5) at (9,5) {}; 
 \node (v8) at (11.5,6.5) {};
 \end{scope}
 \node[Bvertex] (v3) at (2.7,3) {};
 \node[Bvertex] (v6) at (7,3) {};
 \node[Bvertex] (v7) at (9,3) {};
\node[Bvertex] (v9) at (11.5,1.5) {};

\draw[black] (3,4) ellipse (2.3 and 1.8);
\draw[black] (8,4) ellipse (2.3 and 1.8);
\draw[black] (5.5,4) ellipse (5.3 and 2.8);

	\begin{scope}[FlowEdge]
	\draw (v4) to (v2);
	\draw (v2) to (v3);

	\draw (v5) to[bend right=\bs] (v7);
	
	\draw (v8) to (v5);
	\draw (v9) to[bend left=\bs] (v3);
	\end{scope}
	
	\begin{scope}[Edge]
	\draw (v1) to[bend right=\bs] (v2);
	\draw (v2) to[bend right=\bs] (v1);
	\draw (v3) to (v1);
	\draw (v3) to (v6);
	\draw (v5) to (v4);
	
	\draw (v4) to[bend right=\bs] (v6);
	\draw (v6) to[bend right=\bs] (v4);
	\draw (v6) to (v7);
	\draw (v7) to[bend right=\bs] (v5);
	
	\draw (v1) to[bend left=20pt] (v8);;
	\draw (v7) to (v9);
	\draw (v8) to[bend right=\bs] (v9);
	\draw (v9) to[bend right=\bs] (v8);
	\end{scope}
 
\end{scope}
 \end{tikzpicture}
 \vspace*{20mm}
 
   \begin{tikzpicture}[scale = 0.85]
   \tikzset{vertex/.style={fill=black, circle,inner sep=0em,minimum size=4pt }}
   \tikzset{Bvertex/.style={fill=black, inner sep=0em,minimum size=5pt }}
   \tikzset{Edge/.style={->, >=latex, line width=0.5pt,gray, opacity=0.7}}
   \tikzset{BlackEdge/.style={->, >=latex, line width=1.2pt,black, opacity=1.0}}
   \tikzset{RedEdge/.style={->, >=latex, line width=1.2pt,black, opacity=1.0}}
   \tikzset{LightRedEdge/.style={->, >=latex, line width=0.5pt,gray, opacity=0.7}}
   \tikzset{UpperFlowEdge/.style={->, >=latex, line width=1.2pt,darkgreen}}
   \tikzset{LowerFlowEdge/.style={->, >=latex, line width=1.2pt,darkred}}
   \def \bs{15pt}
   \def \rt{10}

   \node () at (-4,3.5) {(c)};
   
      \begin{scope}[canvas is xz plane at y=0, rotate=\rt]  
       \draw[fill = darkred, fill opacity=0.1, draw=none] (0,0) rectangle +(13,8);
       \begin{scope}[every node/.style={vertex, darkred}]
	  \node (w1) at (1.7,5) {};
	  \node (w2) at (4.2,4.5) {};	  
	  \node (w4) at (7,5) {};
	  \node (w5) at (9,5) {};	  
	  \node (w8) at (11.5,6.5) {};
	  \end{scope}
	  \begin{scope}[every node/.style={Bvertex, darkred}]
	  \node (w3) at (2.7,3) {};
	  \node (w6) at (7,3) {};
	  \node (w7) at (9,3) {};
	  \node (w9) at (11.5,1.5) {};
	  \end{scope}
      \end{scope}

        \begin{scope}[canvas is xz plane at y=4, rotate=\rt]
        \draw[fill = darkgreen, fill opacity=0.1, draw=none] (0,0) rectangle +(13,8);
        
         \begin{scope}[every node/.style={vertex, darkgreen}]
	\node (v1) at (1.7,5) {};
	\node (v2) at (4.2,4.5) {};	
	\node (v4) at (7,5) {};
	\node (v5) at (9,5) {};
	\node (v8) at (11.5,6.5) {};
	\end{scope}
        \begin{scope}[every node/.style={Bvertex, darkgreen}]
	\node (v3) at (2.7,3) {};
	\node (v6) at (7,3) {};
	\node (v7) at (9,3) {};
	\node (v9) at (11.5,1.5) {};
	\end{scope}   
      \end{scope} 
      
      \begin{scope}[Edge]
        \foreach \i in {1,2,5} {
          \draw[Edge] (v\i) to (w\i);
        }
        \foreach \i in {4,8} {
          \draw[BlackEdge] (v\i) to (w\i);
        }
        \foreach \i in {3,6,7} {
          \draw[Edge] (v\i) to[bend left=7pt] (w\i);
        }
        \foreach \i in {3,7} {
          \draw[RedEdge] (w\i) to[bend left=7pt] (v\i);
        }
        \draw[LightRedEdge] (w6) to[bend left=7pt] (v6);
        \draw[LightRedEdge] (w9) to[bend left=7pt] (v9);
        \draw[BlackEdge] (v9) to[bend left=7pt] (w9);
      \end{scope}

      \begin{scope}[canvas is xz plane at y=0, rotate=\rt]  
	  \begin{scope}[Edge]
	  \draw (w1) to[bend right=\bs] (w2);
	  \draw (w2) to[bend right=\bs] (w1);
	  \draw (w3) to (w1);
	  
	  \draw (w5) to (w4);
	  \draw (w6) to (w7);
	  
	  \draw (w4) to[bend right=\bs] (w6);
	  \draw (w6) to[bend right=\bs] (w4);
	  
	  \draw (w7) to[bend right=\bs] (w5);
	  
	  \draw (w8) to[bend right=\bs] (w9);
	  \draw (w9) to[bend right=\bs] (w8);
	  \end{scope}
	  
	  \begin{scope}[LowerFlowEdge]
	  \draw (w2) to (w3);
	  \draw (w4) to (w2);
	  
	  \draw (w5) to[bend right=\bs] (w7);
	  
	  \draw (w8) to (w5);
	  \draw (w9) to[bend left=\bs] (w3);
	  \end{scope}
      \end{scope}

      \begin{scope}[canvas is xz plane at y=4, rotate=\rt]

	\begin{scope}[Edge]

        \draw (v2) to (v3);
	\draw (v5) to[bend right=\bs] (v7);
	
	\end{scope}
	
	\begin{scope}[UpperFlowEdge]
	\draw (v1) to[bend right=\bs] (v2);
	\draw (v2) to[bend right=\bs] (v1);
	\draw (v3) to (v1);
	
	\draw (v3) to (v6);
	\draw (v6) to (v7);
	
	\draw (v5) to (v4);
	
	\draw (v4) to[bend right=\bs] (v6);
	\draw (v6) to[bend right=\bs] (v4);
	
	\draw (v7) to[bend right=\bs] (v5);
	
	\draw (v1) to[bend left=20pt] (v8);;
	\draw (v7) to (v9);
	\draw (v8) to[bend right=\bs] (v9);
	\draw (v9) to[bend right=\bs] (v8);
	\end{scope}
              
      \end{scope} 
      
    \end{tikzpicture} 
    \end{center}
    \caption{An example of the construction of the circulation $z$ in $G^{01}$.
    Picture (a) shows the laminar family $\Lscr_{\ge 2}=\{L_2, L_3, L_4\}$ and in blue the backbone $B$.
    Picture (b) shows a solution $x$ to \eqref{eq:atsp_lp} where we have $x_e = \sfrac{1}{2}$ for all 
    edges; a witness flow $f$ is shown in red. The vertices in $V(B)$ are shown as squares.
    Every cycle crossing the boundary of a set $L\in\Lscr_{\ge 2}$ contains both a green and a red edge.
    Picture (c) shows the resulting circulation $z$ in $G^{01}$, where we have $z_e >0$ for every thick edge $e$ and 
    and $z_e = 0$ for all thin edges.
    The green vertices are those on the upper level of the split graph; the red vertices are those on the lower level.
    The flow $x-f$ is mapped to the upper level (green) and the flow $f$ is mapped to the lower level (red).
    \label{fig:construction_of_z}}
\end{figure} 
We now show that the pairs $(x',f')$ where $f'$ is a witness flow for the circulation $x'$ in $G$,
correspond to circulations in the split graph $G^{01}$.
\begin{lemma}
Let $z'$ be a circulation in $G^{01}$. 
Define $\pi(z') := (x',f')$ where $x',f'$ are flows in $G$ defined by
\begin{itemize}
 \item $x'(e) := z'(e^0) + z'(e^1)$, and
 \item $f'(e) := z'(e^0)$,
\end{itemize}
where we set $z'(e^1):=0$ for forward edges $e$ and $z'(e^0):=0$ for backward edges $e$.
Then $x'$ is a circulation in $G$ with $c(x')=c(z')$ and $f'$ is a witness flow for $x'$.
\end{lemma}
\begin{proof}
$(a)$ holds because for a backward edge $e$, the graph $G^{01}$ does not contain an edge $e^0$.
Similarly, $(b)$ holds because for a forward edge $e$, the graph $G^{01}$ does not contain an edge $e^1$.
Property $(c)$ is obvious by construction and $(d)$ holds because for $v\in V\setminus V(B)$ the split graph 
does not contain an edge $e_v^{\uparrow}$.
\end{proof}
Having a circulation $x'$ in $G$ and a witness flow $f'$ for $x'$, we can map $x'$ to a circulation $z'$ in $G^{01}$ 
with $\pi(z')=(x',f')$ as follows:
\begin{itemize}
 \item For every edge $e^0$ of the lower level of $G^{01}$ we set $z'(e^0) = f'(e)$.
 \item For every edge $e^1$ of the upper level of $G^{01}$ we set $z'(e^1) = x'(e) - f'(e)$.
 \item For every edge $e_v^{\uparrow}$ (for $v\in V(B)$) we set $z(e_v^{\uparrow})=\max\{0,f'(\delta^-(v))-f'(\delta^+(v))\}$.
 \item For every edge $e_v^{\downarrow}$ (for $v\in V$) we set $z(e_v^{\downarrow}) = \max\{0,f'(\delta^+(v))-f'(\delta^-(v))\}$. 
\end{itemize}
Notice that $x'(e) = z'(e^0)$ for every forward edge $e$ and $x'(e) = z'(e^1)$ for every backward edge $e$.
Moreover, $x'(e) = z'(e^0) + z'(e^1)$ for every neutral edge $e$.
Furthermore, $z'$ indeed defines a circulation in $G^{01}$ because $f'(\delta^+(v))\ge f'(\delta^-(v))$ for all $v\in V\setminus V(B)$.
    
The following was already proved in~\cite{SveTV18}. Here we give a simpler proof.
\begin{lemma}
\label{lemma:witnessflow}
Let $(\Iscr,B)$ be a vertebrate pair, with $\Iscr=(G,\Lscr,x,y)$.
Then there exists a witness flow $f$ for $x$.
\end{lemma}
\begin{proof}
Consider $G'$, which arises from $G$ by adding a new vertex $a$ and
edges $(a,v)$ for all $v\in V$ and edges $(v,a)$ for all $v\in V(B)$.
Set $l(e')=0$ and $u(e')=\infty$ for the new edges.
Moreover, for $e\in E$ set the lower bound $l(e)$ and the upper bound $u(e)$ according to the requirements from Definition~\ref{def:witness_flow}, 
i.e.\ set $u(e) =x(e)$ if $e$ is a forward or neutral edge and $u(e)=0$ otherwise and set 
$l(e)=x(e)$ if $e$ is a forward edge and $l(e) =0$ otherwise.

Then we are looking for a circulation $f'$ in $G'$ with $l\le f'\le u$.
By Hoffman's circulation theorem, this exists if
\begin{equation}
\label{eq:atsp_witness_hoffman}
l(\delta^-(U)) \ \le \ u(\delta^+(U))
\end{equation}
for all $U\subseteq V\cup\{a\}$. 
We show that this is indeed true.
Suppose not, and let $U$ be a minimal set violating \eqref{eq:atsp_witness_hoffman}. 
Since \eqref{eq:atsp_witness_hoffman} obviously holds whenever $a\in U$ or $B\cap U\not=\emptyset$, we have $U\subseteq V\setminus V(B)$.
Let $i$ be the largest index so that $U\cap L_i \not=\emptyset$. 
See Figure~\ref{fig:witness_flow_exists}.\\[2mm]
\textbf{Case 1:} $U\setminus L_i\not=\emptyset$. \\
Then (by the minimality of $U$) we have
$l(\delta^-(U\cap L_i)) \ \le \ u(\delta^+(U\cap L_i))$ and $l(\delta^-(U\setminus L_i)) \ \le \ u(\delta^+(U\setminus L_i))$.
Since all edges from $U\setminus L_i$ to $U\cap L_i$ are forward edges and
all edges from $U\cap L_i$ to $U\setminus L_i$ are backward edges, we get
\begin{align*}
l(\delta^-(U)) + x(\delta^+(U\setminus L_i)\cap\delta^-(U\cap L_i))
\ &=\ l(\delta^-(U\cap L_i)) + l(\delta^-(U\setminus L_i)) \\
\ &\le \ u(\delta^+(U\cap L_i)) + u(\delta^+(U\setminus L_i)) \\
\ &=\ u(\delta^+(U)) + x(\delta^+(U\setminus L_i)\cap\delta^-(U\cap L_i))
\end{align*}
and hence \eqref{eq:atsp_witness_hoffman}, which is a contradiction to the choice of $U$. \\[2mm]
\textbf{Case 2:} $U\subseteq L_i$.\\
Then $r(u)=i$ for all $u \in U$ and $r(w) \ge i$ for all $w\in L_i$.
Hence $l(\delta^-(U)) \le x(\delta^-(L_i)\cap\delta^-(U))$ because we have $l(e) >0$ only for forward edges and
all edges in $\delta^-(U) \setminus \delta^-(L_i)$ are neutral or backward edges.
Moreover, edges in $\delta^+(U) \setminus \delta^+(L_i)$ are not backward edges, implying
$x(\delta^+(U) \setminus \delta^+(L_i)) = u(\delta^+(U) \setminus \delta^+(L_i)) \le u(\delta^+(U))$.
Therefore,
\begin{align*}
l(\delta^-(U)) \ &\le \ x(\delta^-(L_i)\cap\delta^-(U)) \\
\ &= \ x(\delta^-(L_i))+x(\delta^+(U)\setminus\delta^+(L_i)) - x(\delta^-(L_i\setminus U)) \\
\ &\le \ x(\delta^-(L_i))+u(\delta^+(U)) - x(\delta^-(L_i\setminus U)).
\end{align*}
Since $L_i\setminus U\not=\emptyset$ (because $L_i\cap V(B)\not=\emptyset=U\cap V(B)$), we have $x(\delta^-(L_i\setminus U))\ge 1$.
Moreover, $L_i \in \Lscr \cup \{V\}$ implies $x(\delta(L_i)) \in \{0,2\}$ and hence $x(\delta^-(L_i))\le 1$.
Hence \eqref{eq:atsp_witness_hoffman} follows, which is again a contradiction.
\end{proof}

\begin{figure}
\begin{center}
 \begin{tikzpicture}[scale=0.6]
  \tikzset{BackwardEdge/.style={->, >=latex, line width=1pt,darkgreen}}
  \tikzset{ForwardEdge/.style={->, >=latex, line width=1pt,darkred}}
  \tikzset{NeutralEdge/.style={->, >=latex, line width=1pt,gray}}
  \tikzset{Bvertex/.style={fill=black, inner sep=0em,minimum size=5pt }}
  \tikzset{Bedge/.style={->, >=latex, line width=1.1pt,blue}}

\begin{scope}
\node at (-3.5,8) {Case 1:};

 \node[Bvertex, blue] (b1) at (2.5,6.5) {};
 \node[Bvertex, blue] (b2) at (4,7.8) {};
 \node[Bvertex, blue] (b3) at (5.5,6.5) {};
 \node[Bvertex, blue] (b4) at (4,5.2) {};

\begin{scope}[Bedge]
 \draw (b1) to (b2);
 \draw (b2) to (b3);
 \draw (b3) to (b4);
 \draw (b4) to (b1);
\end{scope}
\node[blue] at (4,7) {$B$};

\node[Bvertex,violet,circle] at (0,8) {};
\node[violet] at (-0.5,8) {$a$};
 
\tikzset{L/.style={thick}}
\draw[L] (2.5,4) ellipse (2.5 and 3) {};
\node at (4.5,1.2) {$L_i$};
  
\draw[brown] (0,4) ellipse (4 and 2) {};
\fill[brown, opacity=0.1] (0,4) ellipse (4 and 2) {};
\node[brown] at (-3,2) {$U$};
\node[brown] at (-2,4) {$U\setminus L_i$};
\node[brown] at (2,4) {$U\cap L_i$};

 \draw[BackwardEdge] (1.5,5) to node[above=0mm] {\small backward} (-1.5,5);
 \draw[ForwardEdge] (-1.5,3) to node[below=0mm] {\small forward} (1.5,3);
  
\end{scope}   

\begin{scope}[shift={(12,0)}]
\node at (-1.5,8) {Case 2:};

 \node[Bvertex, blue] (b1) at (4,6) {};
 \node[Bvertex, blue] (b2) at (5.5,7.3) {};
 \node[Bvertex, blue] (b3) at (7,6) {};
 \node[Bvertex, blue] (b4) at (5.5,4.7) {};

\begin{scope}[Bedge]
 \draw (b1) to (b2);
 \draw (b2) to (b3);
 \draw (b3) to (b4);
 \draw (b4) to (b1);
\end{scope}
\node[blue] at (5.5,6.5) {$B$};

\node[Bvertex,violet,circle] at (2,8) {};
\node[violet] at (1.5,8) {$a$};
 
\tikzset{L/.style={thick}}
\draw[L] (2.3,4) ellipse (4 and 3) {};
\node at (5.2,1.2) {$L_i$};
  
\draw[brown] (0.8,4) ellipse (2 and 2.3) {};
\fill[brown, opacity=0.1] (0.8,4) ellipse (2 and 2.3) {};
\node[brown] at (0,3) {$U$};

 \draw[BackwardEdge] (4,4.8) to node[above=0mm] {\small backward} (1,4.8);
 \draw[NeutralEdge] (4,4.2) to node[below=-0.5mm] {\small neutral} (1,4.2);
 \draw[NeutralEdge] (1,3.4) to (4,3.4);
 \draw[ForwardEdge] (1,2.8) to node[below=0mm] {\small forward} (4,2.8);
  
\end{scope}    
  
 \end{tikzpicture}
\end{center}
\caption{Proof of Lemma~\ref{lemma:witnessflow} (Case 1 and Case 2).
\label{fig:witness_flow_exists}}
\end{figure}
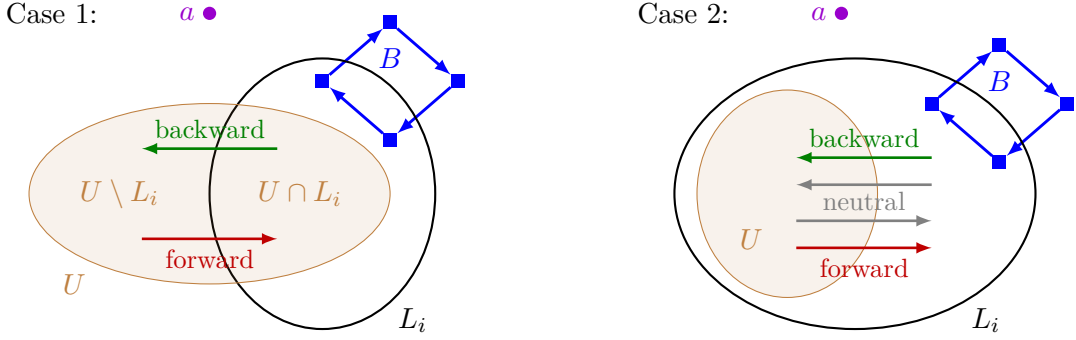

Working with an arbitrary witness flow $f$ is sufficient to obtain a constant-factor approximation for ATSP
and this is essentially what Svensson, Tarnawski, and V\'egh did.
However, to obtain a better approximation ratio we will not work with an arbitrary witness flow $f$, 
but will choose $f$ with some additional properties.
Recall that $W_1, \dots, W_k$ are the vertex sets of the connected components of $(V\setminus V(B), H)$.

\begin{lemma}\label{lemma:improved_witness_flow}
 We can compute in polynomial time a witness flow $f$ for $x$ such that
  \begin{itemize}
  \item[(e)] the support of $f$ is acyclic, and 
  \item[(f)] $\sum_{i=1}^k f(\delta(W_i)) \le \sum_{i=1}^k f'(\delta(W_i))$ for every witness flow $f'$ for $x$.
 \end{itemize}
\end{lemma}
\begin{proof}
We first compute a witness flow $\tilde f$ by minimizing $\sum_{i=1}^k f(\delta(W_i))$ subject to the constraints $(a)-(d)$ from Definition~\ref{def:witness_flow}.
This linear program is feasible by Lemma~\ref{lemma:witnessflow}.
Then the flow $\tilde f$ fulfills property $(f)$.

To compute the flow $f$ we minimize $\sum_{e\in E} f(e)$ subject to the constraints $(a)-(d)$ and $f(e) \le \tilde f(e)$ for all $e\in E$.
This linear program is feasible because $\tilde f$ is a feasible solution.
Then $f$ is a witness flow for $x$ with $\sum_{i=1}^k f(\delta(W_i)) \le \sum_{i=1}^k \tilde f(\delta(W_i))$.
Since the flow $\tilde f$ fulfills property $(f)$, the same holds for the flow $f$.

Suppose $f$ does not fulfill $(e)$, i.e.\ $f$ is not acyclic.
Then there is a cycle $C\subseteq E$ with $f(e) >0$ for all $e\in C$.
As $f$ fulfills $(a)$, the set $C$ does not contain any backward edge.
This implies that $C$ also contains no forward edge because $C$ is a cycle.
Let $\epsilon := \min_{e\in C} f(e)$.
For $e\in E$ we set $f'(e) := f(e) - \epsilon \le \tilde f(e)$ if $e\in C$ and $f'(e) := f(e) \le \tilde f(e)$ otherwise.
Because $C$ contains neither forward nor backward edges, $f'$ fulfills $(a)$ and $(b)$.
By the choice of $\epsilon$, we have $f'(e) \ge 0$ for all $e\in E$, implying $(c)$.
Finally, $f'(\delta^+(v)) - f'( \delta^-(v)) = f(\delta^+(v)) - f( \delta^-(v)) \ge 0$ for all $v\in V \setminus  V(B)$,
where we used that $C$ is a cycle and $f$ fulfills $(d)$.
This shows that $f'$ is a witness flow and $f'(e) \le \tilde f(e)$ for all $e\in E$, 
but $\sum_{e\in E} f'_e < \sum_{e\in E} f_e$, a contradiction to the choice of~$f$.
\end{proof}

\subsection{Rerouting and rounding}
Recall that the sets $W_1, \dots, W_k$ are the vertex sets of the connected components of $(V\setminus V(B), H)$.
Thus they are pairwise disjoint subsets of $V\setminus V(B)$. 
\begin{lemma}\label{lemma:r_values_consistent}
 Let $i\in\{1,\ldots,k\}$ and $v, w\in W_i$. Then $r(v)=r(w)$.
\end{lemma}
\begin{proof}
For all $L\in\Lscr_{\ge 2}$ we have $H\cap \delta(L) = \emptyset$ and therefore $W_i\subseteq L$ or $W_i\cap L=\emptyset$.
This implies $r(v) = \max \{j : v\in L_j\} = \max \{j : w\in L_j\} = r(w)$.
\end{proof}

We will now work with a flow $f$ as in Lemma~\ref{lemma:improved_witness_flow}.
Let $G_f$ denote the residual graph of the flow $f$ and the graph $G$ with edge capacities $x$.
So for every edge $e=(v,w)\in E$ with $f(e) < x(e)$, the residual graph contains an edge $(v,w)$ with residual capacity 
$u_f((v,w)) = x(e) - f(e)$.
For every edge $e=(v,w) \in E$ with $f(e) >0$ the residual graph contains an edge $(w,v)$ with residual capacity
$u_f((w,v))=f(e)$. Parallel edges can arise.

We will transform the graph $G$ into another graph $\bar G$.
The circulation $z$ in $G^{01}$ will be transformed into a circulation $\bar z$ in the split graph $\bar G^{01}$ of $\bar G$.
We construct $\bar G$ from $G$ by doing the following for $i=1,\ldots,k$.

We add an auxiliary vertex $a_i$ to $G$ and set $r(a_i):= r(v)$ for $v\in W_i$; this is well-defined by Lemma~\ref{lemma:r_values_consistent}.
Let $\hat W_i$ be the vertex set of the first strongly connected component of $G_f[W_i]$ in some topological order. 
For every edge $(v,w)\in \delta^-(\hat W_i)$ we add an edge $(v,a_i)$ of the same cost.
Similarly, for every edge $(v,w)\in \delta^+(\hat W_i)$ we add an edge $(a_i,w)$ of the same cost.
Note that then a new edge is a forward/backward/neutral edge if and only if its corresponding edge in $G$ is forward/backward/neutral.
Then the split graph $\bar G^{01}$ of $\bar G$ contains new vertices $a_i^0$ and $a_i^1$, 
connected by an edge $e^{\downarrow}_{a_i} =(a_i^1, a_i^0)$ of cost zero.
Let $\bar G$ the graph resulting from $G$ by the modifications described above and let $\bar G^{01}$ be its split graph.

We will now reroute some of the flow $z$ going through $\hat W_i$ such that it goes through one of the new vertices $a_i^0, a_i^1$.
See Figure~\ref{fig:rerouting_flow}.
We need the following lemma, where $\chi^F \in \mathbb{Z}_{\ge 0}^E$ denotes the incidence vector of $F$ for any multi-subset $F$ of $E$.
\begin{lemma}\label{lemma:construct_flow_rerouting}
 Let $G'$ be a directed graph and $z'$ a circulation in $G'$.
 Let $U\subseteq V(G')$ with $z'(\delta(U))\ge 2$.
 Then we can compute in polynomial time a multiset $\Pscr$ of paths in $G'[U]$ and 
 for every $P\in \Pscr$ starting in $s\in U$ and ending in $t\in U$
 \begin{itemize}\itemsep0pt
  \item a weight $\lambda_P > 0$,
  \item an edge $e^\textsf{in}_P = (s',s)\in \delta^-(U)$, and
  \item an edge $e^\textsf{out}_P = (t,t')\in \delta^+(U)$,
 \end{itemize}
 such that $ \sum_{P\in \Pscr} \lambda_P = 1$ and
 \[ 
   \sum_{P\in \Pscr} \lambda_P \cdot \left( \chi^{e^\textsf{in}_P} + \chi^{E(P)} + \chi^{e^\textsf{out}_P}  \right) \ \le\ z'.
 \] 
\end{lemma}
\begin{proof}
 We contract $V(G')\setminus U$ to a vertex $v_\textsf{outside}$.
 Then $z'(\delta(v_\textsf{outside})) = z'(\delta(U))\ge 2$.
 Because $z'$ remains a circulation, we can compute in polynomial time a set $\Cscr$ of cycles containing $v_\textsf{outside}$ and weights
 $\lambda_C > 0$ for $C\in \Cscr$ with $\sum_{C\in\Cscr} \lambda_C = 1$ such that 
 \[ \sum_{C\in\Cscr} \lambda_C \cdot \chi^{E(C)} \le z'. \]
 After undoing the contraction, every cycle $C$ results in an edge $e^\textsf{in}= (s',s) \in \delta^-(U)$, an edge $e^\textsf{out}= (t,t')\in \delta^+(U)$,
 and an $s$-$t$-path $P$ in $G'[U]$. 
\end{proof}
We construct a circulation $\bar z$ in $\bar G^{01}$ from $z$ by doing the following for $i=1,\ldots,k$.
We apply Lemma~\ref{lemma:construct_flow_rerouting} to the vertex set $U=\hat W_i^{01}$.
We partition the resulting set $\Pscr$ into sets $\Pscr^0$ and $\Pscr^1$ such that
$\Pscr^0$ contains the paths $P\in\Pscr$ for which $e^\textsf{in}_P$ is contained in the lower level of the split graph
and $\Pscr^1$ contains the paths $P\in\Pscr$ for which $e^\textsf{in}_P$ is contained in the upper level of the split graph.
Since $\sum_{P\in \Pscr} \lambda_P = 1$, we have $\sum_{P\in \Pscr^q} \lambda_P \ge \sfrac{1}{2}$ for some $q\in\{0,1\}$.
We can choose values $0 \le \lambda_P' \le \lambda_P$ such that $\sum_{P\in \Pscr^q} \lambda_P' = \sfrac{1}{2}$.
To obtain $\bar z$ from $z$, we do the following
for every $P\in \Pscr^q$:
\begin{itemize}\itemsep0pt
 \item We decrease the flow on $e^\textsf{in}_P$ and increase the flow on its corresponding edge in $\delta^-(a_i^q)$ by $\lambda_P'$.
 \item We decrease the flow on every edge $e\in E(P)$ by $\lambda_P'$.
 \item Let $p=0$ if $e^\textsf{out}_P$ is contained in the lower level of the split graph and $p=1$ otherwise.
       We decrease the flow on $e^\textsf{out}_P$ and increase the flow on its corresponding edge in $\delta^+(a_i^p)$ by $\lambda_P'$.
 \item Because $W_i \cap V(B) = \emptyset$, the path $P$ contains no edge from the lower to the upper level;
       hence $p \le q$.
       If $p<q$, i.e.\ $q=1$ and $p=0$, we increase the flow on $e_{a_i}^{\downarrow}$ by $\lambda_P'$.
\end{itemize}
Note that we maintain a circulation in the split graph $\bar G^{01}$.

Let $\bar z$ be the circulation in $\bar G^{01}$ resulting from $z$.
Note that $c(\bar z)\le c(z)$.
Moreover, $\bar z$ is a circulation such that for every $i\in\{1,\ldots,k\}$
we have $\bar z(\delta^-(a^0_i))=\sfrac{1}{2}$ or $\bar z(\delta^-(a^1_i))=\sfrac{1}{2}$.
Because we could only reroute $\sfrac{1}{2}$ unit of flow through $a_i^0$ or $a_i^1$,  we consider the circulation $2\bar z$.
\begin{figure}
\begin{center}
 \begin{tikzpicture}[scale=0.6]
 
  \tikzset{vertex/.style={fill=black, circle,inner sep=0em,minimum size=4pt }}
  \tikzset{IntegralEdge/.style={->, >=latex, line width=1.5pt, orange, opacity=1.0}}
  \tikzset{PathCompletionEdge/.style={->, >=latex, line width=1.5pt,blue, opacity=1.0}}
  \tikzset{UpperFlowEdge/.style={->, >=latex, line width=0.8pt,darkgreen}}
  \tikzset{LowerFlowEdge/.style={->, >=latex, line width=0.8pt,darkred}}
 
  \begin{scope}[shift={(2.5,9)}]
   \node () at (0,7) {(a)};
   

   \draw[darkblue] (3,3.5) ellipse (2 and 2.8); 
   \draw[fill=darkblue, draw=darkblue, fill opacity=0.1] (2.5,3.5) ellipse (1 and 2.2); 
   
   \begin{scope}[every node/.style={vertex}]
     \node (v1) at (2.5,5) {};
     \node (v2) at (2.5,2) {};
     \node (v3) at (4,4.5) {};
     \node (v4) at (4,2.5) {};
     
     \node (v5) at (1,7) {};
     \node (v7) at (3,7) {};
     
     \node (v14) at (5,7) {};
     \node (v8) at (6,6) {};
     \node (v9) at (6,4.5) {};
     
     \node (v10) at (6,2) {};
     \node (v11) at (5,1) {};
     
     \node (v13) at (1,1) {};
     \node (v12) at (1,0) {};
     
     \node (v5b) at (0.5,5.5) {};
     \node (v12b) at (2.5,0) {};
   \end{scope}
   
   \begin{scope}[LowerFlowEdge]
    \draw (v3) to (v8);
    \draw[double]  (v4) to (v2);
    \draw (v2) to[bend left=20pt] (v1);
    \draw[double] (v1) to (v7);
    \draw[double] (v11) to (v4);
    \draw[double] (v2) to (v13);
   \end{scope}
   
   \begin{scope}[UpperFlowEdge]
    \draw (v5) to (v1);
    \draw (v5b) to (v1);
    \draw[double] (v1) to (v3);
    \draw[double] (v9) to (v3);
    \draw (v3) to[bend right=10pt] (v14);
    \draw[double] (v3) to (v4);
    \draw (v2) to[bend right=20pt] (v1);
    \draw [double] (v4) to (v10);
    \draw (v12) to[bend right=20pt] (v2);
    \draw (v12b) to (v2);
   \end{scope}
   
  \end{scope}

  \begin{scope}[shift={(12,9)}]
   \node () at (-0.5,7) {(b1)};
     
   \draw[darkblue] (3,3.5) ellipse (2 and 2.8); 
   \draw[fill=darkblue, draw=darkblue, fill opacity=0.1] (2.5,3.5) ellipse (1 and 2.2); 
   
   \begin{scope}[every node/.style={vertex}]
     \node[blue, minimum size=5pt] (ai) at (0, 3.5) {};
   
     \node (v1) at (2.5,5) {};
     \node (v2) at (2.5,2) {};
     \node (v3) at (4,4.5) {};
     \node (v4) at (4,2.5) {};
     
     \node (v5) at (1,7) {};
     \node (v7) at (3,7) {};
     
     \node (v14) at (5,7) {};
     \node (v8) at (6,6) {};
     \node (v9) at (6,4.5) {};
     
     \node (v10) at (6,2) {};
     \node (v11) at (5,1) {};
     
     \node (v13) at (1,1) {};
     \node (v12) at (1,0) {};
     
     \node (v5b) at (0.5,5.5) {};
     \node (v12b) at (2.5,0) {};
   \end{scope}
   
   \begin{scope}[LowerFlowEdge]
    \draw (v3) to (v8);
    \draw[double]  (v4) to (v2);
    \draw (v2) to[bend left=20pt] (v1);
    \draw[double] (v1) to (v7);
    \draw[double] (v11) to (v4);
    \draw[double] (v2) to (v13);
   \end{scope}
   
   \begin{scope}[UpperFlowEdge]
    \draw (v5) to (v1);
    \draw (v5b) to (ai);
    \draw[double] (ai) to (v3);
    \draw[double] (v9) to (v3);
    \draw (v3) to[bend right=10pt] (v14);
    \draw[double] (v3) to (v4);
    \draw [double] (v4) to (v10);
    \draw (v12b) to (v2);
    \draw (v12) to[bend left=20pt] (ai);
   \end{scope}
  \end{scope}
  
  \begin{scope}[shift={(0,0)}]
   \node () at (-0.5,7) {(b2)};
     
   \draw[darkblue] (3,3.5) ellipse (2 and 2.8); 
   \draw[fill=darkblue, draw=darkblue, fill opacity=0.1] (2.5,3.5) ellipse (1 and 2.2); 
   
   \begin{scope}[every node/.style={vertex}]
   \node[blue, minimum size=5pt] (ai) at (0, 3.5) {};
   
     \node (v1) at (2.5,5) {};
     \node (v2) at (2.5,2) {};
     \node (v3) at (4,4.5) {};
     \node (v4) at (4,2.5) {};
     
     \node (v5) at (1,7) {};
     \node (v7) at (3,7) {};
     
     \node (v14) at (5,7) {};
     \node (v8) at (6,6) {};
     \node (v9) at (6,4.5) {};
     
     \node (v10) at (6,2) {};
     \node (v11) at (5,1) {};
     
     \node (v13) at (1,1) {};
     \node (v12) at (1,0) {};
         
     \node (v5b) at (0.5,5.5) {};
     \node (v12b) at (2.5,0) {};
   \end{scope}
   
   \begin{scope}[LowerFlowEdge]
    \draw (v3) to (v8);
    \draw[double]  (v4) to (v2);
    \draw (v2) to[bend left=20pt] (v1);
    \draw[double] (v1) to (v7);
    \draw[double] (v11) to (v4);
    \draw (v2) to (v13);
    \draw (ai) to (v13);
   \end{scope}
   
   \begin{scope}[UpperFlowEdge]
    \draw (v5) to (v1);
    \draw (v1) to (v3);
    \draw (v5b) to (ai);
    \draw (ai) to (v3);
    \draw[double] (v9) to (v3);
    \draw (v3) to[bend right=10pt] (v14);
    \draw[double] (v3) to (v4);
    \draw (v2) to[bend right=20pt] (v1);
    \draw [double] (v4) to (v10);
    \draw (v12b) to (v2);
    \draw (v12) to[bend left=20pt] (ai);
   \end{scope}
  \end{scope}
  
    \begin{scope}[shift={(8,0)}]
   \node () at (-0.5,7) {(c)};
   
   \draw[darkblue] (3,3.5) ellipse (2 and 2.8); 
   \draw[fill=darkblue, draw=darkblue, fill opacity=0.1] (2.5,3.5) ellipse (1 and 2.2); 
   
   \begin{scope}[every node/.style={vertex}]
   \node[blue, minimum size=5pt] (ai) at (0, 3.5) {};
   
     \node (v1) at (2.5,5) {};
     \node (v2) at (2.5,2) {};
     \node (v3) at (4,4.5) {};
     \node (v4) at (4,2.5) {};
     
     \node (v5) at (1,7) {};
     \node (v7) at (3,7) {};
     
     \node (v14) at (5,7) {};
     \node (v8) at (6,6) {};
     \node (v9) at (6,4.5) {};
     
     \node (v10) at (6,2) {};
     \node (v11) at (5,1) {};
     
     \node (v13) at (1,1) {};
     \node (v12) at (1,0) {};
     
     \node (v5b) at (0.5,5.5) {};
     \node (v12b) at (2.5,0) {};
   \end{scope}
   
   \begin{scope}[IntegralEdge]
    \draw (v5b) to (ai);
    \draw (ai) to (v13);
   \end{scope}

  \end{scope}
  
    \begin{scope}[shift={(15.5,0)}]
   \node () at (0,7) {(d)};
   
      \draw[darkblue] (3,3.5) ellipse (2 and 2.8); 
   \draw[fill=darkblue, draw=darkblue, fill opacity=0.1] (2.5,3.5) ellipse (1 and 2.2); 
   
   \begin{scope}[every node/.style={vertex}]
     \node (v1) at (2.5,5) {};
     \node (v2) at (2.5,2) {};
     \node (v3) at (4,4.5) {};
     \node (v4) at (4,2.5) {};
     
     \node (v5) at (1,7) {};
     \node (v7) at (3,7) {};
     
     \node (v14) at (5,7) {};
     \node (v8) at (6,6) {};
     \node (v9) at (6,4.5) {};
     
     \node (v10) at (6,2) {};
     \node (v11) at (5,1) {};
     
     \node (v13) at (1,1) {};
     \node (v12) at (1,0) {};
     
     \node (v5b) at (0.5,5.5) {};
     \node (v12b) at (2.5,0) {};
   \end{scope}
   
    \begin{scope}[IntegralEdge]
    \draw (v5b) to (v1);
    \draw (v2) to (v13);
   \end{scope}
   
   \begin{scope}[PathCompletionEdge]
    \draw (v1) to (v3);
    \draw (v3) to (v4);
    \draw (v4) to (v2);
   \end{scope}

  \end{scope} 
 \end{tikzpicture}
 \end{center}
 \caption{
 Example of the construction of the solution $F$ from the witness flow $f$.
 On all pictures, a set $W_i$ (blue with white interior) and the subset $\hat W_i$ (blue and filled) is shown.
 The pictures show only edges with at least one endpoint in $W_i$.
 Picture (a) shows (parts of) a possible solution $x$ to \eqref{eq:atsp_lp} (green and red) and a witness flow $f$ (red). 
 The edges drawn with a single line have value $\sfrac{1}{4}$, the edges drawn with a double line have value $\sfrac{1}{2}$.
 Pictures (b1) and (b2) show two possible circulations $\bar x$ in $\bar G$ that could result from rerouting flow through $a_i$ (blue);
 the witness flow $\bar f$ is shown in red.
 Picture (c) shows in orange an possible integral flow $\bar x^*$ in $\bar G$ that could result if we rerouted flow through 
 $a_i$ as in (b2); The orange edges are elements of the edge set $\bar F$ with $\chi^{\bar F} = \bar x^*$.
 Picture (d) shows the result of mapping $\bar F$ back to $G$. 
 In blue the path $P_i$ in $G[W_i]$ is shown; it completes the orange edges to a circulation.
 \label{fig:rerouting_flow}}
\end{figure}
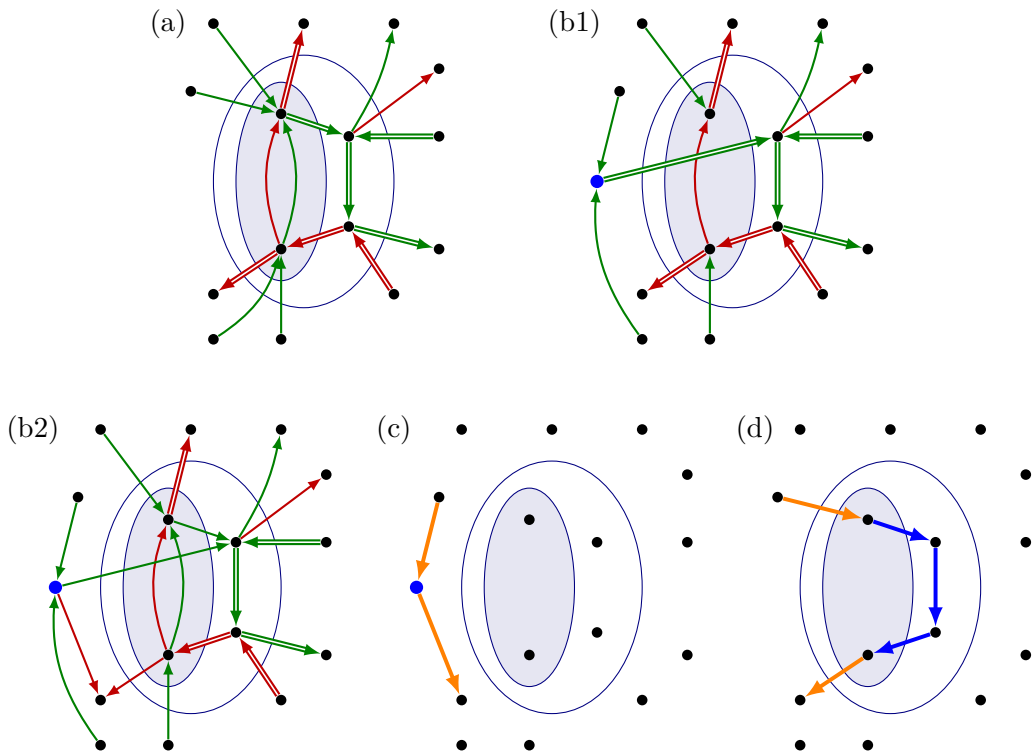

We round $2\bar z$ to an integral circulation:
by Corollary 12.2b of \cite{Sch03}, we can find in polynomial time an integral circulation $\bar z^*$ in $\bar G^{01}$ with 
\begin{enumerate}[(A)]\itemsep0pt
\item\label{item:upper_bound_edges} $0\le \bar z^*(e)\le \lceil 2\bar z(e) \rceil$ for all $e\in E(\bar G^{01})$,
\item $c(\bar z^*) \le c(2\bar z)$,
\item\label{item:flow_bound_through_v_1} $\bar z^*(\delta^-(v^1)) \le \lceil 2\bar z(\delta^-(v^1))\rceil$ for all $v\in V$, and
\item\label{item:flow_one_throu_a_i} for every $i\in\{1,\ldots,k\}$ we have $\bar z^*(\delta^-(a^0_i))=1$ or $\bar z^*(\delta^-(a^1_i))=1$.
\end{enumerate}          

Let $(\bar x, \bar f) :=\pi(\bar z)$ and $(\bar x^*, \bar f^*) :=\pi(\bar z^*)$.
Let $\bar F \subseteq E(\bar G)$ be the multi-set of edges with $\chi^{\bar F} = \bar x^*$; see Figure~\ref{fig:rerouting_flow}~(c).
Then $\bar F$ is Eulerian because $\bar x^*$ is a circulation.

We now show several properties of $\bar F$, before we show how to map $\bar F$ 
to a solution $F$ for Subtour Cover in $G$ (in Section~\ref{sect:map_bar_F_back}).
First we observe 
\begin{equation}\label{eq:cost_bound_bar_F}
 c(\bar F)\ =\ c(\bar x^*)\ =\ c(\bar z^*)\ \le\ 2\cdot c(\bar z)\ \le\ 2\cdot c(z)\ =\ 2\cdot c(x)\ =\ 2\cdot \lp(\Iscr).
\end{equation}
The following lemma will be used in the proof of property~$(ii)$ of Definition~\ref{def:subtour_cover}.
\begin{lemma}\label{lemma:visit_a_i_exatcly_once}
Let $i\in \{1,\dots, k\}$. Then $|\delta^-_{\bar F}(a_i)| = 1$.
\end{lemma}
\begin{proof}
We have 
\[
  |\delta^-_{\bar F}(a_i)| = \bar x^*(\delta^-(a_i)) = \bar z^*(\delta^-(a_i^1)) + \bar z^*(\delta^-(a_i^0) \setminus\{e_{a_i}^{\downarrow}\})
\]
 By property~\eqref{item:flow_one_throu_a_i}, we have $\bar z^*(\delta^-(a_i^0))=1$ or $\bar z^*(\delta^-(a_i^1))=1$.
 Moreover, by property~\eqref{item:upper_bound_edges}, the support of the integral flow $\bar z^*$ is contained in the support of the flow $\bar z$.
 If we have $\bar z^*(\delta^-(a_i^1))=1$, then we have by construction of $\bar z$ that $\bar z^*(e) \le \lceil 2\bar z(e)\rceil =0$ 
 for all $e\in \delta^-(a_i^0)\setminus \{e_{a_i}^{\downarrow}\}$, implying $|\delta^-_{\bar F}(a_i)| = 1$.
 Otherwise,  we have $\bar z^*(\delta^-(a_i^0))=1$ and by construction of $\bar z$ we have $\bar z(\delta^-(a_i^1))=0$ and $\bar z(e_{a_i}^{\downarrow}) = 0$.
 Therefore, by property~\eqref{item:upper_bound_edges} we have $\bar z^*(\delta^-(a_i^1))=0$ and $\bar z^*(e_{a_i}^{\downarrow}) = 0$.
\end{proof}

The proof of the following lemma is where we use our choice of $f$ as in Lemma~\ref{lemma:improved_witness_flow}.
Here, an arbitrary witness flow is not sufficient. 
See Figure~\ref{fig:cycle_free_maintained} (a) -- (b) for an illustration.
\begin{lemma}\label{lemma:cycle_free_maintained}
The flows $\bar f$ and $\bar f^*$ have acyclic support.
\end{lemma}
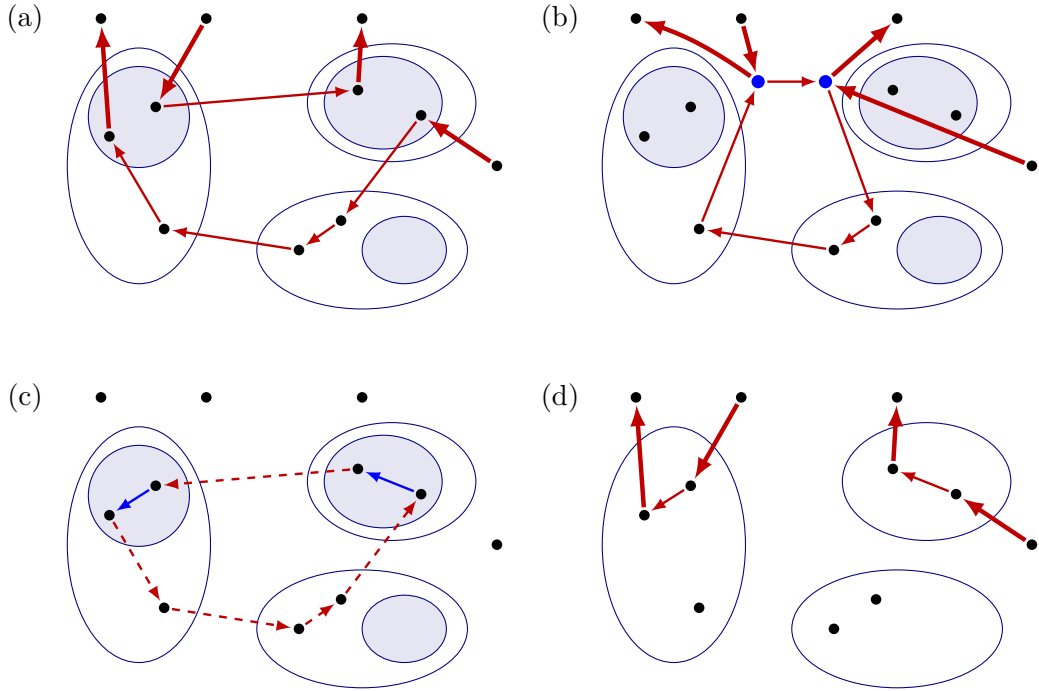
\begin{figure}
\begin{center}
 \begin{tikzpicture}[scale=0.56]
 
  \tikzset{vertex/.style={fill=black, circle,inner sep=0em,minimum size=4pt, outer sep=1pt }}
  \tikzset{PathCompletionEdge/.style={->, >=latex, line width=0.9pt,blue, opacity=1.0}}
  \tikzset{ResidualFlowEdge/.style={->, >=latex, line width=0.9pt,darkred, dashed}}
  \tikzset{LowerFlowEdge/.style={->, >=latex, line width=0.9pt,darkred}}
  \tikzset{ForwardEdge/.style={->, >=latex, line width=1.7pt,darkred}}
 
  \begin{scope}[shift={(0,0)}]
   \node () at (-2,7) {(a)};
      
   \begin{scope}[draw=darkblue]
    \draw  (0.7,3.5) ellipse (1.7 and 2.8); 
    \draw  (6.7, 5) ellipse (2 and 1.4); 
    \draw  (6,1.5) ellipse (2.5 and 1.4); 
   \end{scope}
   
   \begin{scope}[fill=darkblue, draw=darkblue, fill opacity=0.1]
       \fill[draw=darkblue] (0.7,4.66) ellipse (1.2 and 1.2); 
       \fill[draw=darkblue] (6.5,5) ellipse (1.4 and 1.1);
       \fill[draw=darkblue] (7, 1.5) ellipse (1 and 0.8);
   \end{scope}

   \begin{scope}[every node/.style={vertex}]
     \node (v1) at (2.3,7) {};
     \node (v2) at (1.1, 4.9) {};
     \node (v3) at (5.9,5.3) {};
     \node (v4) at (6,7) {};
     \node (v5) at (9.2,3.5) {};
     \node (v6) at (7.4,4.7) {};
     \node (v7) at (5.5,2.2) {};
     \node (v8) at (4.5,1.5) {};
     \node (v9) at (1.3,2) {};
     \node (v10) at (0,4.2) {};
     \node (v11) at (-0.2,7) {};
     
   \end{scope}
   
   \begin{scope}[ForwardEdge]
    \draw (v1) to (v2);
    \draw (v3) to (v4);
    \draw (v5) to (v6);
    \draw (v10) to (v11);
   \end{scope}
   
   \begin{scope}[LowerFlowEdge]
    \draw (v2) to (v3);
    \draw (v6) to (v7);
    \draw (v7) to (v8);
    \draw (v8) to (v9);
    \draw (v9) to (v10);    
   \end{scope}
  
  \end{scope}

  \begin{scope}[shift={(12.7,0)}]
   \node () at (-2,7) {(b)};
      
      \begin{scope}[draw=darkblue]
    \draw  (0.7,3.5) ellipse (1.7 and 2.8); 
    \draw  (6.7, 5) ellipse (2 and 1.4); 
    \draw  (6,1.5) ellipse (2.5 and 1.4); 
   \end{scope}
   
   \begin{scope}[fill=darkblue, draw=darkblue, fill opacity=0.1]
       \fill[draw=darkblue] (0.7,4.66) ellipse (1.2 and 1.2); 
       \fill[draw=darkblue] (6.5,5) ellipse (1.4 and 1.1);
       \fill[draw=darkblue] (7, 1.5) ellipse (1 and 0.8);
   \end{scope}

   \begin{scope}[every node/.style={vertex}]
     \node (v1) at (2.3,7) {};
     \node (v2) at (1.1, 4.9) {};
     \node (v3) at (5.9,5.3) {};
     \node (v4) at (6,7) {};
     \node (v5) at (9.2,3.5) {};
     \node (v6) at (7.4,4.7) {};
     \node (v7) at (5.5,2.2) {};
     \node (v8) at (4.5,1.5) {};
     \node (v9) at (1.3,2) {};
     \node (v10) at (0,4.2) {};
     \node (v11) at (-0.2,7) {};
     
     \node[blue, minimum size=5pt] (a1) at (2.7,5.5) {};
     \node[blue, minimum size=5pt] (a2) at (4.3,5.5) {};
   \end{scope}
   
   \begin{scope}[ForwardEdge]
    \draw (v1) to (a1);
    \draw (a2) to (v4);
    \draw (v5) to (a2);
    \draw[bend right=8] (a1) to (v11);
   \end{scope}
   
   \begin{scope}[LowerFlowEdge]
    \draw (a1) to (a2);
    \draw (a2) to (v7);
    \draw (v7) to (v8);
    \draw (v8) to (v9);
    \draw (v9) to (a1);    
   \end{scope}
  \end{scope}

  \begin{scope}[shift={(0,-9)}]
   \node () at (-2,7) {(c)};
      
      \begin{scope}[draw=darkblue]
    \draw  (0.7,3.5) ellipse (1.7 and 2.8); 
    \draw  (6.7, 5) ellipse (2 and 1.4); 
    \draw  (6,1.5) ellipse (2.5 and 1.4); 
   \end{scope}
   
   \begin{scope}[fill=darkblue, draw=darkblue, fill opacity=0.1]
       \fill[draw=darkblue] (0.7,4.66) ellipse (1.2 and 1.2); 
       \fill[draw=darkblue] (6.5,5) ellipse (1.4 and 1.1);
       \fill[draw=darkblue] (7, 1.5) ellipse (1 and 0.8);
   \end{scope}

   \begin{scope}[every node/.style={vertex}]
     \node (v1) at (2.3,7) {};
     \node (v2) at (1.1, 4.9) {};
     \node (v3) at (5.9,5.3) {};
     \node (v4) at (6,7) {};
     \node (v5) at (9.2,3.5) {};
     \node (v6) at (7.4,4.7) {};
     \node (v7) at (5.5,2.2) {};
     \node (v8) at (4.5,1.5) {};
     \node (v9) at (1.3,2) {};
     \node (v10) at (0,4.2) {};
     \node (v11) at (-0.2,7) {};
     
   \end{scope}
  
  \begin{scope}[PathCompletionEdge]
   \draw (v2) to (v10);
   \draw (v6) to (v3);
  \end{scope}

   \begin{scope}[ResidualFlowEdge]
    \draw (v3) to (v2);
    \draw (v7) to (v6);
    \draw (v8) to (v7);
    \draw (v9) to (v8);
    \draw (v10) to (v9);    
   \end{scope}
  \end{scope}
  
  \begin{scope}[shift={(12.7,-9)}]
   \node () at (-2,7) {(d)};
      
      \begin{scope}[draw=darkblue]
    \draw  (0.7,3.5) ellipse (1.7 and 2.8); 
    \draw  (6.7, 5) ellipse (2 and 1.4); 
    \draw  (6,1.5) ellipse (2.5 and 1.4); 
   \end{scope}
   

   \begin{scope}[every node/.style={vertex}]
     \node (v1) at (2.3,7) {};
     \node (v2) at (1.1, 4.9) {};
     \node (v3) at (5.9,5.3) {};
     \node (v4) at (6,7) {};
     \node (v5) at (9.2,3.5) {};
     \node (v6) at (7.4,4.7) {};
     \node (v7) at (5.5,2.2) {};
     \node (v8) at (4.5,1.5) {};
     \node (v9) at (1.3,2) {};
     \node (v10) at (0,4.2) {};
     \node (v11) at (-0.2,7) {};
   \end{scope}
   
   \begin{scope}[ForwardEdge]
    \draw (v1) to (v2);
    \draw (v3) to (v4);
    \draw (v5) to (v6);
    \draw (v10) to (v11);
   \end{scope}
   
   \begin{scope}[LowerFlowEdge]
    \draw (v2) to (v10);
    \draw (v6) to (v3);  
   \end{scope}
  \end{scope}

 \end{tikzpicture}
 \end{center}
 \caption{Illustration of the proof of Lemma~\ref{lemma:cycle_free_maintained} and 
 the reason why choosing an arbitrary flow $f$ as in Lemma~\ref{lemma:witnessflow} is not sufficient.
 Three sets $W_i$  are shown in blue with white interior;
 pictures (a)--(c) also show their subsets $\hat W_i$ (blue and filled).
 Picture~(a) shows (parts of) a flow $f$ as in Lemma~\ref{lemma:witnessflow} (red); the thick edges show forward edges.
 This flow $f$ will not be chosen by our algorithm; it does not minimize $\sum_{i=1}^k f(\delta(W_i))$.
 Picture~(b) shows what would happen if we chose this flow anyway.
 We see a possible result of rerouting this flow through the vertices $a_i \in V(\bar G)$ (shown in blue).
 In this example, the support of $\bar f$ contains a cycle $\bar C$.
 Picture~(c) shows a corresponding closed walk $C$ in the residual graph $G_f$. 
 The blue edges show paths inside the sets $\hat W_i$; these exist because $G_f[\hat W_i]$ is strongly connected.
 Picture~(d) shows the flow resulting from $f$ by augmenting along $C$.
 The augmentation decreased $\sum_{i=1}^k f(\delta(W_i))$, but did not change the flow on forward edges.
 \label{fig:cycle_free_maintained}}
\end{figure}
\begin{proof}
Since the support of $\bar f^*$ is contained in the support of $\bar f$ by \eqref{item:upper_bound_edges}, it suffices to show that $\bar f$ has acyclic support.
Suppose the support of $\bar f$ contains a cycle $\bar C$.
Then there exists $i\in\{1, \dots, k\}$ such that $a_i \in V(\bar C)$
because otherwise $\bar C$ is contained in the support of $f$ (which is acyclic).
Let $\bar e = (a_i, v)\in E(\bar C)$ and let $e=(u,v)\in \delta^+(\hat W_i)$ be the edge of $G$ corresponding to $\bar e$.
Then $f(e) > 0$ and hence the residual graph $G_f$ contains an edge $(v,u) \in \delta_{G_f}^-(\hat W_i)$.
Therefore $v\notin W_i$ since $\hat W_i$ is the vertex set of the first strongly connected component of $G_f[W_i]$.
This shows $E(\bar C) \cap \delta(W_i \cup \{a_i\}) \ne \emptyset$.

We claim that we can map $\bar C$ to a closed walk $C$ in the residual graph $G_f$.
See Figure~\ref{fig:cycle_free_maintained} (b) -- (c).
We first map every edge of the cycle $\bar C$ to its corresponding edge in $G$. 
Notice that the resulting edge set $F$ is not necessarily a cycle:
if $a_i \in V(\bar C)$ for some $i\in \{1,\dots, k\}$, then $F$ contains an edge entering $\hat W_i$ and an edge leaving 
$\hat W_i$, but might be disconnected in between.

We have $f(e) > 0$ for every edge $e\in F$.
Thus, by reversing all edges in $F$ we obtain edges in $G_f$ (with positive residual capacity $u_f$).
Moreover, we can complete this edge set to a closed walk $C$ in $G_f$ (with positive residual capacity $u_f$) by adding only edges 
of $G_f[\hat W_i]$ for $i\in \{1,\dots, k\}$; 
this is possible because for every $i\in \{1,\dots, k\}$, the subgraph $G_f[\hat W_i]$ is strongly connected by the choice of $\hat W_i$.
We found a closed walk $C$ in $G_f$.
Let $i \in \{1,\dots, k\}$ such that $E(\bar C) \cap \delta(W_i \cup \{a_i\}) \ne \emptyset$.
Then $E(C) \cap \delta(W_i) \ne \emptyset$.

Also note that $r(v)\ge r(w)$ for all $(v,w)\in E(C)$:
every edge $(v,w) \in E(G_f)$ of $C$ has a corresponding edge $(w,v) \in E(G)$ with $f(e) > 0$
or it has both endpoints in the same set $\hat W_i \subseteq W_i$.
In the first case, we can conclude that $(w,v)$ is not a backward edge and hence $r(w) \le r(v)$.
In the latter case, $r(v)=r(w)$ by Lemma~\ref{lemma:r_values_consistent}.
Since $C$ is a closed walk we conclude that $r(v)=r(w)$ for all $v,w \in V(C)$.

This shows that augmenting $f$ along the closed walk $C$ changes flow only on neutral edges.
We augment by some sufficiently small but positive amount and maintain a witness flow.
We claim that this augmentation decreases $\sum_{i=1}^k f(\delta(W_i))$, which contradicts our choice of $f$.
See Figure~\ref{fig:cycle_free_maintained}~(d).
The only edges of $C$ contained in a cut $\delta(W_i)$ for some $i\in \{1,\dots, k\}$ result from mapping the edges of 
the cycle $\bar C$ in $\bar G$ to $G_f$ and reversing them; for these edges the augmentation decreases the flow value.
The other edges that we added to $C$ are contained in some $G_f[\hat W_i]$ for $i\in\{1,\dots,k\}$
and hence they do not cross the boundary of any set $W_i$.
Therefore, augmenting $f$ along $C$ decreases the flow value on all edges in $E(C) \cap (\delta(W_1) \cup \dots \cup \delta(W_k))$
and we have already shown that this set is nonempty.
\end{proof}

\begin{lemma}\label{lemma:zero_flow_bar_f}
 Let $\bar D$ be a connected component of $(V,\bar F)$ with $V(\bar D) \cap V(B) = \emptyset$.
 Then $\bar f^* (E(\bar D)) = 0$.
\end{lemma}
\begin{proof}
Because $\bar f^*$ is a witness flow, we have $\bar f^*(\delta^-(v)) \le \bar f^*(\delta^+(v))$ for every $v\in V(\bar D)$.
Since
\[
 \bar f^*(E(\bar D))\ =\ \sum_{v\in V(\bar D)} \bar f^*(\delta^-(v))\ \le\ \sum_{v\in V(\bar D)} \bar f^*(\delta^+(v))\ =\ \bar f^*(E(\bar D)),
\]
we have $\bar f^*(\delta^-(v)) = \bar f^*(\delta^+(v))$ for every $v\in V(\bar D)$.
In other words, $\bar f^*$ restricted to $E(\bar D)$ is a circulation.
Because the support of $\bar f^*$ is is acyclic by Lemma~\ref{lemma:cycle_free_maintained}, this implies $\bar f^* (E(\bar D)) = 0$.
\end{proof}

\begin{lemma}\label{lemma:bar_F_visits_w_i_a_i}
 Let $i\in\{1,\dots, k\}$. Then $\bar F \cap \delta(W_i\cup \{a_i\}) \ne \emptyset$.
\end{lemma}
\begin{proof}
 By Lemma~\ref{lemma:visit_a_i_exatcly_once} there exists an edge $\bar e=(v, a_i) \in \bar F$.
 If $v \notin W_i$, we have $\bar e\in \bar F \cap \delta(W_i\cup \{a_i\}))$.
 Otherwise, the edge $e$ of $G$ that corresponds to $\bar e$ fulfills $e\in E[W_i] \cap \delta^-(\hat W_i)$.
 Therefore, we have $f(e)= x(e)$ as otherwise also the residual graph $G_f$ contained $e$, contradicting the choice of $\hat W_i$.
 This implies 
 \[ 
 \bar z^*(\bar e^1)\ \le\ \lceil 2 \bar z(\bar e^1)\rceil\ \le\ \lceil 2 z(e^1)\rceil\ =\ \lceil2(x(e) - f(e))\rceil\ =\ 0.
 \]
 But then
 \[
  \bar f^*(\bar e)\ =\  \bar z^*(\bar e^0)\ =\  \bar z^*(\bar e^0) +\bar z^*(\bar e^1)\ =\ \bar x^*(\bar e)\ \ge\ 1,
 \]
 because $\bar e\in \bar F$.
 By Lemma~\ref{lemma:zero_flow_bar_f}, this implies that the connected component $\bar D$ of $(V(\bar G), \bar F)$
 that contains $a_i$ also contains a vertex $w\in V(B)$.
 Since $W_i \cap V(B) = \emptyset$, this completes the proof.
\end{proof}

In the proof of property~$(iii)$ of Definition~\ref{def:subtour_cover} we will use the following observation.
\begin{lemma}\label{lemma:bar_F_crossings_L}
 Let $\bar D$ be a connected component of $(V,\bar F)$ with $V(\bar D) \cap V(B) = \emptyset$.
 Then $E(\bar D)$ contains no forward edge.
\end{lemma}
\begin{proof}
 By Lemma~\ref{lemma:zero_flow_bar_f} we have $\bar f^* (E(\bar D)) = 0$.
 Since $\bar f^*$ is a witness flow for $\bar x^* = \chi^{\bar F}$, this implies that $E(\bar D)$ contains no forward edge.
\end{proof}

The following lemma will be used in the proof of \eqref{eq:def_light1} of Theorem~\ref{thm:subtourpartitioncover}.
\begin{lemma}\label{lemma:bar_F_light}
 Let $\bar D$ be a connected component of $(V,\bar F)$ with $V(\bar D) \cap V(B) = \emptyset$.
 Then for every vertex $v\in V(\bar D)\setminus \{a_1,\dots, a_k\}$ with $y_v > 0$ we have $|\delta_{\bar F}^-(v)| \le 2$.
\end{lemma}
\begin{proof}
For every vertex $v\in V(\bar D)$ we have 
\begin{align*}
 |\delta_{\bar F}^-(v)|\ =\ \bar x^*(\delta^-(v))\ &=\ \bar z^*(\delta^-(v^1)) + \bar z^*(\delta^-(v^0)\setminus\{e_v^{\downarrow}\})\\
                         &=\ \bar z^*(\delta^-(v^1)) + \bar f^*(\delta^-(v))\\
                         &=\ \bar z^*(\delta^-(v^1)),
\end{align*}
where we used Lemma~\ref{lemma:zero_flow_bar_f}.
For all $v\in V(\bar D)\setminus \{a_1,\dots, a_k\}$ with $y_v>0$ we have $\{v\}\in\Lscr$ and hence $x(\delta^-(v)) =1$ by property~\eqref{item:x_in_strongly_laminar} of Definition~\ref{def:strongly_laminar}.
Therefore, by \eqref{item:flow_bound_through_v_1} we get
\[
 |\delta_{\bar F}^-(v)|\ =\ \bar z^*(\delta^-(v^1))\ \le\ \lceil 2\bar z(\delta^-(v^1)) \rceil\ \le\ \lceil 2 x(\delta^-(v)) \rceil\ =\ 2.
\]
\end{proof}

\subsection{Mapping back to $\boldmath{G}$}\label{sect:map_bar_F_back}

We now transform $\bar F$ into a solution $F$ of the Subtour Cover problem in $G$.
See Figure~\ref{fig:rerouting_flow}~(c)--(d).
By Lemma~\ref{lemma:visit_a_i_exatcly_once}, every vertex $a_i$ for $i\in\{1,\dots ,k\}$ has exactly one incoming edge in $\bar F$
and because $\bar F$ is Eulerian, $a_i$ also has exactly one outgoing edge.
We replace all the edges in $\delta_{\bar F}(a_i)$ for $i\in\{1,\dots ,k\}$ by their corresponding edges in $G$.
For every $i\in\{1,\dots ,k\}$ we added one edge $(v,s)\in \delta^-(\hat W_i)$ and 
an edge $(t,w)\in \delta^+(\hat W_i)$; to obtain an Eulerian edge set we add an $s$-$t$-path $P_i$ in $G[W_i]$.
Such a path exists because $G[W_i]$ is strongly connected.
Let $F$ be the resulting Eulerian multi-set of edges in $G$.
Note that if two vertices $a,b \in V(G)$ are in the same connected component of $(V(\bar G), \bar F)$, then they are also in the same connected 
component of $(V(G), F)$.

\begin{lemma}\label{lemma:new_path_do_not_cross_laminar}
 Let $i\in\{1,\dots,k\}$ and $L\in \Lscr_{\ge 2}$.
 Then $E(P_i) \cap \delta(L) = \emptyset$.
\end{lemma}
\begin{proof}\looseness=-1
We have $H \cap \delta(L) = \emptyset$ and the sets $W_1, \dots, W_k$ are the vertex sets of the connected components of $(V\setminus V(B), H)$.
Now $E(P_i) \subseteq E[W_i]$ implies $E(P_i) \cap \delta(L) = \emptyset$.
\end{proof}
We claim that $F$ is a solution to the Subtour Cover problem and fulfills~\eqref{eq:backbone_component_light1} and~\eqref{eq:def_light1} of Theorem~\ref{thm:subtourpartitioncover}.
Property~$(i)$ of a solution of the Subtour Cover problem (Definition~\ref{def:subtour_cover}) holds because $F$ is Eulerian.
Property~$(ii)$ follows from Lemma~\ref{lemma:bar_F_visits_w_i_a_i}.

We now show property~$(iii)$.
Let $D$ be a connected component of $(V, F)$ with $E(D) \cap \delta(L) \ne \emptyset$ for some $L\in \Lscr_{\ge 2}$.
Because $D$ is Eulerian it then contains a cycle $C$ with $E(C) \cap \delta(L) \ne \emptyset$.
But then $E(C) \subseteq E(D)$ contains a forward edge by Lemma~\ref{lemma:simple_berservation_froward_edges}.
By Lemma~\ref{lemma:r_values_consistent}, the edges of the paths $P_i$ (for $i\in \{1,\dots,k\}$) are neutral edges.
Hence, the forward edge in $C$ was already present in $\bar F$ and thus Lemma~\ref{lemma:bar_F_crossings_L} implies $V(D) \cap V(B) \ne \emptyset$.
This shows that $F$ is a solution to the Subtour Cover problem.
It remains to show~\eqref{eq:backbone_component_light1} and~\eqref{eq:def_light1}.

Lemma~\ref{lemma:new_path_do_not_cross_laminar} implies 
\[
 c(E(P_i)) = \sum_{v\in V(P_i)} |E(P_i) \cap \delta(v)| \cdot y_v \le \sum_{v\in W_i} 2y_v.
\]
Moreover, the sets $W_i$ for $i\in\{1,\dots, k\}$ are pairwise disjoint.
Using also $V(B) \cap W_i = \emptyset$ for $i\in\{1,\dots, k\}$, we obtain 
$\sum_{i=1}^k c(E(P_i)) \le \sum_{v\in V\setminus V(B)} 2y_v$.
Together with \eqref{eq:cost_bound_bar_F}, this implies~\eqref{eq:backbone_component_light1}.
 
Finally, we prove~\eqref{eq:def_light1}.
Let $D$ be a connected component of $(V, F)$ with $V(D) \cap V(B) = \emptyset$.
By property $(iii)$, $c(E(D))=\sum_{v\in V(D)} |F \cap \delta^-(v)|\cdot 2y_v$.
Because the sets $W_1, \dots, W_k$ are pairwise disjoint, we have $|F \cap \delta^-(v)| \le |\bar F \cap \delta^-(v)| +1$ for every vertex $v\in V(D)$.
By Lemma~\ref{lemma:bar_F_light}, this implies $|F \cap \delta^-(v)| \le 3$ for every vertex $v\in V(D)$ with $y_v >0$.
This shows~\eqref{eq:def_light1} and concludes the proof of Theorem~\ref{thm:subtourpartitioncover}.

\section{Algorithm for vertebrate pairs}\label{sect:solve_vertebrate_pairs}

In this section we present an algorithm for vertebrate pairs.
This algorithm is essentially due to Svensson~\cite{Sve15} who used it for node-weighted ATSP instances.
Later Svensson, Tarnawski, and V\'egh~\cite{SveTV18} adapted the algorithm to work with vertebrate pairs.
Here, we present an improved variant of their algorithm.

As a subroutine we will use the Subtour Cover algorithm from Theorem~\ref{thm:subtourpartitioncover}.
In order to exhibit the dependence of the approximation guarantee of the algorithm on the subroutine
we introduce the notion of an \emph{$(\alpha,\kappa,\beta)$-algorithm} for Subtour Cover.
Theorem~\ref{thm:subtourpartitioncover} yields a $(3,2,1)$-algorithm for Subtour Cover.

\begin{definition}
Let $\alpha,\kappa,\beta \ge 0$.
An \emph{$(\alpha,\kappa,\beta)$-algorithm} for Subtour Cover is a polynomial-time algorithm that 
computes a solution $F$ for every instance $(\Iscr,B, H)$ such that 
\begin{equation}\label{eq:backbone_component_light}
 c(F) \ \le \ \kappa\cdot\lp(\Iscr)+\beta\cdot\sum_{v\in V\setminus V(B)} 2y_{v},
\end{equation}
and for every connected component $D$ of $(V,F)$ with $V(D)\cap V(B)=\emptyset$ we have
\begin{equation}\label{eq:def_light}
   c(E(D)) \ \le \ \alpha \cdot \sum_{v\in V(D)} 2y_{v}.
\end{equation} 
\end{definition}
Let $\alpha,\kappa,\beta \ge 0$ such that there is an \emph{$(\alpha,\kappa,\beta)$-algorithm} $\Ascr$ for Subtour Cover and
let $\epsilon > 0$ be a fixed constant.
The goal of this section is to show that there is a polynomial-time $(\kappa, 4\alpha+ \beta + 1+\epsilon)$-algorithm for vertebrate pairs.

\subsection{Outline}

Let $(\Iscr, B)$ be a vertebrate pair.
Svensson's algorithm is initialized with an Eulerian multi-set $\tilde H \subseteq E[V \setminus V(B)]$  with $\tilde H\cap \delta(L)=\emptyset$ for all $L\in \Lscr_{\ge 2}$, and then computes either 
a ``better'' initialization $\tilde H'$ or extends $\tilde H$ to a solution $H$ of the given vertebrate pair $(\Iscr, B)$.

The initialization $\tilde H$ of the algorithm will always be \emph{light} (see Definition~\ref{def:light}).
To define what a \emph{light} edge set is, we introduce a function $\ell : V \rightarrow  \mathbb{R}_{\ge 0}$.
For $v\in V$ we set
\begin{equation*}
 \ell(v) \ := \ \begin{cases}
          (1+\epsilon')\cdot 2\alpha \cdot 2y_{v} + \frac{\epsilon'}{n} \cdot \sum_{u\in V\setminus V(B)} 2 y_u  & \text{ if } v\in V\setminus V(B) \\[2mm]
          \frac{\kappa \cdot \lp(\Iscr) + \beta \cdot \sum_{u\in V\setminus V(B)} 2y_{u} }{|V(B)|}
          & \text{ if } v\in V(B),
         \end{cases}
\end{equation*}
where $\epsilon':=\frac{\epsilon}{3+4\alpha+\sfrac{1}{2\alpha}}$.

\begin{definition}\label{def:light}
Let $\tilde H$ be a (multi-)subset of $E$.
We call $\tilde H$ \emph{light} if $c(E(D)) \le \ell(V(D))$ for every connected component $D$ of $(V,\tilde H)$. 
\end{definition}

Note that for $v\in V\setminus V(B)$ the first term of the definition of $\ell(v)$ is proportional to the corresponding dual variable $y_v$.
We need the additional term $\frac{\epsilon'}{n} \cdot \sum_{u\in V\setminus V(B)} 2 y_u$ to guarantee that $\ell(v)$ cannot be too close to zero;
see the proof of Lemma~\ref{lemma:improved_initialization}.
For vertices in $V(B)$ we will only need that 
$\ell(V(B)) = \kappa \cdot \lp(\Iscr) + \beta \cdot \sum_{u\in V\setminus V(B)} 2y_{u}$.

To measure what a ``better'' initialization for Svensson's algorithm is, we introduce a potential function $\Phi$.
For a multi-subset $\tilde H$ of $E[V\setminus V(B)]$ such that the connected components of 
$(V\setminus V(B),\tilde H)$ have vertex sets $\tilde W_1,\ldots,\tilde W_k$, we write 
\[ \Phi(\tilde H)= \sum_{i=1}^k \ell(\tilde W_i)^{1+p},\]
where $p\coloneqq \log_{1+\epsilon'}(\frac{2+\epsilon'}{\epsilon'})$.
Svensson~\cite{Sve15} and Svensson, Tarnawski, and V\'egh \cite{SveTV18} used $p=1$. 
Our choice of the potential function $\Phi$ will lead to an improved approximation ratio.

The following lemma states the result of Svensson's algorithm.

\begin{lemma}\label{lemma:result_svensson}
Let $\alpha,\kappa,\beta \ge 0$ such that there is an \emph{$(\alpha,\kappa,\beta)$-algorithm} for Subtour Cover and let $\epsilon > 0$ be a fixed constant. Then there exists a constant $C > 0$ such that the following holds.

Given a vertebrate pair $(\Iscr,B)$ with $\Iscr=(G,\Lscr,x,y)$ and a light Eulerian multi-subset $\tilde H$ of $E[V\setminus V(B)]$ with $\tilde H\cap \delta(L)=\emptyset$ for all $L\in \Lscr_{\ge 2}$, we can  compute in polynomial time 
 \begin{itemize}
  \item[(a)] a solution $H$ for the vertebrate pair $(\Iscr, B)$ such that 
  \begin{equation}
  \label{eq:cost_bound_H_svensson}
  c(H) \le \ell(V(B)) + (2+\sfrac{1}{2\alpha}) \cdot \ell(V \setminus V(B)),\end{equation}
  or 
  \item[(b)] a light Eulerian multi-subset $\tilde H'$ of $E[V\setminus V(B)]$ such that $\tilde H'\cap \delta(L)=\emptyset$ for all $L\in \Lscr_{\ge 2}$ and
 \begin{equation}\label{eq:potentialincrease}
             \Phi(\tilde H') - \Phi(\tilde H) \ > \ \Bigl(\sfrac{1}{C\cdot n} \cdot \ell(V\setminus V(B))\Bigr)^{1+p}. 
\end{equation}
 \end{itemize}
\end{lemma}
From Lemma~\ref{lemma:result_svensson} we can derive the main result of this section.
\begin{theorem}
\label{thm:reductiontosubtourpartitioncover}
Let $\alpha,\kappa,\beta \ge 0$ such that there is an \emph{$(\alpha,\kappa,\beta)$-algorithm} for Subtour Cover and
let $\epsilon > 0$ be a fixed constant.

Then there is a polynomial-time $(\kappa, 4\alpha + \beta + 1+\epsilon)$-algorithm for vertebrate pairs.
\end{theorem}

\begin{proof}
Define $\epsilon':=\frac{\epsilon}{3+4\alpha+\sfrac{1}{2\alpha}}$, $p$, $\ell$ and $\Phi$ as above.
We start with $\tilde H=\emptyset$ and apply Lemma~\ref{lemma:result_svensson}. 
If we obtain a set $\tilde H'$ as in Lemma~\ref{lemma:result_svensson}~\textit{(b)},
we set $\tilde H:=\tilde H'$ and iterate, i.e.\ we apply Lemma~\ref{lemma:result_svensson} again until
we obtain a set $H$ as in Lemma~\ref{lemma:result_svensson}~\textit{(a)}.
Since $0 \ \le \ \Phi(\tilde H) \ \le\ \ell(V\setminus V(B))^{1+p}$, we need at most $\left(C\cdot n\right)^{1+p}$ iterations.
At the end, the algorithm guaranteed by Lemma~\ref{lemma:result_svensson} returns a solution $H$ for the vertebrate pair $(\Iscr, B)$ such that 
\begin{align*}
c(H) \ \le& \ \ell(V(B)) + \left(2+\sfrac{1}{2\alpha}\right) \cdot \ell(V\setminus V(B)) \\
        \ \le& \ \kappa \cdot \lp(I) + \left( \beta+\left(2+\sfrac{1}{2\alpha}\right)\cdot \left((1+\epsilon')\cdot 2\alpha+ \epsilon'\right) \right) \cdot \sum_{v\in V\setminus V(B)}2y_v \\
        \ =& \ \kappa \cdot \lp(I) + \left(\beta+4\alpha + 1 + 4\alpha \cdot \epsilon' + \epsilon' + (2 +\sfrac{1}{2\alpha})\cdot\epsilon' \right) \cdot \sum_{v\in V\setminus V(B)}2y_v \\
        \ =& \ \kappa \cdot \lp(I) + \left(4\alpha+ \beta+ 1+ \epsilon \right) \cdot \sum_{v\in V\setminus V(B)}2y_v. 
\end{align*}
\vspace*{-5mm}

\end{proof}

\subsection{Basic properties of the function $\boldmath{\ell}$ and algorithm $\Ascr$}\label{sect:basic_properties_svensson}

In this section we describe the key properties of the function $\ell$ and our given $(\alpha,\kappa,\beta)$-algorithm $\Ascr$ for Subtour Cover. 
\begin{lemma}\label{lemma:properties_subtour_cover_l}
 Let $\Ascr$ be an $(\alpha, \kappa, \beta)$-algorithm for Subtour Cover.
 Let $F$ be the output of $\Ascr$ applied to an instance $(\Iscr, B, H)$.
 \begin{enumerate}[(i)]
 \item\label{item:lightness} For every connected component $D$ of $(V,F)$ with $V(D) \cap V(B)= \emptyset$ we have 
       \[ c(E(D)) \le  \frac{1}{2(1+\epsilon')} \cdot \ell(V(D)). \]
 \item\label{item:backbone_comp} \looseness=-1
       Let the graph $D_B$ be the union of all connected components $D$ of $(V,F)$ with $V(D)\cap V(B) \ne \emptyset$.
       Then 
       \[ c(E(D_B)) \le \ell(V(B)). \]
 \end{enumerate}
\end{lemma}
\begin{proof}
The claimed properties follow directly from the definition of $\ell$ and the definition of an $(\alpha,\kappa,\beta)$-algorithm for Subtour Cover:
property~\eqref{item:lightness} follows from~\eqref{eq:def_light} and property~\eqref{item:backbone_comp} follows from~\eqref{eq:backbone_component_light}.
\end{proof}
The next lemma will be needed to show that Svensson's algorithm makes sufficient progress when finding ``a better initialization''.
\begin{lemma}\label{lemma:regularization_of_l}
 There exists a constant $C>0$ such that for every vertex $v\in V\setminus V(B)$ we have 
       \[ \ell(v) \ge \frac{1}{C\cdot n} \cdot \ell(V\setminus V(B)). \]
\end{lemma}
\begin{proof}
 We have
 \begin{align*}
   \ell(V\setminus V(B))\ 
   \le&\   (1+\epsilon')\cdot 2\alpha \cdot \sum_{u\in V\setminus V(B)} 2y_{u} + \epsilon' \cdot \sum_{u\in V\setminus V(B)} 2 y_u \\
   \le&\ \left((1+\epsilon')\cdot 2\alpha + \epsilon'\right) \cdot \sum_{u\in V\setminus V(B)} 2 y_u.
 \end{align*}
 Therefore, for every vertex $v\in V\setminus V(B)$ we have  
 \[ \ell(v)\ \ge\ \frac{\epsilon'}{n} \cdot \sum_{u\in V\setminus V(B)} 2 y_u\ \ge\ \frac{\epsilon'}{\left((1+\epsilon')\cdot 2\alpha + \epsilon'\right)\cdot n}\cdot \ell(V\setminus V(B)), \]
 which completes the proof because $\alpha$ and $\epsilon'$ are constants.
\end{proof}
The following property of $\ell$ is not crucial for obtaining a constant-factor approximation, 
but allows us to obtain a better approximation ratio.
\begin{lemma}\label{lemma:circuits_are_light}
  For every cycle $C$ in $G[V\setminus V(B)]$ with $E(C) \cap \delta(L) = \emptyset$ for all $L\in \Lscr_{\ge 2}$, we have 
 \[ c(E(C)) \le \frac{1}{2\alpha(1+\epsilon')} \cdot \ell(V(C)). \]
\end{lemma}
\begin{proof}
  Let $C$ be a cycle with $E(C) \cap \delta(L) = \emptyset$ for all $L\in \Lscr_{\ge 2}$.
  Then we have $c(E(C)) = \sum_{v\in V(C)} 2y_v \le  \frac{1}{(1+\epsilon')\cdot 2\alpha}\cdot \ell(V(C))$.
\end{proof}

In the following sections we will only use Lemma~\ref{lemma:properties_subtour_cover_l}, 
Lemma~\ref{lemma:regularization_of_l} and Lemma~\ref{lemma:circuits_are_light} and we will not use the precise definition of $\ell$ anymore.

\subsection{Finding a better initialization}

In this section we discuss how Svensson's algorithm finds in certain cases a better initialization $\tilde H'$.
We will need the following well-known statement about the knapsack problem.
\begin{lemma}\label{lemma:knapsack}
 Suppose we are given a finite set $I$ of items and for every item $j\in I$ a weight $w_j > 0$ and a profit $p_j \ge 0$.
 Moreover, let $\bar w < \sum_{j \in I} w_j$ be a given weight limit.
 
 Then we can compute in polynomial time a set $J\subseteq I$ such that
 \begin{itemize}
  \item $ \sum_{j\in J} w_j\ \le\ \bar w$, and \\[-5mm]
  \item $ \sum_{j\in J} p_j\ \ge\ \frac{\bar w}{\sum_{j \in I} w_j} \cdot \sum_{j\in I} p_j-\max_{j\in I}\ p_j$
 \end{itemize}
\end{lemma}
\begin{proof}
We run the following greedy algorithm.
Sort the items by nonincreasing ratio $\sfrac{p_j}{w_j}$.
Consider the items in this order and, starting with $J=\emptyset$, add items to the set $J$ as long as $\sum_{j\in J} w_j \le \bar w$.
Then adding the next item to $J$ would result in a set $J'$ with 
$\sum_{j\in J'} w_j > \bar w$.
By the sorting of the items,
\[ \textstyle \sum_{j\in J'} p_j \ =\ \frac{\sum_{j\in J'} p_j}{\sum_{j\in J'} w_j} \cdot \sum_{j\in J'} w_j\ \ge\ \frac{\sum_{j\in I} p_j}{\sum_{j\in I} w_j} \cdot \sum_{j\in J'} w_j\ >\ 
\frac{\sum_{j\in I} p_j}{\sum_{j\in I} w_j} \cdot \bar w. \]
Because $J'\setminus J$ contains only one element, this implies
\[ \textstyle \sum_{j\in J} p_j\ \ge\ \frac{\sum_{j\in I} p_j}{\sum_{j\in I} w_j} \cdot \bar w - \max_{j\in I}\ p_j  
\ =\ \frac{\bar w}{\sum_{j \in I} w_j} \cdot \sum_{j\in I} p_j-\max_{j\in I}\ p_j. \]
\end{proof}

Let $\tilde H$ be a light Eulerian multi-subset of $E[V\setminus V(B)]$ with $\tilde H\cap \delta(L)=\emptyset$ for all $L\in \Lscr_{\ge 2}$.
Let  $\tilde W_0 = V(B)$ and let $\tilde W_1,\ldots,\tilde W_k$ be the vertex sets of the connected components of $(V\setminus V(B),\tilde H)$, 
ordered so that $\ell(\tilde W_1)\ge\cdots\ge \ell(\tilde W_k)$. 
For a connected multi-subgraph $D$ of $G$ we define the \emph{index} of $D$ to be
\begin{equation*}
\mathrm{ind}(D) \ := \ \min\{j\in\{0,\ldots,k\}: V(D)\cap \tilde W_j\not=\emptyset\}.
\end{equation*}
The following is the main lemma that we will use to find a better initialization $\tilde H'$.

\begin{lemma}\label{lemma:improved_initialization}
Let $D$ be a subtour, i.e., a connected and Eulerian multi-subgraph of $G$, such that $V(D) \cap V(B) = \emptyset$, $E(D)\cap \delta(L)=\emptyset$ for all $L\in \Lscr_{\ge 2}$, and such that
\begin{equation}\label{eq:atsp_D_light}
c(E(D)) \ \le \ \frac{2}{2+\epsilon'} \cdot \ell(V(D)),
\end{equation}
and
\begin{equation}\label{eq:atsp_D_large}
\ell(V(D)) \ > \ (1+\epsilon') \cdot \ell(\tilde W_{\mathrm{ind}(D)}).
\end{equation}
Then we can compute in polynomial time a light  Eulerian multi-subset $\tilde H'$ of $E[V \setminus V(B)]$ such that $\tilde H'\cap \delta(L)=\emptyset$ for all $L\in \Lscr_{\ge 2}$ and \eqref{eq:potentialincrease} holds.
\end{lemma}

\begin{proof}
Let $I:=\{j\in\{0,\ldots,k\}: V(D)\cap \tilde W_j\not=\emptyset\}$ and $i:=\min I =\mathrm{ind}(D)$.
We have $i > 0$ because $V(D) \cap V(B)= \emptyset$.
We will compute a subset $J$ of $I$ and replace the components $\tilde H[\tilde W_j]$ for $j \in I$
by one new component that is the union of $E(D)$ and all $\tilde H[\tilde W_j]$ with $j\in J$. 
More precisely, we set 
\begin{equation*}
\tilde H ' \ :=\ \bigcup_{h\in \{1,\ldots,k\}\setminus I} \tilde H[\tilde W_h]\ \cupp\ E(D)\ \cupp\ \bigcup_{j\in J} \tilde H[\tilde W_j].
\end{equation*}
See Figure~\ref{fig:improved_initialization}.
Let $D^*$ be the connected component of $(V, \tilde H')$ with edge set
\begin{equation*}
 E(D) \cupp \bigcup_{j\in J} \tilde H[\tilde W_j].
\end{equation*}
We will choose $J$ such that
\begin{equation}\label{eq:knapsack_constraint}
 \sum_{j\in J} \ell(\tilde W_j \cap V(D)) \ \le \ \frac{\epsilon'}{2+\epsilon'} \cdot \ell(V(D)).
\end{equation}

 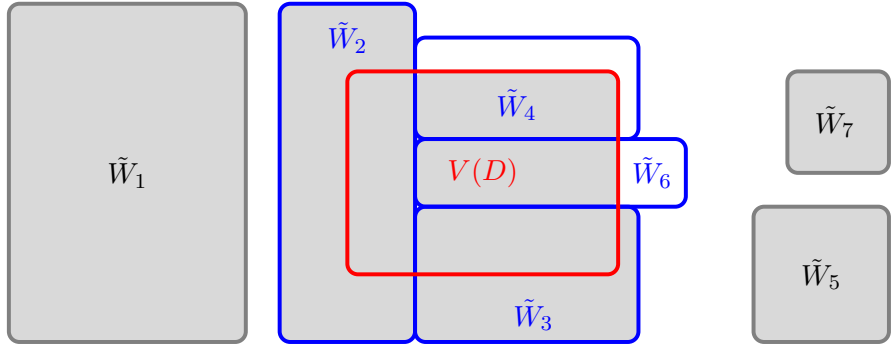
\begin{figure}
 \begin{center}
  \begin{tikzpicture}[scale=0.9]
  \tikzset{vertex/.style={fill=black, circle,inner sep=0em,minimum size=4pt }}
  \tikzset{Iset/.style={rounded corners, line width=1.5pt, blue}}
  \tikzset{Jset/.style={rounded corners, line width=1.5pt, blue, fill =gray, fill opacity=0.3}}
  \tikzset{Dset/.style={rounded corners, line width=1.5pt,red}}
  \tikzset{FillSet/.style={rounded corners, draw=none, fill=gray, fill opacity=0.3}}
  \tikzset{OtherSet/.style={rounded corners, line width=1.5pt,gray, fill=gray, fill opacity=0.3}}
  
  \draw[OtherSet] (-2,2) rectangle (1.5,7) {};
  \draw[OtherSet] (9,2) rectangle (11,4) {};
  \draw[OtherSet] (9.5,4.5) rectangle (11,6) {};
  
  \draw[FillSet] (4,4) rectangle (7,6) {};
  \draw[Jset] (2,2) rectangle (4,7) {};
  \draw[Jset] (4,2) rectangle (7.3,4) {};
  \draw[Iset] (4,4) rectangle (8,5) {};
  \draw[Iset] (4,5) rectangle (7.3,6.5) {};
  \draw[Dset] (3,3) rectangle (7,6) {};
  
  \begin{scope}[black]
   \node () at (-0.25,4.5) {$\tilde W_1$};
   \node () at (10,3) {$\tilde W_5$};
   \node () at (10.2,5.25) {$\tilde W_7$};
  \end{scope}
   \begin{scope}[blue]
    \node() at (3,6.5) {$\tilde W_2$};
    \node() at (5.75,2.4) {$\tilde W_3$};
    \node() at (5.5,5.5) {$\tilde W_4$};
    \node() at (7.5,4.5) {$\tilde W_6$};
  \end{scope}
   \begin{scope}[red]
    \node () at (5,4.5) {$V(D)$};
  \end{scope}

  \end{tikzpicture}
  \end{center}
  \caption{Illustration of the proof of Lemma~\ref{lemma:improved_initialization}.
  The gray and blue rectangles show the partition of $V\setminus V(B)$ into $\tilde W_1,\dots, \tilde W_7$.
  In red we see the vertex set $V(D)$ of the given connected graph $D$.
  The rectangles with blue boundary show the sets $\tilde W_i$ with $i\in I$.
  In this example $I=\{2,3,4,6\}$.
  The filled areas show vertex sets of connected components of $(V\setminus V(B),\tilde H')$.
  In this example we have $J=\{2,3\}$.
  The connected components $\tilde H[\tilde W_1]$, $\tilde H[\tilde W_5]$, and $\tilde H[\tilde W_7]$ remain unchanged
  and we get a new component $D^*$ with vertex set $V(D) \cup \tilde W_2 \cup \tilde W_3$; we also get singleton components (without edges)
  for all vertices in $\tilde W_4 \setminus V(D)$ and $\tilde W_6 \setminus V(D)$.
  \label{fig:improved_initialization}
  }
 \end{figure}  

We first show that then $c(E(D^*)) \le \ell(V(D^*))$, which implies that $\tilde H'$ is light. 
Indeed, using \eqref{eq:atsp_D_light} in the first inequality and \eqref{eq:knapsack_constraint} in the last inequality,
\begin{align*}
 c(E(D^*)) \ \le \ & \frac{2}{2+\epsilon'}\cdot  \ell(V(D)) + \sum_{j\in J} \ell(\tilde W_j) \\
 = \ &   \frac{2}{2+\epsilon'} \cdot \ell(V(D)) + \sum_{j\in J} \ell(\tilde W_j\setminus V(D)) + \sum_{j\in J} \ell(\tilde W_j\cap V(D)) \\
 \le \ & \frac{2}{2+\epsilon'} \cdot \ell(V(D)) + \sum_{j\in J} \ell(\tilde W_j\setminus V(D)) + \frac{\epsilon'}{2+\epsilon'} \cdot \ell(V(D)) \\
 = \ & \ell\left(V(D^*)\right).
\end{align*}
We conclude the proof by showing that we can choose $J$ such that \eqref{eq:knapsack_constraint} and \eqref{eq:potentialincrease} hold.
To this end, we would like to make the new component, spanning $V(D) \cup\bigcup_{j\in J}\tilde W_j$, 
as large as possible. More precisely, we want to maximize $\sum_{j\in J} \ell(\tilde W_j\setminus V(D))$ subject to \eqref{eq:knapsack_constraint}.
This is a knapsack problem: the items are indexed by $I$, and item $j\in I$ has weight $w_j = \ell(\tilde W_j \cap V(D))$ and profit $p_j = \ell(\tilde W_j\setminus V(D))$.
Since $\sum_{j\in I} \ell(\tilde W_j \cap V(D)) = \ell(V(D))$, the weight limit $\bar w = \frac{\epsilon'}{2+\epsilon'}\cdot\ell(V(D))$ is an $\frac{\epsilon'}{2+\epsilon'}$ fraction of the total weight of all items. 
Since any item $j\in I$ has profit at most $\ell(\tilde W_j\setminus V(D))\le \ell(\tilde W_j)\le \ell(\tilde W_i)$, Lemma~\ref{lemma:knapsack} yields a set $J$ with \eqref{eq:knapsack_constraint} and 
\begin{equation*}
\sum_{j\in J} \ell(\tilde W_j\setminus V(D)) \ \ge \ \frac{\epsilon'}{2+\epsilon'} \cdot \sum_{j\in I} \ell(\tilde W_j\setminus V(D)) \ - \ \ell(\tilde W_i).
\end{equation*}

Finally we show \eqref{eq:potentialincrease}.
Using \eqref{eq:atsp_D_large} in both strict inequalities, $(1+\epsilon')^p = \sfrac{2+\epsilon'}{\epsilon'}$ in the second equation,
and $\ell(\tilde W_i)\ge \ell(\tilde W_j)$ for all $j\in I$ in the last inequality, we obtain
\begin{align*}
 \ell(V(D^*))^{1+p} \ =& \ \left( \ell(V(D)) + \sum_{j\in J} \ell(\tilde W_j\setminus V(D)) \right)^{1+p}  \\
  \ \ge& \ \ell(V(D))^{p} \cdot \left( \ell(V(D)) + \frac{\epsilon'}{2+\epsilon'}\sum_{j\in I} \ell(\tilde W_j\setminus V(D)) - \ell(\tilde W_i) \right) \\
 \ >& \ \left(\big.(1+\epsilon') \cdot \ell(\tilde W_i)\right)^{p} \cdot \\
 &\qquad \left( \frac{2}{2+\epsilon'}\ell(V(D)) + \frac{\epsilon'}{2+\epsilon'}\ell(V(D)) + \frac{\epsilon'}{2+\epsilon'}\sum_{j\in I} \ell(\tilde W_j\setminus V(D)) - \ell(\tilde W_i) \right) \\
  \ >& \ \left(\big.(1+\epsilon') \cdot \ell(\tilde W_i)\right)^{p} \cdot \\
  & \qquad \left( \frac{\epsilon'}{2+\epsilon'} \cdot \ell(\tilde W_i) +\frac{\epsilon'}{2+\epsilon'} \ell(V(D)) + \frac{\epsilon'}{2+\epsilon'}\sum_{j\in I} \ell(\tilde W_j\setminus V(D)) \right) \\
  \ =& \ \frac{2+\epsilon'}{\epsilon'} \cdot \ell(\tilde W_i)^{p} \cdot
  \left( \frac{\epsilon'}{2+\epsilon'} \cdot \ell(\tilde W_i) + \frac{\epsilon'}{2+\epsilon'}\sum_{j\in I} \ell(\tilde W_j) \right) \\
  \ \ge& \ \ell(\tilde W_i)^{1+p} + \sum_{j\in I} \ell(\tilde W_j)^{1+p} .
\end{align*}
and hence
\begin{align*}
 \Phi (\tilde H') - \Phi(\tilde H) \ \ge \ \ell(V(D^*))^{1+p} - \sum_{j\in I} \ell(\tilde W_j)^{1+p} 
 \ > \ \ell(\tilde W_i)^{1+p}. 
 \end{align*}
Since $\tilde W_i$ contains at least one vertex, by Lemma~\ref{lemma:regularization_of_l},
$\ell(\tilde W_i)\ge \frac{1}{C\cdot n} \cdot \ell(V\setminus V(B))$ for the constant $C$ from Lemma~\ref{lemma:regularization_of_l}.
\end{proof}

The two different ways how we obtain $D$ during Svensson's algorithm are described by Lemma~\ref{lemma:better_initialization_1} and Lemma~\ref{lemma:better_initialization_2}. See also Figure~\ref{fig:application_improved_initialization}.

 \begin{figure}[t]
 \begin{center}
  \begin{tikzpicture}[scale=0.9]
  
  \begin{scope}[draw=none, fill=gray, fill opacity =0.3]
   \fill (-0.2,3.0) ellipse (1.2 and 1.5);
   \fill[red] (3.5,4.5) ellipse (2 and 1.3);
   \fill (3.5,1.5) ellipse (1.8 and 1.2);
   
   \fill (7.5,5.1) ellipse (1.3 and 0.8);
   \fill (7.5,3.4) ellipse (1.2 and 0.7);
   \fill (7.5,1.9) ellipse (1.0 and 0.6);
   \fill (7.5,0.5) ellipse (1.0 and 0.6);
   
   \fill (10,4.9) ellipse (0.7 and 0.5);
   \fill (10,3.5) ellipse (0.7 and 0.5);
   \fill (10,1.9) ellipse (0.7 and 0.5);
   \fill (10,0.5) ellipse (0.7 and 0.5);
  \end{scope}

  \begin{scope}[black]

   \node () at  (-0.2,3.0) {\footnotesize $\tilde W_0 = V(B)$};
   \node () at  (2.8,4.6) {\footnotesize $\tilde W_1$};
   \node () at  (2.8,1.9) {\footnotesize $\tilde W_2$};
   
   \node () at  (7.5,5.1) {\footnotesize $\tilde W_3$};
   \node () at  (7.8,3.4) {\footnotesize $\tilde W_4$};
   \node () at  (7.7,1.9) {\footnotesize $\tilde W_5$};
   \node () at  (7.5,0.5) {\footnotesize $\tilde W_6$};
   
   \node () at  (10.2,4.9) {\footnotesize $\tilde W_7$};
   \node () at  (10.0,3.5) {\footnotesize $\tilde W_8$};
   \node () at  (10.2,1.9) {\footnotesize $\tilde W_9$};
   \node () at  (10.2,0.5) {\footnotesize $\tilde W_{10}$};
  \end{scope}
   
  \begin{scope}[black, very thick]
    \draw plot [smooth cycle] coordinates {(9.5,4.6) (9.6,5.1) (8.3,5.3) (8.1,4.5)};
    \draw plot [smooth cycle] coordinates {(9.7,1.9) (9.6, 0.4) (8.0, 0.5)};
    \draw plot [smooth cycle] coordinates {(6.7, 2.1) (6.7, 1.7) (4.0,1.8) (3.5, 2.3) (4.2,2.4) };
    \draw[blue] plot [smooth cycle] coordinates {(6.8, 0.9) (6.8, 0.4) (5.4, 0.5) (3.0, 0.8) (3.7,1.4) (5.0, 1.0) };
   \begin{scope}[red]
    \draw plot [smooth cycle] coordinates { (3.5,5) (4, 5.4) (6.5,5.2) (6.8,4.6) };
    \draw plot [smooth cycle] coordinates { (3.5,3.5) (3.5, 4.0) (6.8,3.8) (6.8,3.1) };
   \end{scope}
  \end{scope}
  \end{tikzpicture}
  \end{center}
  \caption{Illustration of Lemma~\ref{lemma:better_initialization_1} and Lemma~\ref{lemma:better_initialization_2}.
  Here the filled ellipses show the partition $\tilde W_0,\dots, \tilde W_{10}$ of $V$.
  The curves show a possible solution $F$ to some instance $(\Iscr, H)$ of Subtour Cover; the set $H$ is not shown here.
  In red we see a subgraph $D$ as in Lemma~\ref{lemma:better_initialization_1}: here the red curves are the graph $F_1$ and 
  $D$ is the union of $F_1$ and $\tilde H[\tilde W_1]$. 
  In blue we see a subgraph $D$ as in Lemma~\ref{lemma:better_initialization_2}: here $D$ is a single connected component of $F$
  and in this example we have $\textrm{ind}(D)=2$.
  \label{fig:application_improved_initialization}
  }
 \end{figure}
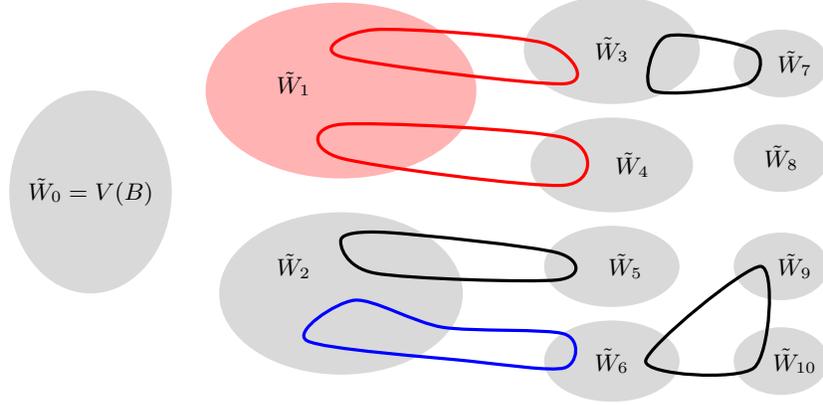  

\begin{lemma}\label{lemma:better_initialization_1}
Let $\Ascr$ be an $(\alpha, \kappa, \beta)$-algorithm for Subtour Cover.
Let $F$ be the output of $\Ascr$ applied to an instance $(\Iscr, B, H)$.
For $i\in\{0,\ldots,k\}$ let the graph $F_i$ be the union of the connected components $D'$
of $(V,F)$ with $\mathrm{ind}(D')=i$.

Suppose we have $c(E(F_i)) > \ell(\tilde W_i)$ for some $i\in \{0,\dots, k\}$.
Then the union 
\[D := \left(\tilde W_i \cup V(F_i),\ \tilde H[\tilde W_i] \cup E(F_i)\right) \]
of $\tilde H[\tilde W_i]$ and $F_i$ fulfills the conditions of Lemma~\ref{lemma:improved_initialization}, i.e.\ $D$ is  a connected Eulerian multi-subgraph of $G$ with $V(D) \cap V(B) = \emptyset$,
 $E(D)\cap \delta(L)=\emptyset$ for all $L\in \Lscr_{\ge 2}$,\eqref{eq:atsp_D_light} and \eqref{eq:atsp_D_large}.
\end{lemma}
\begin{proof}
Let $i\in\{0,\ldots,k\}$ such that $c(E(F_i)) > \ell(\tilde W_i)$.
Note that $i>0$ because $c(E(F_0))\le \ell(\tilde W_0)$ by Lemma~\ref{lemma:properties_subtour_cover_l}~\eqref{item:backbone_comp}.
This implies $V(D) \cap V(B) = \emptyset$ and $E(D)\cap \delta(L)=\emptyset$ for all $L\in \Lscr_{\ge 2}$.
Moreover, we have
\begin{equation*}
\frac{1}{2+2\epsilon'} \cdot \ell(V(D)) \ \ge \ \frac{1}{2+2\epsilon'} \cdot \ell(V(F_i)) 
\ \ge \ c(E(F_i)) \ > \ \ell(\tilde W_i), 
\end{equation*}
where the second inequality holds by Lemma~\ref{lemma:properties_subtour_cover_l}~\eqref{item:lightness}.
This shows \eqref{eq:atsp_D_large} and implies 
\begin{equation*}
 c(E(D)) = c(\tilde H[\tilde W_i] \cup E(F_i)) \ \le \ \ell(\tilde W_i) + c(E(F_i)) \ \le \ \frac{2}{2+2\epsilon'} \cdot \ell(V(D)).
\end{equation*}
Therefore also \eqref{eq:atsp_D_light} holds.
\end{proof}

\begin{lemma}\label{lemma:better_initialization_2}
Let $\Ascr$ be an $(\alpha, \kappa, \beta)$-algorithm for Subtour Cover.
Let $F$ be the output of $\Ascr$ applied to an instance $(\Iscr, B, H)$.
Suppose $(V,F)$ has a connected component $D$ with $\mathrm{ind}(D) > 0$ and 
\[l(V(D)) > (1+\epsilon') \cdot l(\tilde W_{\mathrm{ind}(D)}). \]
Then $D$ fulfills the conditions of Lemma~\ref{lemma:improved_initialization}, i.e.\ $D$ is  a connected Eulerian multi-subgraph of $G$ with $V(D) \cap V(B) = \emptyset$,
 $E(D)\cap \delta(L)=\emptyset$ for all $L\in \Lscr_{\ge 2}$, \eqref{eq:atsp_D_light} and \eqref{eq:atsp_D_large}.
\end{lemma}
\begin{proof}
We have \eqref{eq:atsp_D_large} by assumption.
Moreover, $V(D) \cap V(B) = \emptyset$ and $E(D)\cap \delta(L)=\emptyset$ for all $L\in \Lscr_{\ge 2}$
because $\mathrm{ind}(D) > 0$. 
Since $D$ is a connected component of $(V,F)$ that does not intersect the backbone, Lemma~\ref{lemma:properties_subtour_cover_l}~\eqref{item:lightness} implies
\begin{equation*}
 c(E(D)) \ \le \ \frac{1}{2+2\epsilon'} \cdot l(V(D)) \ \le \ \frac{2}{2+\epsilon'} \cdot l(V(D)),
\end{equation*}
implying \eqref{eq:atsp_D_light}.
\end{proof}

\subsection{Svensson's algorithm}
In this section we prove Lemma~\ref{lemma:result_svensson}.
To this end we consider Algorithm~\ref{algo:svensson_algo}, essentially due to Svensson~\cite{Sve15}.
We maintain an Eulerian edge set $H$ which is initialized with $H = \tilde H$.
Then we iterate the following steps.
First, we call the given algorithm for Subtour Cover,
then we try to find an improved initialization $\tilde H'$ as discussed in the previous section,
and finally, if we could not find a better initialization, we extend the set $H$.
The careful update of $H$ in step~3 of Algorithm~\ref{algo:svensson_algo} is illustrated
in Figure~\ref{figsvensson}.

In addition to our definition of $\Phi$, the other main difference to the version of this algorithm in \cite{SveTV18} are the properties of $C$ in step~(3c). This is inspired by a remark in \cite{Sve15}.
In order to make this work for vertebrate pairs, we exploit our slightly stronger definition of the Subtour Cover problem
(see the proof of Lemma~\ref{lemma:atsp_f_edges}).

To implement step~(3c), consider each edge $e=(v,w) \in \delta^+(V(Z))$ and compute a shortest $w$-$v$-path $P$ in 
$\big(V \setminus V(B),\ E[V\setminus V(B)] \setminus \big( \cup_{L\in \Lscr_{\ge 2}} \delta(L) \big) \big)$ 
and check if $c(e) + c(E(P)) \le \frac{1}{2\alpha} \cdot \ell(\tilde W_{\mathrm{ind}(Z)})$.

Note that adding $E(C)$ to $X$ in step~(3c) decreases the number of connected components of $(V,H\cupp F\cupp X)$,
and adding edges to $H$ in step~(3d) decreases the number of connected components of $(V,H)$.
Thus the procedure terminates after a polynomial number of steps. 

\allowbreak
\begin{algorithm2e}[H]
\vspace*{4mm} 
\KwIn{a vertebrate pair $(\Iscr, B)$ with $\Iscr=(G, \Lscr,x,y)$,\newline 
      a light Eulerian multi-subset $\tilde H \subseteq E[V\setminus V(B)]$ with $\tilde H\cap \delta(L)=\emptyset$ for all $L\in \Lscr_{\ge 2}$, \newline
      $\alpha, \kappa,\beta \ge 0$, $\epsilon > 0$, and \newline 
      an $(\alpha, \kappa, \beta)$-algorithm $\Ascr$ for Subtour Cover
}
\vspace*{3mm}
\KwOut{ either $\tilde H'$ as in Lemma~\ref{lemma:result_svensson}~\textit{(b)} or $H$ as in Lemma~\ref{lemma:result_svensson}~\textit{(a)}}
\vspace*{8mm}
Let $\tilde W_0:=V(B)$ and let $\tilde W_1, \dots, \tilde W_k$ be the vertex sets of the 
connected components of $(V\setminus V(B), \tilde H)$ such that $\ell(\tilde W_1) \ge \ell(\tilde W_2) \ge \dots \ge \ell(\tilde W_k)$.\\[1mm]
Set $H := \tilde H$. \\[5mm]
While $(V,E(B)\cup H)$ is not connected, repeat the following: \\[1mm]
\begin{enumerate}
\item[\textbf{1.}] \textbf{Compute a solution to Subtour Cover:}
 \begin{enumerate}
  \item[(1a)] Apply $\Ascr$ to the Subtour Cover instance $(\Iscr,B,H)$ to obtain a solution $F'$.
  \item[(1b)] Let $F$ result from $F'$ by deleting all edges of connected components of $(V,F')$ whose vertex sets are contained in 
              a connected component of $(V,E(B)\cup H)$.
 \end{enumerate}
 \item[\textbf{2.}] \textbf{Try to find a better initialization $\bm{\tilde H'}$:} \\[1mm]
 For $i\in\{0,\ldots,k\}$ let the graph $F_i$ be the union of the connected components $D'$ of $(V,F)$ with $\mathrm{ind}(D')=i$.
 \begin{enumerate}
  \item[(2a)]  If for some $i\in\{0,\ldots,k\}$ we have $c(E(F_i)) > \ell(\tilde W_i)$, 
               apply Lemma~\ref{lemma:improved_initialization} to $D=(\tilde W_i \cup V(F_i), \tilde H[\tilde W_i] \cup E(F_i))$
               to obtain an edge set $\tilde H'$. Then return $\tilde H '$.
  \item[(2b)] If $(V,F)$ has a connected component $D$ with $\ell(V(D)) > (1+\epsilon')\cdot \ell(\tilde W_{\mathrm{ind}(D)})$ and $\mathrm{ind}(D) > 0$, apply Lemma~\ref{lemma:improved_initialization} to obtain an edge set $\tilde H'$.
  Then return~$\tilde H '$.
 \end{enumerate}
\item[\textbf{3.}] \textbf{Extend $\bm{H}$:}
  \begin{enumerate}
   \item[(3a)] Set $X:=\emptyset$.
   \item[(3b)] Select the connected component $Z$ of $(V,E(B) \cupp H\cupp F\cupp X)$ for which $\mathrm{ind}(Z)$ is largest.
   \item[(3c)]  If there is a cycle $C$ in $G[V\setminus V(B)]$ with
               \begin{itemize}
                \item $E(C)\cap\delta(V(Z)) \not=\emptyset$,
                \item $E(C)\cap \delta(L)=\emptyset$ for all $L\in \Lscr_{\ge 2}$, and 
                \item $c(E(C)) \le \frac{1}{2\alpha} \cdot \ell(\tilde W_{\mathrm{ind}(Z)})$,
               \end{itemize}
          then add $E(C)$ to $X$ and go to step~(3b).
   \item[(3d)]  Add the edges of $(V,F\cupp X)[V(Z)]$ to $H$.
  \end{enumerate}
\end{enumerate}
Return $H$.
\vspace*{4mm}
\caption{Svensson's Algorithm \label{algo:svensson_algo}
}
\end{algorithm2e}
\allowbreak
\newpage

\begin{figure}
\begin{center}
\begin{tikzpicture}[thick, xscale=0.89, yscale=-0.89]
\tikzset{set/.style={circle, draw=none, fill=gray, fill opacity =0.3, text=black, text opacity=1.0, inner sep=1}}

\node[set, minimum size=44] (w0)  at ( -8, 2) {\tiny $\tilde W_0 = V(B)$};
\node[set, minimum size=41] (w1)  at ( -5.8, 2) {\tiny $\tilde W_1$};
\node[set, minimum size=38] (w2)  at ( -3.8, 2) {\tiny $\tilde W_2$};
\node[set, minimum size=35] (w3)  at ( -1.9, 2) {\tiny $\tilde W_3$};
\node[set, minimum size=32] (w4)  at ( -0.1, 2) {\tiny $\tilde W_4$};
\node[set, minimum size=29] (w5)  at ( 1.5, 2) {\tiny $\tilde W_5$};
\node[set, minimum size=26] (w6)  at ( 3.0, 2) {\tiny $\tilde W_6$};
\node[set, minimum size=23] (w7)  at ( 4.4, 2) {\tiny $\tilde W_7$};
\node[set, minimum size=20] (w8)  at ( 5.7, 2) {\tiny $\tilde W_8$};
\node[set, minimum size=17] (w9)  at ( 6.9, 2) {\tiny $\tilde W_9$};

\begin{scope}[line width = 0.3mm, ]
  \draw[darkred] plot [smooth cycle] coordinates {(-3.5,2.6) (-2,3.1)
  (-0.1,2.4) (2.0,3.2)  (3.9,3.2) (6.8,2.3)(3.5,3.6) (-1.5,3.6)};
  
  \draw[darkred] plot [smooth cycle] coordinates {(-1.8, 1.6)  (-0.1, 0.8) (1.6, 1.6) (-0.1, 1.1)  };
  
  \draw[darkgreen] plot [smooth cycle] coordinates { (0.4,2.1) (0.3, 1.5) (1.0,1.8) (1.2,2.4)  };
\end{scope}
\begin{scope}[line width= 0.3mm, dashed]
  \draw[darkred] plot [smooth cycle] coordinates {(4.8,1.7) (4.7,2.3) (5.4,2.2) (5.3,1.8) };
  
  \draw[darkred] plot [smooth cycle] coordinates {  (-5.7, 1.4) (-0.4,0.4) (2.7,1.7) (-0.4, 0.1) };
  
  \draw[blue] plot [smooth cycle] coordinates { (3.0,1.7) (4.4,1.3) (5.7,1.7) (4.4,1.1)  };
\end{scope}
\end{tikzpicture}
\caption{\label{figsvensson} 
An illustration of step~3 in the first iteration of Svensson's algorithm.
The edge set $F$ is shown in red. 
First the component $Z$ with vertex set $\tilde W_7\cup \tilde W_8$ is considered, with $\mathrm{ind}(Z)=7$.
We may find the blue cycle $C$ with $c(E(C))\le \frac{1}{2\alpha} \ell(\tilde W_7)$.
After adding $E(C)$ to $X$, the component $Z$ with vertex set $\tilde W_3\cup \tilde W_5$ is considered next, with $\mathrm{ind}(Z)=3$.
Then we may find the green cycle $C'$ with $c(E(C'))\le \frac{1}{2\alpha} \ell(\tilde W_3)$.
Then $E(C')$ is added to $X$, and now $(V,H\cupp F\cupp X)$ has three connected components.
The component $Z$ with vertex set $\tilde W_2\cup \tilde W_3\cup \tilde W_4\cup \tilde W_5\cup \tilde W_9$ is considered next.
Suppose there is no cycle $C''$ connecting it to the rest and with $c(E(C''))\le \frac{1}{2\alpha} \ell(\tilde W_2)$.
Then the edges drawn as solid curves are added to $H$, concluding the first iteration.
}
\end{center}
\end{figure}
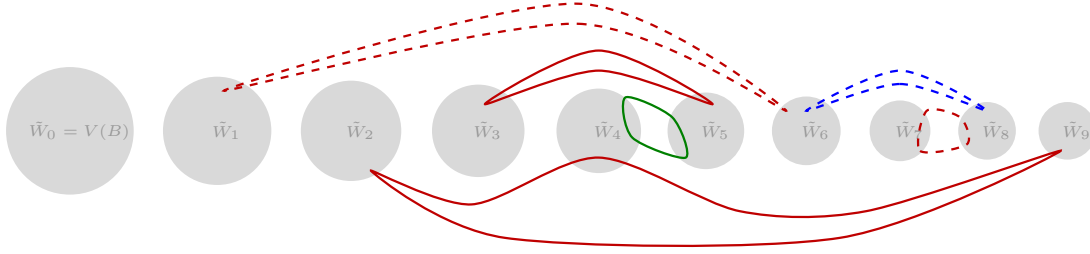

The following observation implies that $(\Iscr, B, H)$ is indeed an instance of Subtour Cover in step~(1a) of Algorithm~\ref{algo:svensson_algo}.

\begin{lemma}
As long as $(V,E(B) \cup H)$ is not connected in Algorithm~\ref{algo:svensson_algo}, $H$ is an Eulerian multi-subset of $E[V\setminus V(B)]$ with $H\cap \delta(L)=\emptyset$ for all $L\in \Lscr_{\ge 2}$.
\end{lemma}
\begin{proof}
At the beginning of the algorithm we set $H := \tilde H$ and thus $H$ is a an Eulerian multi-subset of $E[V\setminus V(B)]$ with $H\cap \delta(L)=\emptyset$ for all $L\in \Lscr_{\ge 2}$.
The cycles that we find in step~(3c) neither contain a vertex from the backbone nor an edge in $\delta(L)$ for any $L\in \Lscr_{\ge 2}$ by construction.
Moreover, by the definition of the Subtour Cover problem (Definition~\ref{def:subtour_cover}), we have $F_i \cap \delta(L)=\emptyset$ for every $i\in \{1,\dots, k\}$ and $L\in \Lscr_{\ge 2}$. By the choice of $Z$ in step~3 of Algorithm~\ref{algo:svensson_algo}, the component $Z$ contains edges from $F_0$ only if $(V, E(B) \cupp H\cupp F \cupp X)$ is connected, and in this case $(V, E(B) \cupp H)$ becomes connected when the edges in $F\cupp X$ are added to $H$.
\end{proof}

Also notice that step~(1b) maintains all properties required for the output of an $(\alpha,\kappa,\beta)$-algorithm for Subtour Cover.
Hence, the computation of $F$ in step~1 (including both step~(1a) and step~(1b)) is an  $(\alpha,\kappa,\beta)$-algorithm for Subtour Cover.
Therefore, we can apply Lemma~\ref{lemma:better_initialization_1} for step~(2a) and Lemma~\ref{lemma:better_initialization_2} for step~(2b)
to show that the application of Lemma~\ref{lemma:improved_initialization} is indeed possible.

We conclude that if Algorithm~\ref{algo:svensson_algo} returns a (multi-)set $\tilde H'$ in step~2, 
then $\tilde H'$ is a multi-set as in Lemma~\ref{lemma:result_svensson}~\textit{(b)}.

Now suppose the algorithm does not terminate in step~2.
Since $H$ remains Eulerian throughout the algorithm and  $(V, E(B)\cup H)$ is connected at the end of 
Algorithm~\ref{algo:svensson_algo}, the returned edge set $H$ is a solution for the vertebrate pair $(\Iscr, B)$.
It remains to show the upper bound \eqref{eq:cost_bound_H_svensson} on the cost of $H$.
Initially we have $c(H) = c(\tilde H) \le \ell(V\setminus V(B))$.
We bound the cost of the $X$-edges and the cost of the $F$-edges added to $H$ separately.

\begin{lemma}\label{lemma:atsp_x_edges}
The total cost of all $X$-edges that are added to $H$ is at most $\frac{1}{2\alpha} \cdot \ell(V\setminus V(B))$.
\end{lemma}

\begin{proof}
A cycle $C$ that is selected in step~(3c) and will later be added to $H$ connects $Z$ with another connected component $Y$
with $\mathrm{ind}(Y) < \mathrm{ind}(Z)$. We say that it marks $\mathrm{ind}(Z)$. 
It has cost at most $\frac{1}{2\alpha}\cdot \ell(\tilde W_{\mathrm{ind}(Z)})$.
No cycle added later can mark $\mathrm{ind}(Z)$ because the new connected component of $(V,H\cupp F\cupp X)$ containing $Y\cup Z$
will have smaller index by the choice of $Z$.
Hence the total cost of the added cycles is at most 
$\frac{1}{2\alpha} \cdot \sum_{i=1}^k \ell(\tilde W_i) = \frac{1}{2\alpha} \cdot \ell(V\setminus V(B))$.
\end{proof}

\begin{lemma}\label{lemma:atsp_f_edges}
The total cost of all $F$-edges that are added to $H$ is at most $\ell(V)$.
\end{lemma}

\begin{proof}
Let $Z^t$ denote $Z$ at the end of iteration $t$ of the while-loop.
Let $F^t_i$ be the graph $F_i$ in iteration $t$ if the set of edges of $F_i$ is nonempty and is added to $H$ at the end of this iteration, 
and let $F^t_i=\emptyset$ otherwise.

For $i=0,\ldots,t$ the total cost of $F^t_i$ is $c(E(F^t_i)) \le \ell(\tilde W_i)$ by step~(2a).
We claim that for any $i$, at most one of the $F^t_i$ is nonempty.
Then summing over all $i$ and $t$ concludes the proof.

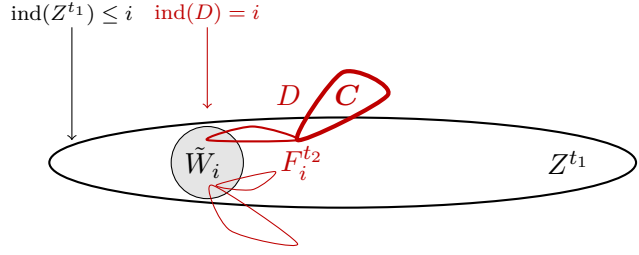
\begin{figure}
\begin{center}
 \begin{tikzpicture}[scale=0.6]
  \tikzset{BackwardEdge/.style={->, >=latex, line width=1pt,darkgreen}}
  \tikzset{ForwardEdge/.style={->, >=latex, line width=1pt,darkred}}
  \tikzset{NeutralEdge/.style={->, >=latex, line width=1pt,gray}}

\draw[thick] (6,1) ellipse (6.5 and 1) {};
\node at (11,1) {$Z^{t_1}$};
\node at (0,4.3) {\scriptsize $\mathrm{ind}(Z^{t_1})\le i$};
\draw[->] (0,4) to (0,1.5);
\draw[fill=gray, fill opacity =0.2] (3,1) ellipse (0.8 and 0.8) {};
\node at (2.9,1) {$\tilde W_i$};

\draw[darkred] plot [smooth cycle] coordinates {(3.1,0.5) (3.5,-0.5) (5,-0.8) };
\draw[darkred] plot [smooth cycle] coordinates {(3.2,0.5) (4,0.4) (4.5,0.8) };
\draw[darkred, ultra thick] plot [smooth cycle] coordinates {(5,1.5) (7,2.5) (6,3) };
\draw[darkred, thick] plot [smooth cycle] coordinates {(3,1.5) (5,1.5) (4,1.8) };
\node[darkred] at (6.1,2.5) {{\boldmath $C$}};
\node[darkred] at (4.8,2.5) {$D$};
\node[darkred] at (5.1,1) {$F^{t_2}_i$};
\node[darkred] at (3,4.3) {\scriptsize $\mathrm{ind}(D)=i$};
\draw[darkred,->] (3,4) to (3,2.2);

 \end{tikzpicture}
\end{center}
\caption{Proof of Lemma~\ref{lemma:atsp_f_edges}: an example of the graph $F^{t_2}_i$ is shown in red.
The circuit $C$ contains an edge of $\delta(V(Z^{t_1}))$, and $D$ is the connected component containing $C$.
\label{fig:astp_f_edges}}
\end{figure}

Suppose there are $t_1<t_2$ such that $F^{t_1}_i\not=\emptyset$ and $F^{t_2}_i\not=\emptyset$.
We have $i>0$ because otherwise the algorithm would terminate after iteration $t_1$ by the choice of $Z^{t_1}$.
Then $V(F^{t_1}_i)\subseteq V(Z^{t_1})$ and thus $\tilde W_i\subseteq V(Z^{t_1})$. See Figure~\ref{fig:astp_f_edges}.
Moreover, $F^{t_2}_i$ contains a vertex of $\tilde W_i$ and is not completely contained in $Z^{t_1}$ by step~(1b) of the algorithm.
Thus, $F^{t_2}_i$ contains a cycle $C$ with $E(C)\cap\delta(V(Z^{t_1})) \not=\emptyset$.
We have $V(F^{t_2}_i) \cap V(B) =\emptyset$ (since $i>0$).
This implies that $C$ is a cycle in $G[V\setminus V(B)]$ and
$E(C) \cap \delta(L) \subseteq E(F^{t_2}_i) \cap \delta(L) =\emptyset$ for all $L\in \Lscr_{\ge 2}$ because 
 $F$ is a solution to Subtour Cover.

If $c(E(C)) \le \frac{1}{2\alpha} \cdot \ell(\tilde W_{\mathrm{ind}(Z^{t_1})})$,
due to step~(3c), this is a contradiction to reaching step~(3d) in iteration $t_1$ and adding $Z^{t_1}$ there.
Otherwise, let $D$ be the connected component of $F_i^{t_2}$ containing $C$.
Note that $\mathrm{ind}(D) = i \ge \mathrm{ind}(Z^{t_1})$.

Since $C$ is a cycle with $E(C)\cap \delta(L) = \emptyset$ for all $L\in \Lscr_{\ge 2}$,
we can apply Lemma~\ref{lemma:circuits_are_light} to obtain
\begin{equation*}
 \frac{1}{(1+\epsilon')\cdot 2\alpha}\cdot \ell(V(C))\ \ge\ c(E(C)) \ > \ \frac{1}{2\alpha} \cdot \ell(\tilde W_{\mathrm{ind}(Z^{t_1})}) \ \ge \ 
 \frac{1}{2\alpha} \cdot \ell(\tilde W_{\mathrm{ind}(D)}).
\end{equation*}
This shows
\begin{equation*}
 \ell(V(D))\ \ge\ \ell(V(C)) \ >\ (1+\epsilon')\cdot \ell(\tilde W_{\mathrm{ind}(D)}).
\end{equation*}
Due to step~(2b), this is a contradiction to reaching step~(3d) in iteration $t_2$ and adding $F_i^{t_2}$ there.
\end{proof}
Using $c(\tilde H) \le \ell(V\setminus V(B))$, Lemma~\ref{lemma:atsp_x_edges}, and Lemma~\ref{lemma:atsp_f_edges}, we conclude that 
the cost of the returned edge set $H$ is at most $\ell(V(B)) + \left(2 + \sfrac{1}{2\alpha}\right) \cdot \ell(V\setminus V(B))$.
This concludes the proof of Lemma~\ref{lemma:result_svensson}.

\section{The main result}

We can now combine the results of the previous sections and obtain the following.

\begin{theorem}\label{thm:main_atsp}
 For every $\epsilon >0$ there is a polynomial-time algorithm that computes for every instance $(G,c)$ of ATSP
 a solution of cost at most $22+\epsilon$ times the cost of an optimum solution to \eqref{eq:atsp_lp}.
\end{theorem}
\begin{proof}
 Theorem~\ref{thm:subtourpartitioncover} yields a $(3,2,1)$-algorithm for Subtour Cover and by  
 Theorem~\ref{thm:reductiontosubtourpartitioncover} this implies that there is a polynomial-time
 $(2, 14 +\epsilon)$-algorithm for vertebrate pairs. 
 Using Theorem~\ref{thm:atspwithvertebratepairs} we then obtain a polynomial-time algorithm that finds a solution of cost at most  $\left(22+\epsilon\right)\cdot \lp(\Iscr)$ for every ATSP instance $\Iscr$.
\end{proof}

As a consequence of Theorem~\ref{thm:main_atsp} we obtain the following.

\begin{corollary}\label{cor:main_atsp}
 The integrality ratio of \eqref{eq:atsp_lp} is at most $22$.
\end{corollary}
\begin{proof}
 Suppose there is an instance $\Iscr$ of ATSP where $\sfrac{\opt(\Iscr)}{\lp(\Iscr)} > 22$.
 Then there exists $\epsilon >0 $ such that $\sfrac{\opt(\Iscr)}{\lp(\Iscr)} > 22 + \epsilon$.
 By Theorem~\ref{thm:main_atsp} we can compute an integral solution for $\Iscr$ with cost at most 
 $(22+\epsilon)\cdot \lp(\Iscr) < \opt(\Iscr)$, a contradiction.
\end{proof}

Using the observation from Remark~\ref{remark:improve_approx_ratio}, 
one could slightly improve Theorem~\ref{thm:main_atsp} and Corollary~\ref{cor:main_atsp},
but the improvement would be less than $1$.

Using the black-box reductions of \cite{FeiS07} and \cite{KohTV19}, our results immediately imply:
\begin{corollary}
There is a $(44+\epsilon)$-approximation algorithm for the path version of ATSP.
The integrality ratio of its classic LP relaxation is at most $85$. \qed
\end{corollary}

Using our new algorithm for ATSP not only as a black-box, one can achieve even better bounds~\cite{Tra19}:
there is a $(43+\epsilon)$-approximation algorithm  for the path version of ATSP and 
the integrality ratio of its classic LP relaxation is at most $43$.
Moreover, our improved version of Svensson's algorithm yields a $(13+\epsilon)$-approximation algorithm for the special case of unit weights,
improving on Svensson's~\cite{Sve15} factor $27$. See~\cite{Tra19} for details.

Although we have reduced the upper bounds on the integrality ratios substantially and proved matching approximation ratios,
 the remaining gaps to the known lower bound of 2 on the integrality ratios \cite{ChaGK06,KohTV19}, let alone to the inapproximability lower bound of $\frac{75}{74}$ \cite{KarLS15}, are much larger than
for Symmetric TSP and Symmetric Path TSP, for which approximation ratios of $\frac{3}{2}-10^{-36}$ \cite{KarKO20,TraVZ20} have been obtained, improving on \cite{Chr76,Ser78,Hoo91,AKS15,Seb13,Vyg16,GotV18,SebZ19,TraV19a,Zen19}.
(For the unit weight special cases of Symmetric TSP and Symmetric Path TSP the best known approximation ratios are $\frac{7}{5}$ and $\frac{7}{5}+\epsilon$ \cite{SebV14,TraVZ20}, improving on \cite{AKS15,GhaSS11,MomS16,Muc14}.)
Improving the upper bounds further remains interesting.

\subsubsection*{Acknowledgment}
We thank the three anonymous reviewers for their careful reading and useful remarks.

%

\end{document}